\documentclass[11pt]{article}
\usepackage{url}
\usepackage{amsmath,amssymb,amsthm,fullpage,bm}
\usepackage[colorlinks=true]{hyperref}
\usepackage{color}
\usepackage{amsfonts}
\usepackage{algpseudocode}
\usepackage{cleveref}
\usepackage[linesnumbered,ruled,vlined]{algorithm2e}
\usepackage{booktabs,multirow}
\SetKwProg{Procedure}{Procedure}{}{}
\allowdisplaybreaks

\newtheorem{theorem}{Theorem}[section]
\newtheorem{lemma}[theorem]{Lemma}

\newtheorem{corollary}[theorem]{Corollary}

\theoremstyle{definition}
\newtheorem{definition}[theorem]{Definition}

\newcommand{\ynote}[1]{\textcolor{red}{(Yuichi: #1)}}

\usepackage{comment}

\title{Lipschitz Continuous Algorithms for Graph Problems} 

  \author{Soh Kumabe\footnote{Supported by JST, PRESTO Grant Number JPMJPR192B and JSPS KAKENHI Grant Number JP21J20547.}\\
  The University of Tokyo\\
  \texttt{soh\_kumabe@mist.i.u-tokyo.ac.jp}
  \and
  Yuichi Yoshida\footnote{Supported by JST, PRESTO Grant Number JPMJPR192B and JSPS KAKENHI Grant Number JP20H05965.}\\
  National Institute of Informatics\\
  \texttt{yyoshida@nii.ac.jp}}

\newcommand{\E}{\mathop{\mathbb{E}}}
\newcommand{\EM}{\mathrm{EM}}
\newcommand{\EMU}{\mathrm{EM}_\mathrm{u}}
\newcommand{\EMW}{\mathrm{EM}_\mathrm{w}}
\newcommand{\TV}{\mathrm{TV}}

\newcommand{\OPT}{\mathrm{opt}}

\newcommand{\LP}{\mathrm{LP}}
\newcommand{\LPmod}{\mathrm{LP}_{\mathrm{ent}}}

\newcommand{\floor}[1]{\left \lfloor{#1}\right \rfloor}
\newcommand{\ceil}[1]{\left \lceil{#1}\right \rceil}

\usepackage{fancyhdr} 
\fancypagestyle{pg}
{%
\lhead{}
\rhead{}
\cfoot{--\ \thepage\ --}

}

\begin{document}
\pagestyle{pg}

\maketitle

\begin{abstract}
Graph algorithms are widely used for decision making and knowledge discovery.
To ensure their effectiveness, it is essential that their output remains stable even when subjected to small perturbations to the input because frequent output changes can result in costly decisions, reduced user trust, potential security concerns, and lack of replicability.
In this study, we consider the Lipschitz continuity of algorithms as a stability measure and initiate a systematic study of the Lipschitz continuity of algorithms for (weighted) graph problems.

Depending on how we embed the output solution to a metric space, we can think of several Lipschitzness notions. 
We mainly consider the one that is invariant under scaling of weights, and we provide Lipschitz continuous algorithms and lower bounds for the minimum spanning tree problem, the shortest path problem, and the maximum weight matching problem.
In particular, our shortest path algorithm is obtained by first designing an algorithm for unweighted graphs that are robust against edge contractions and then applying it to the unweighted graph constructed from the original weighted graph.

Then, we consider another Lipschitzness notion induced by a natural mapping from the output solution to its characteristic vector.
It turns out that no Lipschitz continuous algorithm exists for this Lipschitz notion, and we instead design algorithms with bounded pointwise Lipschitz constants for the minimum spanning tree problem and the maximum weight bipartite matching problem.
Our algorithm for the latter problem is based on an LP relaxation with entropy regularization.
\end{abstract}

\thispagestyle{empty}

\newpage

\thispagestyle{empty}
\tableofcontents
\thispagestyle{empty}

\setcounter{page}{0}
\newpage

\section{Introduction}



Graph algorithms are widely used for decision making and knowledge discovery.
To ensure their effectiveness, it is essential that their output remains stable even when subjected to small perturbations to the input.
Failure to do so may result in several problems:
\begin{itemize}
\itemsep=0pt
\item The output is often used for resource allocation, and it is important to avoid frequent changes as it can result in additional costs.
\item Frequent changes to decisions made by the algorithm can erode user trust.
\item Inconsistency of the output in critical applications such as financial or medical can lead to unsafe decisions being made.
\item The lack of consistency in the output during the process of knowledge discovery can significantly contribute to the difficulty of replicating scientific findings.
\end{itemize}


Unfortunately, many fundamental graph algorithms are vulnerable to small perturbations to the input.
For example, consider Dijkstra's algorithm, a classical algorithm for the shortest path problem that can be found in any textbook. 
Let $G$ be an edge-weighted graph consisting of two vertices $s$ and $t$ connected by two disjoint paths $P_1$ and $P_2$ with almost the same but different lengths.
Given $s$ and $t$ as source and target vertices, respectively, Dijkstra's algorithm outputs the shorter path from $s$ to $t$, say, $P_1$.
However, if the length of an edge in $P_1$ is slightly increased such that $P_2$ becomes shorter than $P_1$, then Dijkstra's algorithm will output $P_2$.
This indicates the vulnerability of Dijkstra's algorithm to input perturbations, which is not desirable for some applications because the output path is often used to determine a transportation route and changing it might incur a huge cost.


In this study, we consider mitigating the effects of input perturbations on weighted graph algorithms.
As a stability measure, we consider the \emph{Lipschitz continuity} of graph algorithms, which guarantees that a small change in the input does not significantly change the output.
Then, we design algorithms with small Lipschitz constants for various graph problems. 

Before presenting our algorithmic results, we discuss how the Lipschitz continuity of a graph algorithm must be defined.

\subsection{Lipschitz Continuity}

\paragraph{Metrics.}
Let $\mathcal{A}$ be an algorithm that, given a graph $G=(V,E)$ and a weight vector $w \in \mathbb{R}_{\geq 0}^E$ over edges, outputs an edge set $\mathcal{A}(G,w) \subseteq E$.
Then, how should we define the Lipschitz constant of $\mathcal{A}$, or more specifically, which metric should we impose on the input and output spaces?

In this study, we always adopt the $\ell_1$ metric for the input space, that is, the distance between two weight vectors $w,w'\in \mathbb{R}_{\geq 0}^E$ is defined to be $\|w-w'\|_1$.
We do so because for combinatorial problems, it is natural to assume that the distance between $w$ and $w'$ is calculated as the sum of the distances between $w(e)$ and $w'(e)$ over $e\in E$, and the $\ell_1$ metric satisfies this property.

We also use the $\ell_1$ metric for the output space.
Depending on how we map the output edge set to a vector in the $\ell_1$ space, we can think of the following two variations.
\begin{description}
    \item[(Unweighted mapping)] We map an edge set $F \subseteq E$ to the characteristic vector $\bm{1}_F \in \mathbb{R}^E$ of $F$, where $\bm{1}_F(e)=1$ if $e \in F$ and $0$ otherwise.
    Then for two edge sets $F,F' \subseteq E$, we define
    \[
      d_u(F,F'):= \|\bm{1}_F - \bm{1}_{F'}\|_1 = |F \triangle F'|.
    \]
    \item[(Weighted mapping)] We map an edge set $F \subseteq E$ to a vector $\sum_{e \in F}w(e) \bm{1}_e$ using the weight vector $w \in \mathbb{R}_{\geq 0}^E$, where $\bm{1}_e \in \mathbb{R}^E$ is the characteristic vector of $e \in E$, that is, $\bm{1}_e(f) = 1$ if $f = e$ and $0$ otherwise.
    Then for two edge sets $F,F' \subseteq E$ and weight vectors $w,w'\in \mathbb{R}_{\geq 0}^E$, we define
    \begin{align*}
        & d_w((F,w),(F',w')) := \left\|\sum_{e \in F}w(e) \bm{1}_e - \sum_{e \in F'}w'(e) \bm{1}_e\right\|_1  \\
        & = \sum_{e \in F \cap F'}|w(e) - w'(e)| + \sum_{e \in F \setminus F'}w(e) + \sum_{e \in F' \setminus F}w'(e).          
    \end{align*}
    We note that $d_w((F,w),(F',w)) = \sum_{e \in F \triangle F'}w(e)$ holds.
\end{description}
To understand the difference between the unweighted and weighted mappings, consider the shortest path problem, where the graph and weight vector  model the road network and the time required to pass through roads, respectively, and the output path represents the roads used in the trip.
The distance $d_u$ considers how many roads are changed between two trips, whereas the distance $d_w$ considers the time spent on different roads between the two trips.

The weighted mapping is more natural than the unweighted one for Lipschitzness.
To see this, let us consider a shortest path algorithm $\mathcal{A}$.
It is natural to ask $\mathcal{A}$ to output the same path regardless of whether the distance is measured in kilometers or miles.
This implies that $\mathcal{A}$ is \emph{scale invariant}, that is, it outputs the same path when edge weights are multiplied by a constant.
Let $G=(V,E)$ be a graph and $w\in \mathbb{R}_{\geq 0}^E$ be a weight vector measured in kilometers, and consider another weight vector $w'\in \mathbb{R}_{\geq 0}^E$ obtained from $w$ by setting $w'(e) = w(e) + \delta$, where $e \in E$ and $\delta > 0$ is measured in kilometers.
Then, if we measure the distance between outputs using the weighted mapping, the relative change of the output is $d_w((\mathcal{A}(G,w),w),(\mathcal{A}(G,w'),w'))/\delta$.
Let $c \approx 1.609$ be the ratio of a mile to a kilometer.
Then, if we calculate edge weights in miles, the relative change of the output is 
\begin{align*}
    \frac{d_w((\mathcal{A}(G,w/c),w/c),(\mathcal{A}(G,w'/c),w'/c))}{\delta/c} 
    & = 
    \frac{d_w((\mathcal{A}(G,w),w),(\mathcal{A}(G,w'),w'))/c}{\delta/c} \\
    & = \frac{d_w((\mathcal{A}(G,w),w),(\mathcal{A}(G,w'),w'))}{\delta},   
\end{align*}
and hence the two relative changes coincide.
We do not have this property if we use the unweighted mapping.
Hence, we first focus on the weighted mapping and then discuss the unweighted one later.

\paragraph{First Attempt: Lipschitz continuity of deterministic algorithms.}
Using the metric based on the weighted mapping imposed on the input and output spaces as mentioned previously, we can define the Lipschitz constant of a deterministic algorithm as follows:
\begin{definition}[Lipschitz constant of a deterministic algorithm]\label{def:deterministic}
    Let $\mathcal{A}$ be a deterministic algorithm that, given a graph $G=(V,E)$ and a weight vector $w \in \mathbb{R}_{\geq 0}^E$, outputs an edge set $\mathcal{A}(G,w) \subseteq E$.
    Then, the \emph{Lipschitz constant} of the algorithm $\mathcal{A}$ on a graph $G=(V,E)$ is
    \begin{align*}
        \sup_{\substack{w,w' \in \mathbb{R}_{\geq 0}^E,\\w\neq w'}}\frac{d_w((\mathcal{A}(G,w),w),(\mathcal{A}(G,w'),w'))}{\|w-w'\|_1}.
    \end{align*}

\end{definition}


Note that we only take the supremum over weight vectors and not over underlying graphs.
To explain why we adopt this definition, let us consider the shortest path problem again.
The weight vector can frequently change owing to traffic jams or inclement weather,  whereas the underlying graph may change because of construction or disasters, which occur less frequently.
Hence, it would be more useful to consider the former type of changes than the latter.

Another reason for not taking the supremum over pairs of graphs is that the change in the underlying graph often forces any (reasonable) algorithm to change its output drastically, and hence it is impossible to bound the Lipschitz constant if we allow changes in the underlying graph.
For example, consider an instance of the shortest path problem such that there are two disjoint paths---one short and the other long---between source and target vertices.
Any algorithm with a reasonable approximation guarantee must output the shorter path.
However, if an edge in the shorter path is removed, the algorithm must change its output to the longer path.

Unfortunately, even though we do not take the supremum over underlying graphs in Definition~\ref{def:deterministic}, any (reasonable) deterministic algorithm for the shortest path problem is not Lipschitz continuous:
\begin{theorem}\label{thm:sp-deterministic-lower-bound}
    Any deterministic algorithm for the shortest path problem with a finite approximation ratio is not Lipschitz continuous, that is, its Lipschitz constant is unbounded.
\end{theorem}
To see the reason, consider a graph having two disjoint paths between the source and target vertices and the transition from a weight vector for which the first path is shorter to one for which the second path is shorter.
Because the algorithm is deterministic, there is some point in the transition where the output path discontinuously changes from the first path to the second one, which implies that the algorithm is not Lipschitz.



\paragraph{Second Attempt: Lipschitz continuity of randomized algorithms.}
To remedy the aforementioned issue, we consider the Lipschitz continuity of randomized algorithms.
First, we extend $d_w$, which is a metric over outputs, to a metric over output distributions.
For two probability distributions $\mathcal{F},\mathcal{F}'$ over subsets of $E$, the \emph{earth mover's distance} between $\mathcal{F}$ and $\mathcal{F}'$ is defined as
\[
    \EMW\left((\mathcal{F},w),(\mathcal{F}',w')\right):=\min_{\mathcal{D}}\E_{(F,F')\sim \mathcal{D}} d_w\left((F,w),(F',w')\right),
\]
where the minimum is taken over \emph{couplings} of $\mathcal{F}$ and $\mathcal{F'}$, that is, distributions over pairs of sets such that its marginal distributions on the first and second coordinates are equal to $\mathcal{F}$ and $\mathcal{F}'$, respectively.
We note that $\EMW$ coincides with $d_w$ if the distributions $\mathcal{F}$ and $\mathcal{F}'$ are supported by single edge sets.

For a randomized algorithm $\mathcal{A}$, a graph $G=(V,E)$, and a weight vector $w \in \mathbb{R}_{\geq 0}^E$, let $\mathcal{A}(G,w)$ denote the (random) output of $\mathcal{A}$ on $G$ and $w$.
Abusing the notation, we often identify it with its distribution.
Then, we define the Lipschitz constant of a randomized algorithm as follows:
\begin{definition}[Lipschitz constant of a randomized algorithm]\label{def:randomized}
    Let $\mathcal{A}$ be a randomized algorithm that, given a graph $G=(V,E)$ and a weight vector $w \in \mathbb{R}_{\geq 0}^E$, outputs a (random) edge set $\mathcal{A}(G,w) \subseteq E$.
    Then, the \emph{Lipschitz constant} of the algorithm $\mathcal{A}$ on a graph $G=(V,E)$ is
    \begin{align*}
        \sup_{\substack{w,w' \in \mathbb{R}_{\geq 0}^E,\\w\neq w'}}\frac{\EMW\left((\mathcal{A}(G,w),w),(\mathcal{A}(G,w'),w')\right)}{\|w-w'\|_1}.
    \end{align*}
    We say that $\mathcal{A}$ is \emph{Lipschitz continuous} if its Lipschitz constant is bounded for any graph $G=(V,E)$ and is \emph{$L$-Lipschitz} if its Lipschitz constant is at most $L$.
\end{definition}
If the algorithm is deterministic, then this definition coincides with Definition~\ref{def:deterministic}.

Consider again the graph having two disjoint paths between the source and target vertices and the transition from a weight vector for which the first path is shorter to one for which the second path is shorter.
Then, the output of a randomized algorithm can also make a transition from a distribution with most of its mass on the first path to one with most of its mass on the second path, and hence we can alleviate the issue of deterministic algorithms.
However, designing Lipschitz continuous algorithms is a nontrivial task because we need to bound the ratio in Definition~\ref{def:randomized} for any pair of weight vectors, which can be very close.

As we will discuss in Section~\ref{subsec:intro-graph-problems}, we can design randomized Lipschitz continuous algorithms for various graph problems.
However, it is not still apparent whether bounding Lipschitz constant of a randomized algorithm is practically relevant because, even if the output distributions for two weighted graphs $(G,w)$ and $(G,w')$ are close, the paths we obtain by running the algorithm on $(G,w)$ and $(G,w')$ independently may be completely different.

\paragraph{Third Attempt: Lipschitz constant of randomized algorithms with shared randomness.}
To overcome the issue mentioned above, we consider randomized algorithms such that the outputs for two similar weighted graphs are close in expectation over the internal randomness.
More specifically, for a randomized algorithm $\mathcal{A}$, let $A_\pi$ denote the deterministic algorithm obtained from $\mathcal{A}$ by fixing its internal randomness to $\pi$.
Then, we define the Lipschitz constant with shared randomness of $\mathcal{A}$ as follows:
\begin{definition}[Lipschitz constant of a randomized algorithm with shared randomness]\label{def:randomized-shared-randomness}
    Let $\mathcal{A}$ be a randomized algorithm that, given a graph $G=(V,E)$ and a weight vector $w \in \mathbb{R}_{\geq 0}^E$, outputs a (random) edge set $\mathcal{A}(G,w) \subseteq E$.
    Then, the \emph{Lipschitz constant under shared randomness} of the algorithm $\mathcal{A}$ on a graph $G=(V,E)$ is
    \begin{align*}
        \sup_{\substack{w,w' \in \mathbb{R}_{\geq 0}^E,\\w\neq w'}}\E_\pi\left[\frac{d_w\left((\mathcal{A}_\pi(G,w),w),(\mathcal{A}_\pi(G,w'),w')\right)}{\|w-w'\|_1}\right],
    \end{align*}
    where $\pi$ represents the internal randomness of $\mathcal{A}$.
    We say that $\mathcal{A}$ is \emph{Lipschitz continuous under shared randomness} if its Lipschitz constant under shared randomness is bounded for any graph $G=(V,E)$ and is \emph{$L$-Lipschitz under shared randomness} if its Lipschitz constant with shared randomness is at most $L$.
\end{definition}
Note that if a randomized algorithm is $L$-Lipschitz under shared randomness, then it is $L$-Lipschitz, and hence the former is a stronger property.

Suppose that we have a randomized algorithm $\mathcal{A}$ for the shortest path problem with a small Lipschitz constant under shared randomness and we want to compute short paths on a graph with varying weights in a stable fashion.
By fixing the randomness $\pi$ first and then using $\mathcal{A}_\pi$ throughout, we can guarantee that the output paths for two similar weights are similar in expectation over $\pi$.

\subsection{Lipschitz Continuous Algorithms for Graph Problems}\label{subsec:intro-graph-problems}

In this section, we discuss the Lipschitz continuity of randomized algorithms with shared randomness for several graph problems.
Our results are summarized in Table~\ref{tab:Lipschitz}.

\begin{table}[t!]
    \centering
    \caption{Results for Lipschitz continuity under shared randomness. $n$ represents the number of vertices in the input graph, and $\epsilon,\alpha \in (0,1)$ are arbitrary constants.}\label{tab:Lipschitz}
    \begin{tabular}{llll}
        \toprule
        \multirow{2}{*}{Problem} & Approximation & Lipschitz & \multirow{2}{*}{Reference}\\
        & Ratio & Constant & \\
        \midrule    
        \multirow{2}{*}{Minimum Spanning Tree} & $1+\epsilon$ & $O(\epsilon^{-1})$ & Sec.~\ref{sec:mst-weighted} \\
        & $1+\epsilon$ & $\Omega(\epsilon^{-1})$ & Sec.~\ref{sec:lower-bounds} \\
        \midrule
        \multirow{2}{*}{Shortest Path} & $1+\epsilon$ & $O(\epsilon^{-1}\log^3 n)$ & Sec.~\ref{sec:sp-Lipschitz} \\
        & $1+\epsilon$ & $\Omega(\epsilon^{-1})$ & Sec.~\ref{sec:lower-bounds} \\
        \midrule
        \multirow{2}{*}{Maximum Weight Matching} & $1/8-\epsilon$ & $O(\epsilon^{-1})$ & Sec.~\ref{sec:maximum-weight-matching} \\
        & $\alpha$ & $\Omega(\alpha)$ & Sec.~\ref{sec:lower-bounds} \\
        \bottomrule
    \end{tabular}
\end{table}        

\paragraph{Minimum spanning tree.}
In the \emph{(weighted) minimum spanning tree problem}, we are given a (connected) undirected graph $G=(V,E)$, and a weight vector $w \in \mathbb{R}_{\geq 0}^E$, and the goal is to output a spanning tree $T \subseteq E$ that minimizes the total weight $\sum_{e \in T}w(e)$.
We show that for any $\epsilon>0$, there exists a polynomial-time $(1+\epsilon)$-approximation algorithm for the minimum spanning tree problem with Lipschitz constant $O(\epsilon^{-1})$ under shared randomness.
To understand this upper bound, suppose that $w$ is $\{0,1\}$-valued and that we change the value of $w(f)$ from $1$ to $0$ for some edge $f \in E$.
This is essentially equivalent to contracting the edge $f$, and the upper bound indicates that we only need to change $O(\epsilon^{-1})$ edges in the spanning tree, which is far smaller than the spanning tree size, i.e., $\Theta(n)$.
We complement the upper bound by showing that any (randomized) $(1+\epsilon)$-approximation algorithm for the minimum spanning tree problem must have Lipschitz constant $\Omega(\epsilon^{-1})$ (even without shared randomness).

\paragraph{Shortest path.}
In the \emph{(weighted) shortest path problem}, we are given an undirected graph $G=(V,E)$, two vertices $s,t \in V$, and a weight vector $w \in \mathbb{R}_{\geq 0}^E$, and the goal is to output the shortest path between $s$ and $t$, where the length of a path $P \subseteq E$ is
$\sum_{e \in P}w(e)$.
We show that for any $\epsilon \in (0,1)$, there exists a polynomial-time $(1+\epsilon)$-approximation algorithm for the shortest path problem with Lipschitz constant $O\left(\epsilon^{-1}\log^3 n\right)$ under shared randomness, where $n$ represents the number of vertices in the input graph. 
Our algorithm may output a walk, i.e., the same edge may be used a multiple times in the output.
We regard a walk $P$ as a multiset of edges, and we map it to a vector $\sum_{e \in P}w(e)\bm{1}_e$, where an edge $e \in E$ appears in the sum the same number of times that it appears in the walk $P$.
Then, the distance $d_w(\cdot,\cdot)$ for walks and the Lipschitz constant of an algorithm that outputs a walk can be naturally defined. 
To understand the upper bound, suppose that $w$ is $\{0,1\}$-valued and that we change the value of $w(f)$ from $1$ to $0$ for some edge $f \in E$.
This is essentially equivalent to contracting the edge $f$, and the upper bound indicates that we only need to change $O(\epsilon^{-1} \log^3 n)$ edges in the output path, which is nontrivially small when the shortest path from $s$ to $t$ is $\omega(\log^3 n)$.
We also show that any (randomized) $(1+\epsilon)$-approximation algorithm for the shortest path problem must have Lipschitz constant $\Omega(\epsilon^{-1})$ (even without shared randomness), which implies that our upper bound is tight up to a polylogarithmic factor in $n$.

\paragraph{Maximum weight matching.}


In the \emph{maximum weight matching problem}, given a graph $G=(V,E)$ and a weight vector $w \in \mathbb{R}_{\geq 0}^E$, we want to find a matching $M \subseteq E$ with the maximum weight, i.e., $\sum_{e \in M}w(e)$.
We show that for any $\epsilon > 0$, there exists a polynomial-time $(1/8-\epsilon)$-approximation algorithm with Lipschitz constant $O(\epsilon^{-1})$ with shared randomness.
To understand this upper bound, suppose again that $w$ is $\{0,1\}$-valued and that we change the value of $w(f)$ from $1$ to $0$ for some $f \in E$.
This is essentially equivalent to deleting the edge $f$, and the upper bound indicates that we only need to change $O(\epsilon^{-1} )$ edges in the matching, which is nontrivially small when the matching size is $\omega(1)$.
We also show that any (randomized) $\alpha$-approximation algorithm for the maximum weight matching problem must have Lipschitz constant $\Omega(\alpha)$ (even without shared randomness).

~\\
We note that the proof of Theorem~\ref{thm:sp-deterministic-lower-bound} can be easily extended to the minimum spanning tree problem and the maximum weight matching problem, and hence randomness is necessary to obtain Lipschitz continuous algorithms for them.

\subsection{Pointwise Lipschitz Continuity for Unweighted Mapping}

In this section, we discuss Lipschitz continuity in the case where the distances between outputs are measured using the unweighted mapping.
First, we define the earth mover's distance between output distributions $\mathcal{F}$ and $\mathcal{F}'$ with respect to the unweighted mapping as follows:
\[
    \EMU\left(\mathcal{F},\mathcal{F}'\right):=\min_{\mathcal{D}}\E_{(F,F')\sim \mathcal{D}} d_u(F,F'),
\]
where the minimum is taken over couplings of $\mathcal{F}$ and $\mathcal{F}'$.

We cannot hope that a scale-invariant algorithm has a bounded Lipschitz constant with respect to the unweighted mapping.
To see this, let $w,w'\in \mathbb{R}_{\geq 0}^E$ be arbitrary weight vectors.
Then for any constant $c>0$, we have
\[
    \frac{d_u(\mathcal{A}(G,w/c),\mathcal{A}(G,w'/c))}{
    \|w/c-w'/c\|_1} = \frac{c \cdot  d_u(\mathcal{A}(G,w),\mathcal{A}(G,w))}{
    \|w-w'\|_1},
\]
which implies that the Lipschitz constant is unbounded.
Hence, we consider the following variant that look at the relative change in a local neighborhood:
\begin{definition}[Pointwise Lipschitz constant of a randomized algorithm with shared randomness with respect to the unweighted mapping]\label{def:pointwise-randomized}
    Let $\mathcal{A}$ be a randomized algorithm that, given a graph $G=(V,E)$ and a weight vector $w \in \mathbb{R}_{\geq 0}^E$, outputs a (random) edge set $\mathcal{A}(G,w) \subseteq E$.
    Then, the \emph{pointwise Lipschitz constant under shared randomness} of the algorithm $\mathcal{A}$ on a graph $G=(V,E)$ at a weight vector $w \in \mathbb{R}_{\geq 0}^E$ with respect to the unweighted mapping is
    \begin{align*}
        & \mathop{\lim\sup}_{\substack{w' \in \mathbb{R}_{\geq 0}^E, w' \to w}} \E_\pi\left[\frac{d_u\left(\mathcal{A}_\pi(G,w),\mathcal{A}_\pi(G,w')\right)}{\|w-w'\|_1}\right],
    \end{align*}
    where $\pi$ represents the internal randomness of $\mathcal{A}$.
\end{definition}
In contrast to Lipschitz constant, the pointwise Lipschitz constant can depend on the weight vector $w$, and hence a non-scale-invariant algorithm can have a bounded pointwise Lipschitz constant.

\begin{table}[t!]
    \centering
    \caption{Results for pointwise Lipschitz continuity with shared randomness with respect to the unweighted mapping. For the minimum spanning tree problem, $n$ represents the number of vertices in the input graph, and for the maximum weight bipartite matching problem, $n$ and $m$ represent the number of vertices in the left and right parts of the input bipartite graph, respectively. $\epsilon \in (0,1)$ is an arbitrary constant, and $\OPT$ represents the optimal value.}\label{tab:pointwise-Lipschitz}
    \begin{tabular}{llll}
        \toprule
        \multirow{2}{*}{Problem} & Approximation & Lipschitz & \multirow{2}{*}{Reference}\\
        & Ratio & Constant & \\
        \midrule    
        Minimum Spanning Tree & $1+\epsilon$ & $O(\epsilon^{-1}n/\OPT)$ & Sec.~\ref{subsec:unweighted-mst} \\
        \midrule
        Maximum Weight Bipartite Matching & $1/2-\epsilon$ & $O(\epsilon^{-1} n^{3/2}\log m /\OPT)$ & Sec.~\ref{subsec:unweighted-bipartite-matching} \\
        \bottomrule
    \end{tabular}
\end{table}     

We consider the pointwise Lipschitz continuity of algorithms with respect to the unweighted mapping for the minimum spanning tree problem and the maximum weight bipartite matching problem. 
Our results are summarized in Table~\ref{tab:pointwise-Lipschitz}.
Below, we discuss them in detail.

For any $\epsilon>0$, we show that there exists a polynomial-time $(1+\epsilon)$-approximation algorithm for the minimum spanning tree problem with pointwise Lipschitz constant $O(\epsilon^{-1} n / \OPT)$ under shared randomness, where $n$ is the number of vertices in the input graph and $\OPT$ is the minimum weight of a spanning tree.
As discussed previously, the dependency on $\mathrm{opt}$ (or some other function depending on edge weights) is unavoidable.
Suppose $w \in \mathbb{R}_{\geq 0}^E$ is $\{0,1\}$-valued.
Then, the bound shows that the change in the output tree is smaller than the tree size, $n-1$, when $\OPT=\omega_n(\epsilon^{-1})$.


In the \emph{maximum weight bipartite matching problem}, given a complete bipartite graph $G=(U \cup V, E = U \times V)$ and a weight vector $w \in \mathbb{R}_{\geq 0}^E$, the goal is to output a matching $M \subseteq E$ that maximizes its weight, i.e., $\sum_{e \in M}w(e)$.
For this problem, we show that there exists a polynomial-time $(1/2-\epsilon)$-approximation algorithm with pointwise Lipschitz constant $O(\epsilon^{-1} n^{3/2}\log m /\OPT)$ under shared randomness, where $n$ and $m$ are the numbers of vertices in the left and right parts of the input bipartite graph, respectively, and $\OPT$ is the maximum weight of a matching.
Suppose $n < m$ and the weight vector $w \in \mathbb{R}_{\geq 0}^E$ is $\{0,1\}$-valued.
Then, the bound shows that the change in the output matching is smaller than the maximum matching size, $n$, when $\OPT=\omega(\epsilon^{-1}\sqrt{n}\log m)$.

We note that, in general, a Lipschitz continuous algorithm does not imply an algorithm with a bounded pointwise Lipschitz constant, and vice versa.

\subsection{Application: Online Algorithm with Small Recourse}
Let $\mathcal{P}$ be a weighted graph problem for which, given a weighted graph $(G=(V,E),w)$, the goal is to output a (feasible) edge set $F \subseteq E$ that minimizes or maximizes the objective function $f_{\mathcal{P}}:2^E \to \mathbb{R}_{\geq 0}$.
Let us consider the online version of $\mathcal{P}$ in which a graph $G=(V,E)$ is given at the beginning, and then weight vectors $w_1,\ldots,w_T \in \mathbb{R}_{\geq 0}^E$ are given online.
Then, our goal is to design a (possibly randomized) online algorithm that outputs solutions $F_1,\ldots,F_T$, where $F_t$ solely depends on $G$ and $w_1,\ldots,w_t$, such that $F_t$ is a good solution to the weighted graph $(G,w_t)$ and the \emph{recourse} $\sum_{t = 2,\ldots,T} d_w((F_t,w_t),(F_{t-1},w_{t-1}))/\|w_t - w_{t-1}\|_1$ is small. 

Online algorithms with small recourse have been studied for various problems such as the set cover problem~\cite{gupta2017online}, the facility location problem~\cite{bhattacharya2022efficient},
the edge coloring problem~\cite{bhattacharya2021online}, and the $k$-clustering problem~\cite{lattanzi2017consistent}, though they studied the setting that, at each time step, a small part of the input is given instead of a weight vector.

We can use a Lipschitz continuous algorithm $\mathcal{A}$ under shared randomness to design an online algorithm with a small recourse:
Fix the randomness $\pi$ first, and then run the deterministic algorithm $\mathcal{A}_\pi$ on the weighted graph $(G,w_t)$ to obtain a solution $F_t$ for every $t = 1,\ldots,T$.
Then, from the Lipschitzness guarantee of $\mathcal{A}$, we can bound the expected recourse as in the following theorem:
\begin{theorem}\label{thm:online-alg-with-small-recourse}
    Suppose that there exists a randomized $\alpha$-approximation algorithm for a weighted graph problem $\mathcal{P}$ that is $L$-Lipschitz under shared randomness.
    Then, there exists an $\alpha$-approximation algorithm for the online version of $\mathcal{P}$ whose expected recourse is at most $LT$.
\end{theorem}
Combining Theorem~\ref{thm:online-alg-with-small-recourse} with our Lipschitz continuous algorithms, we immediately obtain online algorithms with small recourse for the minimum spanning tree problem, the shortest path problem, and the maximum weight matching problem.

\subsection{Related Work}

\paragraph{Worst-case and average sensitivity.}

Lipschitz continuity is closely related to the \emph{sensitivity} of algorithms introduced in~\cite{Murai:2019hG,Varma2021}.
The \emph{worst-case and average sensitivities} of a randomized algorithm $\mathcal{A}$ on an (unweighted) graph $G=(V,E)$ are defined as
\begin{align}
    \max_{e \in E}\EM_u(\mathcal{A}(G),\mathcal{A}(G\setminus e)) \quad \text{and} \quad
    \frac{1}{|E|}\sum_{e \in E}\EM_u(\mathcal{A}(G),\mathcal{A}(G\setminus e)), \label{eq:sensitivity}
\end{align}
respectively, where $G \setminus e$ is the graph obtained from $G$ by deleting the edge $e \in E$.
Clearly, the average sensitivity is bounded from above by the worst-case sensitivity.
The sensitivity of algorithms has been investigated for various graph problems including the minimum cut problem~\cite{Varma2021}, the maximum matching problem~\cite{Varma2021,Yoshida2021}, and spectral clustering~\cite{Peng2020}.
It is known that there is no algorithm with $o(n)$ worst-case/average sensitivity for the shortest path problem~\cite{Varma2021}.
As the definition of sensitivity~\eqref{eq:sensitivity} can be easily generalized, other non-graph problems such as dynamic programming problems~\cite{Kumabe22,kumabe_et_al:LIPIcs.ESA.2022.75} and Euclidean $k$-means~\cite{YI22} have also been studied from the viewpoint of sensitivity.

To see the connection to Lipschitz continuity, suppose that in the supremum of Definition~\ref{def:randomized}, we fix $w$ to be the all-one vector $\bm{1}_E$ and we restrict the domain of $w'$ to $\{0,1\}$-valued vectors.
Then by the triangle inequality, the (modified) Lipschitz constant on a graph $G=(V,E)$ can be bounded from above as
\begin{align}
    & \max_{w' \in \{0,1\}^E}\frac{\EMW((\mathcal{A}(G,\bm{1}_E),\bm{1}_E),(\mathcal{A}(G,w'),w'))}{\|\bm{1}_E - w'\|} 
    = \max_{F \subseteq E}\frac{\EMW((\mathcal{A}(G,\bm{1}_E),\bm{1}_E),(\mathcal{A}(G,\bm{1}_{E \setminus F}),\bm{1}_{E \setminus F}))}{|F|} \nonumber  \\
    & \leq 
    \max_{e \in E}\EMW((\mathcal{A}(G,\bm{1}_E),\bm{1}_E),(\mathcal{A}(G,\bm{1}_{E\setminus \{e\}}),\bm{1}_{E \setminus \{e\}})). \label{eq:sensitivity-against-zeroing-out}
\end{align}

For the shortest path problem, the weighted graph $(G,w)$ for a $\{0,1\}$-valued weight vector $w$ is equivalent to the graph obtained from $G$ by contracting edges $e \in E$ with $w(e) = 0$.
In particular, the weighted graph $(G,\bm{1}_{E \setminus \{e\}})$ is equivalent to $G/e$, which is the graph obtained from $G$ by contracting the edge $e$.
Hence,~\eqref{eq:sensitivity-against-zeroing-out} can be seen as a variant of the worst-case sensitivity, where the operation applied to the graph is edge contraction instead of edge deletion, that is,
\begin{align}
    \max_{e \in E}\EM_u(\mathcal{A}(G),\mathcal{A}(G / e)). \label{eq:sensitivity-contraction}
\end{align}
Indeed, our Lipschitz continuous algorithm for the shortest path problem is based on an algorithm with a bounded sensitivity with respect to edge contraction.

A notable difference between the sensitivity and the Lipschitz constant is that the former is trivially bounded by the maximum solution size whereas it is not clear a priori whether there is an algorithm for which the latter is bounded.

\paragraph{Replicability}

Note that, when an algorithm has a small Lipschitz constant under shared randomness, it has the ability to replicate the output with respect to weight perturbations. Another form of replicability has been introduced for the stochastic setting by Impagliazzo~et al.~\cite{impagliazzo2022reproducibility}. Consider two sets of $n$ i.i.d.\ elements, $X = (x_1,...,x_n)$ and $Y = (y_1,...,y_n)$. Then, a randomized algorithm $A$ is called \emph{$\rho$-replicable} for $\rho>0$ if $\Pr_{X,Y,\pi}[A_\pi(X) = A_\pi(Y)] \geq 1-\rho$, where $\pi$ denotes the internal randomness of $A$. Impagliazzo et al.~\cite{impagliazzo2022reproducibility} showed a replicable algorithm for learning halfspaces, and since then, replicable algorithms have been developed for bandit problems~\cite{esfandiari2023replicable} and clustering problems~\cite{esfandiari2023replicable-clustering}.


\paragraph{Adversarial attacks and Lipschitz continuity of neural networks.}
In the field of machine learning, it is known that the predictions of the trained model can be significantly changed by slight perturbations of the input, and such perturbations are called \emph{adversarial attacks}~\cite{goodfellow2015explaining,szegedy2013intriguing}.
Because adversarial attacks threaten the security of machine learning-based systems, there is active research on training models that are robust against them~\cite{carlini2017towards,madry2018towards}.
See surveys~\cite{akhtar2018threat,xu2020adversarial,zhang2020adversarial} for more details.

To mitigate the effects of adversarial attacks, neural networks with small Lipschitz constants have been proposed and their properties have been investigated~\cite{cisse2017parseval,fazlyab2019efficient,tsuzuku2018lipschitz,virmaux2018lipschitz,weng2018evaluating}.
It is also reported that bounding Lipschitz constants of neural networks stabilizes the training process and often produces models with better output quality~\cite{gouk2021regularisation,miyato2018spectral,yoshida2017spectral}.

We note that bounding the Lipschitz constant of a neural network is often easy because it is bounded by the product of the Lipschitz constants of the activation functions and linear transformations used in the neural network, which are easy to calculate.
However, to design Lipschitz continuous algorithms for graph problems, we need to bound approximation ratio and Lipschitz constant simultaneously, and we often need nontrivial techniques as we will see in this paper.

\subsection{Technical Overview}

We first discuss designing Lipschitz continuous algorithms (without shared randomness) and then briefly discuss the necessary modifications that ensure Lipschitz continuity under shared randomness.

\paragraph{Minimum spanning tree.}
It is known that Kruskal's algorithm has (worst-case) sensitivity $O(1)$ against edge deletions~\cite{Varma2021}. 
However, it is not Lipschitz continuous because it is deterministic (see  Theorem~\ref{thm:sp-deterministic-lower-bound}), and hence some modification is required.

Our Lipschitz continuous algorithm for the minimum spanning tree problem works as follows:
Given a graph $G=(V,E)$ and a weight vector $w \in \mathbb{R}_{\geq 0}^E$, we sample $\widehat{\bm{w}}(e)$ uniformly from $[\bm{w}(e),(1+\epsilon)\bm{w}(e)]$ for each edge $e \in E$, and then apply Kruskal's algorithm to the new weight vector $\widehat{\bm{w}}$.
This algorithm clearly achieves $(1+\epsilon)$-approximation.

We can show that, to bound the Lipschitz constant, it suffices to consider a pair of weight vectors $(w,w')$ such that $w'$ is obtained from $w$ by setting $w'(e) = w(e) + \delta$ for some $e \in E$ and $\delta > 0$.
Then, the total variation distance between $\widehat{\bm{w}}$ and $\widehat{\bm{w}}'$ is  $O(\epsilon^{-1}\delta/w(e))$, where $\widehat{\bm{w}}'$ is constructed from $w'$ in the same way as $\widehat{\bm{w}}$ is constructed from $w$.
Then, we can define a coupling, i.e., a joint distribution, between $\widehat{\bm{w}}$ and $\widehat{\bm{w}}'$ such that $\widehat{\bm{w}} \neq \widehat{\bm{w}}'$ with probability $O(\epsilon^{-1}\delta/w(e))$. 
Also, we can show that when $\bm{w} \neq \bm{w}'$ occurs in the coupling, the distance between the output spanning trees is $O(w(e))$.
This implies that the earth mover's distance is $O(\epsilon^{-1}\delta)$ and hence the Lipschitz constant is $O(\epsilon^{-1})$.

Our algorithm and analysis for pointwise Lipschitzness with respect to the unweighted mapping is similar though we need some care because we use the optimal value to determine the range from which we sample $\widehat{\bm{w}}(e)$ and it varies depending on the weight vector.

\paragraph{Shortest path.}

To obtain Lipschitz continuous algorithm for the shortest path problem, we first design an algorithm for the (unweighted) shortest path problem with a low sensitivity with respect to edge contraction (see~\eqref{eq:sensitivity-contraction}). 
We will use this algorithm as a subroutine in our Lipschitz continuous algorithm.

The subroutine takes an unweighted graph $\widehat{G}$ and two vertices $s,t\in V(G)$ as the input and returns an approximate $s$-$t$ shortest path in $\widehat{G}$.
This subroutine is recursive: It samples a vertex $\bm{v}$ called a \emph{pivot}, recursively computes $s$-$\bm{v}$ and $\bm{v}$-$t$ walks that are nearly optimal, and then returns the walk obtained by concatenating the two walks.
We choose the pivot $\bm{v}$ so that it is roughly in the middle of a nearly optimal $s$-$t$ path.
By doing so, we can bound the depth of the recursion by $O(\log n)$ and guarantee that the output walk is nearly optimal.

Now, we turn to analyzing the sensitivity of the subroutine with respect to edge contraction.
Let $e\in E(\widehat{G})$ and we want to bound the earth mover's distance between the output distributions for $\widehat{G}$ and $\widehat{G}/e$.
Informally, we say that a recursion call for computing a $u$-$v$ shortest path is \emph{active} if there is a nearly optimal $u$-$v$ path passing through $e$ and is \emph{inactive} otherwise. 
Then, we can show the following:
\begin{itemize}
    \item recursion calls invoked in an inactive recursion call are also inactive, and
    \item with high probability, at most one of the two recursion calls invoked in an active recursion call is active.
\end{itemize}
These properties imply that the expected number of active recursion calls in each recursion depth is $O(1)$.
Also, we can prove that in an inactive call, the pivot is sampled from exactly the same distributions for $\widehat{G}$ and $\widehat{G}/e$, and thus it does not contribute to the sensitivity at all. 
Additionally, we can show that each active call contributes to the sensitivity by $O(\log^2 n)$. 
Regarding that the number of active calls is $O(\log n)$, the sensitivity of the algorithm can be bounded by $O(\log^3 n)$.

To obtain a Lipschitz continuous algorithm for the shortest path problem, we construct an unweighted graph $\widehat{\bm{G}}(w)$ from the input weighted graph $(G,w)$, and apply the subroutine to it to obtain a walk in $\widehat{\bm{G}}(w)$, and then output the corresponding walk in $G$.
Here, the graph $\widehat{\bm{G}}(w)$ is obtained by replacing each edge of the input graph $G$ with a path of suitable length so that each path in $G$ naturally corresponds to a path in $\widehat{\bm{G}}(w)$. 
The length of each path is sampled from a certain distribution so that the total variation distance between $\widehat{\bm{G}}(w)$ and $\widehat{\bm{G}}(w')$ is proportional to $\|w-w'\|_1$. 
Then, we can bound the Lipschitz continuity of the algorithm by using the sensitivity of the subroutine.


\paragraph{Maximum weight matching.}

Our algorithm is based on the algorithm for the maximum weight matching proposed by Yoshida and Zhou~\cite{Yoshida2021}. 
For a parameter $\alpha>2$, their algorithm first classifies edges $e$ according to the value $\floor{\log_{\alpha} w(e)}$. 
Then, it runs a randomized greedy to compute a matching for each edge class, and then returns a matching obtained by combining them.

Although their algorithm has a bounded \emph{weighted sensitivity}, a discrete analogue of Lipschitz constant, its Lipschitz constant is not bounded.
This is because an arbitrarily small change in the weight of an edge may cause the edge to be classified into a different class. 
To resolve this issue, we sample a parameter $\bm{b}\in [1,\alpha]$ and classify edges $e$ according to the value $\floor{\log_{\alpha} \frac{w(e)}{\bm{b}}}$.
Then, the total variation distance between the classifications obtained from weight vectors $w$ and $w'$ is proportional to $\|w-w'\|_1$, and we can bound the Lipschitz constant.

\paragraph{Maximum weight bipartite matching.}

In contrast to the previous Lipschitz continuous algorithm for the maximum weight matching problem, our pointwise Lipschitz continuous algorithm for the maximum weight bipartite matching problem is based on linear programming (LP). 
The standard LP relaxation for the maximum weight bipartite matching problem is not stable against perturbations to the edge weight.
Hence, we consider LP with \emph{entropy regularization}~\cite{cuturi2013sinkhorn}.
Although entropy regularization was originally introduced to speed up the computation of the earth mover's distance, we use it here to stabilize the computation of LP.
For a weight vector $w \in \mathbb{R}_{\geq 0}^E$, let $\LPmod(w)$ denote the LP for the weighted graph $(G,w)$ with entropy regularization.
Then, we can show that (i) the solution to $\LPmod$ is nearly optimal to the original LP, and (ii) for any weight vector $w'$, the $\ell_1$ distance between the solutions of $\LPmod(w)$ and $\LPmod(w')$ is proportional to $\|w-w'\|_1$.
Then, we carefully round the obtained fractional solution to an integral one in such a way that for any two fractional solutions $x$ and $x'$, the earth mover's distance between the (random) integer solutions obtained from $x$ and $x'$ is proportional to $\|x-x'\|_1$ with respect to the unweighted mapping.

\paragraph{Lipschitz continuity under shared randomness}
In our algorithm, all randomness is caused by sampling processes from continuous or discrete distributions. For example, in our Lipschitz continuous algorithm for the minimum spanning tree, randomness is caused by sampling each edge's modified weight $\widehat{\bm{w}}(e)$ from a uniform distribution over $[w(e), (1+\epsilon)w(e)]$. Let us consider fixing this randomness using shared randomness $\pi$.

Let $\mathcal{F}$ be a function that maps weight vectors $w\in \mathbb{R}^E_{\geq 0}$ to probability distributions, and suppose that in our algorithm, we sample values from the distribution $\mathcal{F}(w)$ at some point. We define a sampling process \Call{Sample}{$\cdot,\pi$} to be \emph{stable} for $\mathcal{F}$ if it satisfies the following two conditions:
\begin{itemize}
\item For any $w\in \mathbb{R}^E_{\geq 0}$, when $\pi$ is sampled from a uniform distribution, the output distribution of \Call{Sample}{$\mathcal{F}(w),\pi$} is equal to $\mathcal{F}(w)$.
\item For any $w,w' \in \mathbb{R}^E_{\geq 0}$, when $\pi$ is sampled from a uniform distribution, the probability that \Call{Sample}{$\mathcal{F}(w),\pi$} and \Call{Sample}{$\mathcal{F}(w'),\pi$} differ is proportional to the total variation distance between $\mathcal{F}(w)$ and $\mathcal{F}(w')$.
\end{itemize}
By replacing all sampling processes in the algorithm with stable processes, we can make our algorithm Lipschitz continuous under shared randomness.

There are four types of sampling processes in the algorithm. Three of them sample from a continuous uniform distribution $\mathcal{F}(w) = \mathcal{U}[l(w),r(w)]$, where $\mathcal{U}[a,b]$ denotes the uniform distribution over $[a,b]$, with $l(w)$ and $r(w)$ being constant functions, with $l(w)$ and $r(w)$ having a constant ratio, and with $l(w)$ and $r(w)$ having a constant difference. The last one samples from a discrete distribution with fixed support.

For each of these processes, we construct a stable algorithm. Furthermore, by discussing how to modify each part of the five algorithms mentioned above, we obtain Lipschitz continuous algorithms under shared randomness with the same approximation ratio and Lipschitz constant, as well as polynomial-time complexity.





\subsection{Preliminaries}
For a function $f\colon X \to Y$ and $S \subseteq X$, let $f|_S\colon S \to Y$  denote the function obtained from $f$ by restricting its domain to $S$.

We use bold symbols to denote random variables.
For a random variable $\bm{X}$ and an event $E$, we use $(\bm{X} \mid E)$ to denote $\bm{X}$ conditioned on $E$.
We often identify a random variable with its distribution.

For two random variables $\bm{X}$ and $\bm{X}'$, the \emph{total variation distance} between them is given as follows:
\begin{align*}
    \TV\left(\bm{X},\bm{X}'\right):=\min_{\mathcal{D}}\Pr_{(x,x')\sim \mathcal{D}} \left[x=x'\right],
\end{align*}
where the minimum is taken over couplings between $\bm{X}$ and $\bm{X}'$, that is, distributions over pairs such that its marginal distributions on the first and the second coordinates are equal to $\bm{X}$ and $\bm{X}'$, respectively.

For an element $f \in E$, we use $\bm{1}_f\in \mathbb{R}_{\geq 0}^E$ to denote the \emph{characteristic vector} of $f$, that is, $\bm{1}_f(f) = 1$ and $\bm{1}_f(e) = 0$ for $e \in E \setminus \{f\}$.
The following lemma indicates that, to bound the Lipschitz constant, it suffices to consider pairs of weight vectors that differ by one coordinate.
The proof is deferred to Appendix~\ref{sec:seeoneelement}.
\begin{lemma}\label{lem:seeoneelement}
Suppose that there exist some $c>0$ and $L>0$ such that
\[
    \EMW\left((\mathcal{A}(G,w),w), (\mathcal{A}(G,w+\delta \mathbf{1}_f),w+\delta \mathbf{1}_f)\right)\leq \delta L
\]
holds for all $w\in \mathbb{R}_{\geq 0}^E$, $f\in E$ and $0<\delta\leq c$. Then, $\mathcal{A}$ is $L$-Lipschitz.
\end{lemma}

\section{Minimum Spanning Tree}\label{sec:mst-weighted}
In this section, we consider the Lipschitz continuity of the minimum spanning tree problem and prove the following:
\begin{theorem}\label{thm:mst-weighted}
For any $\epsilon>0$, there exists a polynomial-time $(1+\epsilon)$-approximation algorithm for the minimum spanning tree problem with Lipschitz constant $O(\epsilon^{-1})$.
\end{theorem}

Our algorithm is simple.
Given a graph $G=(V,E)$, a weight vector $w\in \mathbb{R}_{\geq 0}^E$, and $\epsilon > 0$, we sample $\widehat{\bm{w}}(e)$ uniformly from $[w(e),(1+\epsilon)w(e)]$ for each $e \in E$, and then return an arbitrary minimum spanning tree with respect to the weight vector $\widehat{\bm{w}}$ (see \Call{LipMST}{} in Algorithm~\ref{alg:mst_main_w}).

\begin{algorithm}[t!]
\caption{Lipschitz-continuous algorithm for the minimum spanning tree problem}\label{alg:mst_main_w}
\Procedure{\emph{\Call{LipMST}{$G=(V,E), w, \epsilon$}}}{
    \For{$e\in E$}{
        Sample $\widehat{\bm{w}}(e)$ uniformly from $[w(e),(1+\epsilon)w(e)]$\;\label{line:mst_main_w_sample}
    }
    \Return arbitrary minimum spanning tree of $G$ with respect to $\widehat{\bm{w}}$.
}
\end{algorithm}

First, we analyze the approximation ratio.
\begin{lemma}\label{lem:mst-weighted-approximation}
The approximation ratio of \Call{LipMST}{} is (at most) $1+\epsilon$.
\end{lemma}
\begin{proof}
Let $T^*$ be the optimal solution with respect to the weight vector $w$, and let $\bm{T}$ be the output of \Call{LipMST}{}.
Let $\OPT$ (resp., $\widehat{\bm{\OPT}}$) be the minimum weight of a spanning tree with respect to $w$ (resp., $\widehat{\bm{w}}$).
Then, we have
\begin{align*}
    \widehat{\bm{\OPT}} = \sum_{e\in \bm{T}}\widehat{\bm{w}}(e)\leq \sum_{e\in T^*}\widehat{\bm{w}}(e)\leq \sum_{e\in T^*}(1+\epsilon)w(e)=(1+\epsilon)\OPT,
\end{align*}
where the first inequality is from the fact that $\bm{T}$ is a minimum spanning tree with respect to $\widehat{\bm{w}}$, and the claim holds.
\end{proof}

Now we analyze the Lipschitz continuity of \Call{LipMST}{}.
For a weight vector $w\in \mathbb{R}_{\geq 0}^E$ with distinct elements, the minimum spanning tree with respect to $w$ uniquely exists, and we denote it by $T(w)$.
The following lemma is simple but useful.
\begin{lemma}\label{lem:mst-local-update}
Let $f \in E$ be an edge, $\delta > 0$, and $w\in \mathbb{R}_{\geq 0}^E$ be such that both $w$ and $w+\delta \bm{1}_f$ have distinct elements.
Then, $T^{\delta f}:=T(w+\delta \mathbf{1}_f)$ is either identical to $T:= T(w)$ or is obtained from $T$ by deleting $f$ and adding an edge with weight at most $w(f)+\delta$.
\end{lemma}
\begin{proof}
If deleting $f$ disconnects $G$, we have $T^{\delta f}=T$. 
Suppose otherwise.
Let $T^{-f}$ be the minimum spanning tree of $G^{-f}:=(V, E\setminus \{f\})$ with respect to the weight vector $w|_{E\setminus \{f\}}$. Then, $T$ (resp., $T^{\delta f}$) is obtained from $T^{-f}\cup \{f\}$ by deleting the (unique) edge $g$ (resp., $g^{\delta f}$) of the largest weight  with respect to the weight vector $w$ (resp., $w+\delta \mathbf{1}_f$) in the unique cycle $C$ in $T^{-f} \cup \{f\}$.
If $g=f$, then the weight of $f$ is the largest in $C$ with respect to $w$, and hence it is the largest with respect to $w+\delta \mathbf{1}_f$ too.
Therefore, we have $T=T^{\delta f}=T^{-f}$. 
If $g^{\delta f}\neq f$, then the weight of $f$ is not the largest in $C$ with respect to $w+\delta \mathbf{1}_f$, and hence it is not the largest with respect to $w$ either. 
Therefore, the edges with the largest weight in $C$ with respect to $w$ and $w+\delta \mathbf{1}_f$ are the same, and hence we have $T=T^{\delta f}$.
If $g\neq f$ and $g^{\delta f}=f$, then $T$ uses $f$, and $T^{\delta f}$ is obtained from $T$ by deleting $f$ and adding an edge with the largest weight in $C$ with respect to $w$, which is at most the weight of $f$ on $w+\delta \mathbf{1}_f$, that is, $w(f)+\delta$.
\end{proof}


Let $f \in E$ be an edge and $\delta>0$. 
To bound the Lipschitz constant, by Lemma~\ref{lem:seeoneelement}, it suffices to bound
\begin{align}
    \EMW\left((\Call{LipMST}{G,w,\epsilon},w), (\Call{LipMST}{G,w+\delta \mathbf{1}_f,\epsilon},w+\delta \mathbf{1}_f)\right).\label{eq:mstw}
\end{align}
To bound~\eqref{eq:mstw}, we consider a coupling (i.e., a joint distribution) $\mathcal{W}$ between $\widehat{\bm{w}}$ and $\widehat{\bm{w}}^{\delta f}$ defined as follows:
For each weight vector $\widehat{w} \in \mathbb{R}_{\geq 0}^E$, we transport the probability mass for $\widehat{\bm{w}} = \widehat{w}$ to that for $\widehat{\bm{w}}^{\delta f} = \widehat{w}$ as far as possible in such a way that the probability mass for $\widehat{\bm{w}}|_{E \setminus \{f\}} = \widehat{w}|_{E \setminus \{f\}}$ is always transported to that for $\widehat{\bm{w}}^{\delta f}|_{E \setminus \{f\}} = \widehat{w}|_{E \setminus \{f\}}$, which is possible because the distributions of $\widehat{\bm{w}}(e)$ and $\widehat{\bm{w}}^{\delta f}(e)$ are identical for $e \in E \setminus \{f\}$.
The remaining probability mass is transported arbitrarily.
%
\begin{lemma}\label{lem:mst-remaining}
    We have
    \[
        \Pr_{(\widehat{\bm{w}},\widehat{\bm{w}}^{\delta f}) \sim \mathcal{W}}\left[\widehat{\bm{w}} \neq \widehat{\bm{w}}^{\delta f}\right] = \Pr_{(\widehat{\bm{w}},\widehat{\bm{w}}^{\delta f}) \sim \mathcal{W}}\left[\widehat{\bm{w}}(f) \neq \widehat{\bm{w}}^{\delta f}(f) \right]  \leq \frac{(1+\epsilon)\delta}{\epsilon(w(f)+\delta)}.
    \]
\end{lemma}
\begin{proof}
    The equality holds by the design of the coupling $\mathcal{W}$.

    Now, we show the inequality.
    Recall that $\widehat{\bm{w}}(f)$ is sampled uniformly from $\left[w(f),(1+\epsilon)w(f)\right]$ and $\widehat{\bm{w}}^{\delta f}(f)$ is sampled uniformly from $\left[w(f)+\delta,(1+\epsilon)(w(f)+\delta)\right]$. 
    Hence, we have
    \begin{align*}
        & \Pr_{(\widehat{\bm{w}},\widehat{\bm{w}}^{\delta f}) \sim \mathcal{W}}\left[\widehat{\bm{w}}(f) \neq \widehat{\bm{w}}^{\delta f}(f)\right]  
        \leq
        1-\Pr\left[\widehat{\bm{w}}^{\delta f}(f)\in \left[w(f),(1+\epsilon)w(f)\right]\right] \\
        & \leq \frac{(1+\epsilon)(w(f)+\delta)-(1+\epsilon)w(f)}{(1+\epsilon)(w(f)+\delta)-(w(f)+\delta)}
        = \frac{(1+\epsilon)\delta}{\epsilon(w(f)+\delta)},        
    \end{align*}
    as desired.
\end{proof}

\begin{lemma}\label{lem:mst-weighted-Lipschitz}
    \Call{LipMST}{} is $O(\epsilon^{-1})$-Lipschitz.
\end{lemma}
\begin{proof}
    Note that $\widehat{\bm{w}}$ has distinct elements with probability one and hence we only need to consider the case that it happens when calculating the Lipschitz constant.
    By Lemmas~\ref{lem:mst-local-update} and~\ref{lem:mst-remaining}, we have
    \begin{align*}
        & \EMW\left(\left(\Call{LipMST}{G,w,\epsilon},w\right),\left(\Call{LipMST}{G,w+\delta \mathbf{1}_f, \epsilon},w+\delta \mathbf{1}_f\right)\right) \\
        & \leq \E_{(\widehat{\bm{w}},\widehat{\bm{w}}^{\delta f}) \sim \mathcal{W}} d_w\left((T(\widehat{\bm{w}}),w),(T(\widehat{\bm{w}}^{\delta f}),w+\delta \bm{1}_f)\right) \\
        & \leq \Pr_{(\widehat{\bm{w}},\widehat{\bm{w}}^{\delta f}) \sim \mathcal{W}}\left[\widehat{\bm{w}}(f) \neq \widehat{\bm{w}}^{\delta f}(f)\right] \cdot ((1+\epsilon)w(f)+(1+\epsilon)(w(f)+\delta)) \\
        &\leq \frac{(1+\epsilon)\delta}{\epsilon\cdot (w(f)+\delta)}\cdot 2(1+\epsilon)(w(f)+\delta)
        = \frac{2(1+\epsilon)^2 \delta}{\epsilon}
        \leq O\left(\frac{\delta}{\epsilon}\right).
    \end{align*}
    Therefore, the Lipschitz constant is bounded as
    \[
        \sup_{\delta > 0, f\in E}\frac{\EMW\left(\left(\Call{LipMST}{G,w,\epsilon},w\right),\left(\Call{LipMST}{G,w+\delta \mathbf{1}_f, \epsilon},w+\mathbf{1}_f\right)\right)}{\delta}\leq O\left(\frac{1}{\epsilon}\right).
        \qedhere
    \]
\end{proof}
Theorem~\ref{thm:mst-weighted} follows by combining Lemmas~\ref{lem:mst-weighted-approximation} and~\ref{lem:mst-weighted-Lipschitz}.


\section{Contraction Sensitivity of Shortest Path}\label{sec:sp}

In this section, we present an algorithm for the shortest path problem that is stable against edge contractions.
We will use this algorithm as a subroutine in our Lipschitz continuous algorithm for the shortest path problem in Section~\ref{sec:sp-Lipschitz}.
In this section, we only consider unweighted graphs.


To define the sensitivity to edge contractions, we first formally define edge contraction.
Let $G=(V,E)$ be a graph and $e = \{u,v\}\in E$ be an edge.
Then, the \emph{contraction} of an edge $e$ in $G$ is an operation that yields a graph $G / e = (V',E')$ with
\begin{align*}
    V' & = (V \setminus e) \cup \{v_e\}, \\
    E' & = E \setminus (\{\{u,x\} : x \in N(u)\} \cup \{\{v,x\} : x \in N(v)\}) \cup \{\{v_e, x\} : x \in N(u) \cup N(v) \setminus \{u,v\}\},
\end{align*}
where $v_e$ is a newly introduced vertex and $N(a)$ denotes the set of neighbors of a vertex $a \in V$.
Then, we define the sensitivity to edge contraction for the shortest path problem as follows:
\begin{definition}
  Let $\mathcal{A}$ be a (randomized) algorithm for the shortest path problem that, given a graph $G=(V,E)$ and two vertices $s,t \in V$, outputs an $s$-$t$ walk.
  Then, the \emph{contraction sensitivity} of $\mathcal{A}$ on a graph $G=(V,E)$ is 
  \[
    \max_{e \in E, e \cap \{s,t\} = \emptyset}\EMU(\mathcal{A}(G),\mathcal{A}(G/e)).
  \]
\end{definition}
Recall that when calculating the distance $\EMU$, a walk $P = (e_1,\ldots,e_\ell)$ is mapped to a vector $\sum_{i=1}^\ell\bm{1}_{e_i}$, and hence an edge that appears multiple times in $P$ contributes by the same number of times in the vector.
We only consider contracting edges not incident to $s$ and $t$, because we do not need to contract such edges in our Lipschitz continuous algorithm in Section~\ref{sec:sp-Lipschitz}, and this simplifies the analysis.

The goal of this section is to show the following:
\begin{theorem}\label{thm:spsensitivity}
For any $\epsilon>0$, there is a $(1+\epsilon)$-approximation algorithm for the  shortest path problem with contraction sensitivity $O\left(\epsilon^{-1}\log^3 |V|\right)$. 
\end{theorem}

Varma and Yoshida~\cite{Varma2021} showed that any algorithm with a finite approximation ratio must have $\Omega(n)$ sensitivity with respect to edge deletion.
This is because in the edge-deletion case, an edge in the output walk may disappear and we may need to completely change the output walk.
However, in the edge-contraction case, we can use (essentially) the same path even if an edge in the path is contracted.

In Section~\ref{subsec:sp-overview}, we present our algorithm and an overview of the analysis.
We analyze the approximation ratio in Section~\ref{subsec:sp-approx} and the contraction sensitivity in Sections~\ref{subsec:sp-pivot} and~\ref{subsec:sp-sensitivity}.

\begin{algorithm}[t!]
\caption{Shortest path algorithm stable against edge contractions}\label{alg:main}
\Procedure{\emph{\Call{Rec}{$G,s,t,\gamma$}}}{
    Sample $\bm{d}$ uniformly from $\left[\frac{1}{4}+2\gamma, \frac{3}{4}-2\gamma\right]$\;\label{line:choosed}
    Sample $\bm{l}$ uniformly from $\left[\gamma, 2\gamma\right]$\;\label{line:choosel}
    \If{$\OPT(s,t)\leq \gamma^{-1}$}{
        \Return the $s$-$t$ shortest path on $G$ computed by breadth-first search.
    }
    Let $V_{\bm{d},\bm{l}}$ be the set of vertices $v \in V$ with $\OPT(s,v)\leq (\bm{d}+\bm{l}) \OPT(s,t)$ and $\OPT(v,t)\leq \left(1-\bm{d}+\bm{l}\right) \OPT(s,t)$\label{line:Ud}\;
    Sample $\bm{v}\in V_{\bm{d},\bm{l}}$ uniformly at random\label{line:main-sample}\;
    \Return the walk obtained by concatenating \Call{Rec}{$G,s,\bm{v},\gamma$} and \Call{Rec}{$G,\bm{v},t,\gamma$}.
}
\Procedure{\emph{\Call{SP}{$G=(V,E), s, t, \epsilon$}}}{
    Sample $\bm{\gamma}^{-1}$ uniformly from $\left[720\epsilon^{-1}\log|V|,1440\epsilon^{-1}\log|V|\right]$\;\label{line:main_samplegamma}
    \Return \Call{Rec}{$G,s,t,\bm{\gamma}$}.
}
\end{algorithm}

\subsection{Proof Overview}\label{subsec:sp-overview}
In this section, we present our algorithm for the unweighted shortest path problem (pseudocode is presented in Algorithm~\ref{alg:main}).

To construct a walk between two vertices $s,t\in V$ in the input graph $G=(V,E)$, we first sample a vertex $\bm{v} \in V$, which we call a \emph{pivot}, such that there is an $s$-$t$ path passing through $\bm{v}$ whose length is close to that of the shortest one and the sampling process of $\bm{v}$ is stable against edge contractions. 
Then, we recursively construct $s$-$\bm{v}$ and $\bm{v}$-$t$ walks and output an $s$-$t$ walk obtained by combining them.

Before discussing the approximation and contraction sensitivity of the algorithm, we  introduce useful notation.
Let $G=(V,E)$ be an unweighted graph.
For two vertices $s,t \in V$, let $\OPT_G(s,t)$ denote the length of the $s$-$t$ shortest path of $G$.
For a vertex $v \in V$, let $\OPT_G(s,t,v)$ denote the length of the shortest walk from $s$ to $t$ passing through $v$, and for an edge $e \in E$, let $\OPT_G(s,t,e)$ denote the length of the shortest walk from $s$ to $t$ passing through $e \in E$. 
We call a path (resp., walk) from a vertex $s \in V$ to one of the endpoints of an edge $e \in E$ an \emph{$s$-$e$ path} (resp., \emph{$s$-$e$ walk}). 
Let $\OPT_G(s,e)$ denote the length of the shortest $s$-$e$ path.
If the underlying graph $G$ is clear from the context, we omit the subscripts from these symbols.

\paragraph{Approximation Ratio}
Let us discuss the approximation ratio of the algorithm above.
First, the choice of the range from which we sample $\bm{l}$ in Line~\ref{line:choosel} ensures that restricting the feasible domain to the family of $s$-$t$ walks passing through the pivot $\bm{v}$ increases the optimal value only by a $(1+O(\gamma))$ factor.
Specifically, we have
\begin{align*}
    \OPT(s,t,\bm{v})=\OPT(s,\bm{v})+\OPT(\bm{v},t)\leq (1+2l)\OPT(s,t)\leq (1+4\gamma)\OPT(s,t).
\end{align*}
According to the choice of the range from which we sample $\bm{d}$ in Line~\ref{line:choosed}, the optimal values for the two subproblems induced by the pairs $(s,\bm{v})$ and $(\bm{v},t)$ are smaller than the original one by a constant factor.
Specifically, we have
\begin{align*}
    \max\{\OPT(s,\bm{v}), \OPT(\bm{v},t)\}\leq \frac{3}{4}\OPT(s,t).
\end{align*}
This implies that the recursion depth is $O(\log |V|)$.
It follows that the approximation ratio is bounded by
\begin{align*}
    (1+O(\gamma))^{O(\log |V|)}\leq 1+O(\gamma \log |V|)=1+O(\epsilon).
\end{align*}



\paragraph{Contraction Sensitivity}
Now, we consider the contraction sensitivity of the algorithm.
First, we focus on a particular call of \Call{Rec}{} and show that the pivot distributions for \Call{Rec}{$G,s,t,\gamma$} and \Call{Rec}{$G/e,s,t,\gamma$} are close, or more specifically, their total variation distance is $O\left(\frac{\log |V|}{\gamma \cdot \OPT(s,t)}\right)$. 
To bound the earth mover's distance in the definition of contraction sensitivity, we consider transporting the probability mass of \Call{Rec}{$G,s,t,\gamma$} corresponding to a particular choice of the pivot $\bm{v}$ to that of \Call{Rec}{$G/e,s,t,\gamma$} corresponding to the same pivot as far as possible. 
We transport the remaining probability mass arbitrarily, which increases the contraction sensitivity of the algorithm by
\begin{align*}
    O\left(\frac{\log |V|}{\gamma \cdot \OPT(s,t)}\right)\cdot (1+O(\gamma\log |V|))\OPT(s,t)\leq O\left(\frac{\log |V|}{\gamma}\right),
\end{align*}
where $(1+O(\gamma\log |V|))\OPT(s,t)$ is the upper bound on the length of the output walk.

Next, we consider the contraction sensitivity of the whole algorithm.
A recursion call \Call{Rec}{$G,s,t,\gamma$} is \emph{active} if $\OPT_G(s,t,e)\leq (1+k\gamma)\OPT_G(s,t)$ and \emph{inactive} otherwise, where $k=16$. We can show that, if \Call{Rec}{$G,s,t,\gamma$} is inactive, then
\begin{itemize}
    \item the distribution of the pivots chosen in \Call{Rec}{$G,s,t,\gamma$} and \Call{Rec}{$G/e,s,t,\gamma$} are exactly the same, and
    \item the recursion calls \Call{Rec}{$G,s,\bar{v},\gamma$} and \Call{Rec}{$G,\bar{v},t,\gamma$} are also inactive.
\end{itemize}
Thus, inactive recursion calls are irrelevant to the contraction sensitivity.

Assume that \Call{Rec}{$G,s,t,\gamma$} is active and the same pivot $\bm{v}$ is chosen in \Call{Rec}{$G,s,t,\gamma$} and \Call{Rec}{$G/e,s,t,\gamma$}. 
Then, we can show that, with probability $1-O(\gamma)$, at least one of the two recursion calls \Call{Rec}{$G/e,s,\bm{v},\gamma$} and \Call{Rec}{$G/e,\bm{v},t,\gamma$} is inactive. 
Intuitively, this holds because both recursion calls are active only when $\left|\OPT_G(s,e)-\bm{d}\right|\leq O(\gamma)$, which occurs with probability $O(\gamma)$ because we sampled $\bm{d}$ uniformly from $\left[\frac{1}{4}+2\gamma, \frac{3}{4}-2\gamma\right]$. 
Then, the number of recursion calls relevant to contraction sensitivity is bounded by $(1+O(\gamma))^{O(\log |V|)}=O(\log |V|)$. 
The contraction sensitivity of the algorithm is bounded as
\[
    O\left(\frac{\log |V|}{\gamma}\right) \cdot O(\log |V|)=O\left(\frac{\log^2 |V|}{\gamma}\right)=O\left(\frac{\log^3 |V|}{\epsilon}\right).
\]

\subsection{Approximation Ratio}\label{subsec:sp-approx}

In this section, we analyze the approximation ratio of Algorithm~\ref{alg:main}.
First, we analyze the approximation ratio of \Call{Rec}{}.
\begin{lemma}\label{lem:expapprox}
    For the pivot $\bm{v}$ sampled in  $\Call{Rec}{G,s,t,\gamma}$, we have
    \begin{align*}
        \OPT(s,\bm{v})+\OPT(\bm{v},t)\leq (1+4\gamma)\OPT(s,t).
    \end{align*}
\end{lemma}
\begin{proof}
We have
\[
    \OPT(s,\bm{v})+\OPT(\bm{v},t)\leq (\bm{d}+\bm{l})\OPT(s,t)+(1-\bm{d}+\bm{l})\OPT(s,t) = (1+2\bm{l})\OPT(s,t)\leq (1+4\gamma)\OPT(s,t).
    \qedhere
\]
\end{proof}
\begin{lemma}\label{lem:approx}
For $\gamma > 0$, we have
\[
    \left|\Call{Rec}{G,s,t,\gamma}\right|\leq \OPT(s,t)^{1+14\gamma}.
\]
\end{lemma}
\begin{proof}
We prove by induction on $\OPT(s,t)$. 
The statement clearly holds when $\OPT(s,t) \leq \gamma^{-1}$.

Suppose that $\OPT(s,t)>\gamma^{-1}$.
Then, we have
\begin{align*}
    \left|\Call{Rec}{G,s,t,\gamma}\right|
    &\leq \max_{d,l,v\in V_{d,l}}\left(\left|\Call{Rec}{G,s,v,\gamma}\right|+\left|\Call{Rec}{G,v,t,\gamma}\right|\right)\\
    &\leq \max_{d,l,v\in V_{d,l}}\left(\OPT(s,v)^{1+14\gamma}+\OPT(v,t)^{1+14\gamma}\right)\\
    &= \max_{d,l,v\in V_{d,l}}\left(\OPT(s,v)\cdot \OPT(s,v)^{14\gamma}+\OPT(v,t)\cdot \OPT(v,t)^{14\gamma}\right)\\
    &\leq \left(\frac{3}{4}\OPT(s,t)\right)^{14\gamma} \max_{d,l,v\in V_d}\left(\OPT(s,v)+\OPT(v,t)\right) \\
    &\leq \left(\frac{3}{4}\OPT(s,t)\right)^{14\gamma}(1+4\gamma)\OPT(s,t)\\
    &= \left(\left(\frac{3}{4}\right)^{14\gamma}(1+4\gamma)\right)\OPT(s,t)^{1+14\gamma}\\
    &\leq \OPT(s,t)^{1+14\gamma},
\end{align*}
where the second inequality is from the induction hypothesis, the third inequality is from how we select the pivot, the fourth inequality is from Lemma~\ref{lem:expapprox}, and the last inequality is from $1+4\gamma \leq \left(\frac{4}{3}\right)^{14\gamma}$.
\end{proof}

Now, we bound the approximation ratio of \Call{SP}{}.
\begin{lemma}\label{lem:spapp}
The approximation ratio of $\Call{SP}{G,s,t,\epsilon}$ is (at least) $1+\epsilon$.
\end{lemma}
\begin{proof}
By Lemma~\ref{lem:approx}, \Call{SP}{} outputs a walk of length $\OPT(s,t)^{1+14\gamma}$, which is bounded as
\begin{align*}
    & \OPT(s,t)^{1+14\gamma} \leq 
    \OPT(s,t)^{\frac{14\epsilon}{720\log |V|}}\cdot \OPT(s,t)
    \leq |V|^{\frac{14\epsilon}{720\log |V|}}\cdot \OPT(s,t)\\
    &= \exp\left(\frac{14\epsilon}{720}\right)\cdot \OPT(s,t) 
    \leq (1+\epsilon)\OPT(s,t),
\end{align*}
where the first inequality is from $\gamma\leq \frac{\epsilon}{720\log |V|}$, the second inequality is from $\OPT(s,t)\leq |V|$, and the last inequality is from $\exp(x)\leq 1+2x$, which holds for $x\leq \frac{1}{2}$.
\end{proof}


\subsection{Total Variation Distance Between Pivots}\label{subsec:sp-pivot}
In this section, we bound the total variation distance between pivots selected in a particular call of \Call{Rec}{}.
Let $\bm{v}$ and $\bm{v}^e$ be the pivots sampled in \Call{Rec}{$G,s,t,\gamma$} and \Call{Rec}{$G/e,s,t,\gamma$}, respectively. 
We denote the variables $V,\bm{d},\bm{l}$ used in \Call{Rec}{$G/e,s,t,\gamma$} as $V^e,\bm{d}^e,\bm{l}^e$ to distinguish them from those used in \Call{Rec}{$G,s,t,\gamma$}.
Our goal is to prove the following:
\begin{lemma}\label{lem:TVD}
We have
\begin{align*}
    \TV\left(\bm{v},\bm{v}^e\right)
    \leq \frac{4\log n}{\gamma \cdot \OPT_G(s,t)}.
\end{align*}
\end{lemma}
We start with several lemmas.
\begin{lemma}\label{lem:Vinclusion}
    For any $d$ and $l$ (in the ranges specified in \Call{Rec}{}), we have $V^e_{d,l}  \subseteq V_{d,l+\frac{1}{\OPT_G(s,t)}}\cup \{v_e\}$
    (recall that $v_e$ is the vertex introduced when constructing $G/e$ from $G$).
\end{lemma}
\begin{proof}
    Let $v\in V^e_{d,l} \setminus \{v_e\}$. Then, we have
    \begin{align*}
        \OPT_{G}(s,v)&\leq \OPT_{G/e}(s,v) + 1 \leq (d+l)\OPT_{G/e}(s,t) + 1\\
        &\leq (d+l)\OPT_{G}(s,t) + 1 = \left(d+l+\frac{1}{\OPT_{G}(s,t)}\right)\OPT_{G}(s,t),\\
        \OPT_{G}(v,t)&\leq \OPT_{G/e}(v,t) + 1 \leq (1-d+l)\OPT_{G/e}(s,t) + 1\\
        &\leq (1-d+l)\OPT_{G}(s,t) + 1 = \left(1-d+l+\frac{1}{\OPT_{G}(s,t)}\right)\OPT_{G}(s,t),
    \end{align*}
    and the claim holds.
\end{proof}

\begin{lemma}\label{lem:sizeofV}
For any $d$ and $l$ (in the ranges specified in \Call{Rec}{}), we have $|V_{d,l}|\geq \gamma\cdot \OPT_G(s,t)$.
\end{lemma}
\begin{proof}
Let $P$ be an arbitrary $s$-$t$ shortest path. 
From the definition of $V_{d,l}$, any vertex $v'$ on $P$ with $(d-l)\OPT(s,t)\leq \OPT_G(s,v')\leq (d+l)\OPT(s,t)$ is contained in $V_{d,l}$. 
Therefore, we have 
\[
    |V_{d,l}|\geq \floor{(d+l)\OPT(s,t)}-\ceil{(d-l)\OPT(s,t)}\geq 2l\OPT(s,t)-1\geq \gamma\cdot \OPT(s,t).
    \qedhere
\]
\end{proof}
\begin{corollary}\label{cor:specvert}
    For any $v\in V$, $d$, and $l$ (in the ranges specified in \Call{Rec}{}), we have $\Pr[\bm{v}=v\mid \bm{d}=d,\bm{l}=l]\leq \frac{1}{\gamma\cdot \OPT_G(s,t)}$.
\end{corollary}
The next technical lemma is useful in our analysis.
\begin{lemma}\label{lem:arraylog}
Let $a_0,\dots, a_k$ be a non-decreasing array with $a_0>0$. Then,
\begin{align*}
    \sum_{d=0}^{k-1}\left(1-\frac{a_d}{a_{d+1}}\right)\leq \log \frac{a_k}{a_0}.
\end{align*}
\end{lemma}
\begin{proof}
We prove by induction on $k$. 
When $k=0$, the statement clearly holds. 

Suppose $k\geq 1$ and the statement holds for $k-1$.
Then, we have
\begin{align*}
    \sum_{d=0}^{k-1}\left(1-\frac{a_d}{a_{d+1}}\right)
    \leq \log \frac{a_{k-1}}{a_0} + \left(1 - \frac{a_{k-1}}{a_{k}}\right)
    \leq \log \frac{a_{k-1}}{a_0} + \log\frac{a_k}{a_{k-1}} = \log \frac{a_k}{a_0},
\end{align*}
where the last inequality is from $1-\frac{1}{x}\leq \log x$ for $x\geq 1$.
\end{proof}

The following lemma proves that $\TV(\bm{v},\bm{v}^e)$ is small when $\bm{d}=\bm{d}^e$.
\begin{lemma}\label{lem:dcond}
For any $d$ (in the range specified in \Call{Rec}{}), we have
\begin{align*}
    \TV\left(\left(\bm{v}\mid \bm{d}=d\right),\left(\bm{v}^e\mid \bm{d}^e=d\right)\right)
    \leq \frac{2}{\gamma \cdot \OPT_G(s,t)}. 
\end{align*}

\end{lemma}
\begin{proof}
We have
\begin{align*}
    &\TV\left(\left(\bm{v}\mid \bm{d}=d\right),\left(\bm{v}^e\mid \bm{d}^e=d\right)\right)\\
    &\leq \frac{1}{\gamma}\int_{\gamma}^{2\gamma}\TV\left(\left(\bm{v}\mid \bm{d}=d, \bm{l}=l\right),\left(\bm{v}^e\mid \bm{d}^e=d,\bm{l}^e=l\right)\right)\mathrm{d}l\\
    &= \frac{1}{\gamma}\int_{\gamma}^{2\gamma}\left(\sum_{v\in V_{d,l}}\max\left(0,\Pr\left[\bm{v}=v\mid \bm{d}=d,\bm{l}=l\right]-\Pr\left[\bm{v}^e=v\mid \bm{d}^e=d,\bm{l}^e=l\right]\right)\right)\mathrm{d}l\\
    &\leq \frac{1}{\gamma}\int_{\gamma}^{2\gamma}\left(\Pr[\bm{v}\in e\mid \bm{d}=d,\bm{l}=l] + \sum_{v\in V_{d,l}\setminus e}\max\left(0,\frac{1}{\left|V_{d,l}\right|}-\frac{1}{\left|V^e_{d,l}\right|}\right)\right)\mathrm{d}l\\
    &\leq \frac{1}{\gamma}\int_{\gamma}^{2\gamma}\left(\frac{2}{\gamma\cdot \OPT_G(s,t)} + \sum_{v\in V_{d,l}\setminus e}\max\left(0,\frac{1}{\left|V_{d,l}\right|}-\frac{1}{\left|V^e_{d,l}\right|}\right)\right)\mathrm{d}l\\
    &\leq \frac{2}{\gamma \cdot \OPT_G(s,t)} + \frac{1}{\gamma}\int_{\gamma}^{2\gamma}\left(1-\frac{\left|V_{d,l}\right|}{\left|V^e_{d,l}\right|}\right)\mathrm{d}l\\
    &\leq \frac{2}{\gamma \cdot \OPT_G(s,t)} + \gamma^{-1}\int_{\gamma}^{2\gamma}\left(1-\frac{\left|V_{d,l}\right|}{\left|V_{d,l+\frac{1}{\OPT_G(s,t)}}\cup  \{v_e \}\right|}\right)\mathrm{d}l,
\end{align*}
where the second inequality is from Line~\ref{line:Ud} of Algorithm~\ref{alg:main}, the third inequality is from Corollary~\ref{cor:specvert}, the fourth inequality is from $\left|V_{d,l} \setminus e \right|\leq \left|V_{d,l}\right|$, and the last inequality is from Lemma~\ref{lem:Vinclusion}.

Let $c=\frac{1}{\OPT_G(s,t)}$. Then, we have
\begin{align*}
    &\int_{\gamma}^{2\gamma}\left(1-\frac{\left|V_{d,l}\right|}{\left|V_{d,l+\frac{1}{\OPT_G(s,t)}}\cup \{v_e\}\right|}\right)\mathrm{d}l\\
    &=\int_{\gamma}^{2\gamma}\left(1-\frac{\left|V_{d,l}\right|}{\left|V_{d,l+c}\right|+1}\right)\mathrm{d}l\\
    &= \int_{\gamma}^{\gamma+c}\left(\sum_{k=0}^{\floor{\frac{2\gamma-a}{c}}}\left(1-\frac{\left|V_{d,a+kc}\right|}{\left|V_{d,a+(k+1)c}\right|+1}\right)\right)\mathrm{d}a \\
    &\leq \int_{\gamma}^{\gamma+c}\left(\sum_{k=0}^{\floor{\frac{2\gamma-a}{c}}}\left(1-\frac{\left|V_{d,a+kc}\right|}{\left|V_{d,a+(k+1)c}\right|}+\frac{1}{\left|V_{d,a+(k+1)c}\right|}\right)\right)\mathrm{d}a\\
    &\leq \int_{\gamma}^{\gamma+c}\left(\frac{1}{\gamma\cdot \OPT_G(s,t)}\cdot \floor{\frac{2\gamma-a}{c}}+\sum_{k=0}^{\floor{\frac{2\gamma-a}{c}}}\left(1-\frac{\left|V_{d,a+kc}\right|}{\left|V_{d,a+(k+1)c}\right|}\right)\right)\mathrm{d}a\\
    &\leq \int_{\gamma}^{\gamma+c}\left(1+\sum_{k=0}^{\floor{\frac{2\gamma-a}{c}}}\left(1-\frac{\left|V_{d,a+kc}\right|}{\left|V_{d,a+(k+1)c}\right|}\right)\right)\mathrm{d}a\\
    &\leq \int_{\gamma}^{\gamma+c}\left(1+\log\left(\frac{\left|V_{d,a+\left(\floor{\frac{2\gamma-a}{c}}+1\right)c}\right|}{\left|V_{d,a}\right|}\right)\right)\mathrm{d}a\\
    &\leq \int_{\gamma}^{\gamma+c}\left(1+\log|V|\right)\mathrm{d}a
    = \frac{1+\log |V|}{\OPT_G(s,t)},
\end{align*}
where the first inequality is from $\left|V_{d,a+kc}\right|\leq \left|V_{d,a+(k+1)c}\right|$ and
\begin{align*}
    1-\frac{x}{x+y+z} = 1-\frac{x}{x+y}+\frac{xz}{(x+y)(x+y+z)}\leq 1-\frac{x}{x+y}+\frac{z}{x+y}
\end{align*}
holds for $x,y,z\geq 0$, the second inequality is from $\left|V_{d,a+(k+1)c}\right|\geq \gamma\cdot \OPT_G(s,t)$, the third inequality is from the definition of $c$, the fourth inequality is from Lemma~\ref{lem:arraylog}, and the last inequality is from $1\leq \left|V_{d,l}\right|\leq |V|$ for all $l$. Therefore, we have
\begin{align*}
    &\TV\left(\left(\bm{v}\mid \bm{d}=d\right),\left(\bm{v}^e\mid \bm{d}^e=d\right)\right)
    \leq \frac{2}{\gamma \cdot \OPT_G(s,t)} + \frac{1}{\gamma} \cdot \frac{1+\log |V|}{\OPT_G(s,t)}
    \leq \frac{4\log |V|}{\gamma \cdot \OPT_G(s,t)}.
    \qedhere
\end{align*}
\end{proof}

\begin{proof}[Proof of Lemma~\ref{lem:TVD}]
Because the distributions of $\bm{d}$ and $\bm{d}^e$ are identical, we can omit the condition $\bm{d}=\bm{d}^e=d$ of Lemma~\ref{lem:dcond} by considering the natural probability transportation.
\end{proof}

\subsection{Contraction Sensitivity}\label{subsec:sp-sensitivity}
In this section, we analyze the contraction sensitivity of Algorithm~\ref{alg:main}.
Our goal is to prove the following.
\begin{lemma}\label{lem:spsens}
The contraction sensitivity of $\Call{SP}{G,s,t,\epsilon}$ is $O(\epsilon^{-1}\log^3 |V|)$.
\end{lemma}

Let $e\in E$ be the edge that will be contracted.
To evaluate the earth mover's distance between $\Call{SP}{G,s,t,\epsilon}$ and $\Call{SP}{G/e,s,t,\epsilon}$, we consider transporting the probability mass of \Call{SP}{$G,s,t,\epsilon$} corresponding to a particular choice of $\bm{\gamma}$ to that of \Call{SP}{$G/e,s,t,\epsilon$} corresponding to the same choice of $\bm{\gamma}$. 
The remaining mass is transported arbitrarily. 
Let us fix $\gamma$ and analyze
\begin{align*}
    \EMU(\Call{Rec}{G,s,t,\gamma},\Call{Rec}{G/e,s,t,\gamma}).
\end{align*}

We transport the probability mass of \Call{Rec}{$G,s,t,\gamma$} corresponding to a particular choice of the pivot $\bm{v}$ to that of \Call{Rec}{$G/e,s,t,\gamma$} corresponding to the same choice of the pivot as far as possible. 
For a fixed $\bm{v}$, we transport the probability mass from \Call{Rec}{$G,s,\bm{v},\gamma$} (resp., \Call{Rec}{$G,\bm{v},t,\gamma$}) to \Call{Rec}{$G/e,s,\bm{v},\gamma$} (resp., \Call{Rec}{$G/e,\bm{v},t,\gamma$}) as far as possible, where the mass is transported recursively, as in the case of \Call{Rec}{$G,s,t,\gamma$}.
Because \Call{Rec}{$G,s,\bm{v},\gamma$} and \Call{Rec}{$G,\bm{v},t,\gamma$} (resp., \Call{Rec}{$G/e,s,\bm{v},\gamma$} and \Call{Rec}{$G/e,\bm{v},t,\gamma$}) are independent for fixed $\bm{v}$, we obtain a probability transportation from the probability mass of \Call{Rec}{$G,s,t,\gamma$} to that of \Call{Rec}{$G/e,s,t,\gamma$} as their direct product.
The remaining probability mass is transported arbitrarily.


In Section~\ref{subsubsec:active-calls}, we introduce the notion of active and inactive recursion calls, which are useful for bounding the earth mover's distance.
We prove Lemma~\ref{lem:spsens} in Section~\ref{subsubsec:proof-of-spsens}.

\subsubsection{Active and Inactive Recursion Calls}\label{subsubsec:active-calls}
Let $k=16$. 
Recall that we say that a call $\Call{Rec}{G,s,t,\gamma}$ is \emph{active} if $\OPT_G(s,t,e)\leq (1+k\gamma)\OPT_G(s,t)$ and \emph{inactive} otherwise. 
In this section, we present several properties of active and inactive calls of \Call{Rec}{}.
We start with the following.
\begin{lemma}
If $\Call{Rec}{G,s,t,\gamma}$ is inactive, then both $\Call{Rec}{G,s,\bm{v},\gamma}$ and $\Call{Rec}{G,\bm{v},t,\gamma}$ are inactive.
\end{lemma}
\begin{proof}
By symmetry, it suffices to prove the claim only for $\Call{Rec}{G,s,\bm{v},\gamma}$. Assume that $\Call{Rec}{G,s,\bm{v},\gamma}$ is active. Then,
\begin{align*}
    \OPT_G(s,t,e) &\leq \OPT_G(s,\bm{v},e)+\OPT_G(\bm{v},t) \leq (1+k\gamma)\OPT_G(s,\bm{v}) + \OPT_G(\bm{v},t)\\
    &= k\gamma\cdot \OPT_G(s,\bm{v}) + \left(\OPT_G(s,\bm{v})+\OPT_G(\bm{v},t)\right) \\
    &\leq \frac{3}{4}k\gamma\cdot \OPT_G(s,t) + (1+4\gamma)\OPT_G(s,t)\leq (1+k\gamma)\OPT_G(s,t),
\end{align*}
where the first inequality is from the triangle inequality, the second inequality is from the assumption that $\Call{Rec}{G,s,\bm{v},\gamma}$ is active, the third inequality if from Line~\ref{line:Ud} of Algorithm~\ref{alg:main}, and the last inequality is from $k\geq 16$.
It follows that $\Call{Rec}{G,s,t,\gamma}$ is active, which is a contradiction.
\end{proof}

The next lemma shows that inactive recursive calls have no contribution to the sensitivity.
\begin{lemma}\label{lem:inactivetoitself}
If $\Call{Rec}{G,s,t,\gamma}$ is inactive, then we have $\TV\left(\bm{v},\bm{v}^e\right)=0$.
\end{lemma}
\begin{proof}
Because $\Call{Rec}{G,s,t,\gamma}$ is inactive, we have $\OPT_G(s,t)=\OPT_{G/e}(s,t)$. Therefore, we have $V_{\bm{d},\bm{l}}\subseteq V^e_{\bm{d},\bm{l}}\cup e$. For $v\in e$, we have
\begin{align}
    \OPT_G(s,v)+\OPT_G(v,t)\geq \OPT_G(s,t,e)-2 \geq (1+k\gamma)\OPT_G(s,t)-2 > (1+4\gamma)\OPT_G(s,t),\label{eq:breakdirected1}
\end{align}
where the last inequality is from $\gamma\cdot \OPT_G(s,t)\geq 1$.
Therefore, $v\not \in V_{\bm{d},\bm{l}}$. 
For the vertex $v_e$ in $G/e$ corresponding to $e$, we have
\begin{align}
    \OPT_{G/e}(s,v_e)+\OPT_{G/e}(v_e,t)&= \OPT_G(s,e)+\OPT_G(e,t) \geq \OPT_G(s,t,e)-2\label{eq:breakdirected2}\\
    &> (1+k\gamma)\OPT_G(s,t)-1\geq (1+4\gamma)\OPT_G(s,t) = (1+4\gamma)\OPT_{G/e}(s,t),\nonumber
\end{align}
where the last inequality is from $\gamma\cdot \OPT_G(s,t)\geq 1$.
Therefore, $v_e\not \in V^e_{\bm{d},\bm{l}}$.
Now, assume that $v\in V^e_{\bm{d},\bm{l}}$. 
Then, $s$-$v$ shortest paths and $v$-$t$ shortest paths on $G$ do not include $e$,  because otherwise we have
\begin{align*}
    \OPT_G(s,t,e)\leq \OPT_G(s,v)+\OPT_G(v,t)\leq (1+4\gamma)\OPT_G(s,t).
\end{align*}
and $\Call{Rec}{G,s,t,\gamma}$ is active. Therefore, we have 
\begin{align*}
    \OPT_G(s,v)+\OPT_G(v,t)=\OPT_{G/e}(s,v)+\OPT_{G/e}(v,t)\leq (1+4\gamma)\OPT_{G/e}(s,t) = (1+4\gamma)\OPT_G(s,t)
\end{align*}
and $v\in V_{\bm{d},\bm{l}}$. 
Therefore, $V_{\bm{d},\bm{l}} =  V^e_{\bm{d},\bm{l}}$ holds, and the lemma is proved.
\end{proof}

The next lemma shows that with probability $1-O(\gamma)$, at most one of the two recursive calls in \Call{Rec}{$G,s,t,\gamma$} is active.
\begin{lemma}\label{lem:bothactive}
If $\Call{Rec}{G,s,t,\gamma}$ is active, then the probability that both $\Call{Rec}{G,s,\bm{v},\gamma}$ and $\Call{Rec}{G,\bm{v},t,\gamma}$ are active is at most $(8+4k)\gamma = 72\gamma$.
\end{lemma}
\begin{proof}
Assume that both $\Call{Rec}{G,s,\bm{v},\gamma}$ and $\Call{Rec}{G,\bm{v},t,\gamma}$ are active. Then, we have
\begin{align*}
    \OPT_G(s,e) + 1&\leq \OPT_G(s,\bm{v},e)\leq (1+k\gamma)\OPT_G(s,\bm{v})\leq (1+k\gamma)(\bm{d}+\bm{l})\OPT_G(s,t),\\
    \OPT_G(e,t) + 1&\leq \OPT_G(\bm{v},t,e)\leq (1+k\gamma)\OPT_G(\bm{v},t)\leq (1+k\gamma)(1-\bm{d}+\bm{l})\OPT_G(s,t),
\end{align*}
where the second inequalities of the two derivations are from the activeness of the recursive calls and the last inequalities are from Line~\ref{line:Ud} of Algorithm~\ref{alg:main}.
Therefore, we have
\begin{align*}
    \frac{\OPT_G(s,e) + 1}{(1+k\gamma)\OPT_G(s,t)}-\bm{l} \leq \bm{d} \leq 1+\bm{l}-\frac{\OPT_G(e,t) + 1}{(1+k\gamma)\OPT_G(s,t)}.
\end{align*}
Thus, we have
\begin{align*}
    &\Pr[\text{both }\Call{Rec}{G,s,\bm{v},\gamma}\text{ and }\Call{Rec}{G,\bm{v},t,\gamma}\text{ are active}]\\
    &\leq \Pr\left[\bm{d}\in \left[\frac{\OPT(s,e) + 1}{(1+k\gamma)\OPT(s,t)}-\bm{l},1+\bm{l}-\frac{\OPT(e,t) + 1}{(1+k\gamma)\OPT(s,t)}\right]\right]\\
    &= \frac{1}{(3/4-2\gamma)-(1/4+2\gamma)}\cdot \left(\left(1+\bm{l}-\frac{\OPT(e,t) + 1}{(1+k\gamma)\OPT(s,t)}\right)-\left(\frac{\OPT(s,e) + 1}{(1+k\gamma)\OPT(s,t)}-\bm{l}\right)\right)\\
    &\leq 2\cdot \left(1+2\bm{l}-\frac{\OPT(s,e)+\OPT(e,t)+2}{(1+k\gamma)\OPT(s,t)}\right)\\
    &\leq 2\cdot \left(1+2\bm{l}-\frac{\OPT(s,t)}{(1+k\gamma)\OPT(s,t)}\right)\\
    &= 2\cdot \left(1+2\bm{l}-\frac{1}{1+k\gamma}\right)\\
    &\leq 2\cdot \left(1+2\bm{l}-(1-2k\gamma)\right)
    = 4\bm{l}+4k\gamma \leq (8+4k)\gamma,
\end{align*}
where the first equality is from the range from which we sample $d$, the third inequality is from the triangle inequality, the fourth inequality is from $k\gamma\leq \frac{1}{2}$ (because $k=16$ and $\gamma\leq \frac{\epsilon}{360\log |V|}\leq 0.01$), and the last inequality is from the range from which we sample $l$.
\end{proof}

\subsubsection{Proof of Lemma~\ref{lem:spsens}}\label{subsubsec:proof-of-spsens}
First we bound the sensitivity caused by sampling $\gamma$.
\begin{lemma}\label{lem:contsp_samplegamma}
We have
\begin{align*}
    \TV\left(\bm{\gamma}, \bm{\gamma}^{e}\right)\leq \frac{2}{|V|}.
\end{align*}
\end{lemma}
\begin{proof}
We have
\begin{align*}
    \TV\left(\bm{\gamma}, \bm{\gamma}^{e}\right) &= \Pr\left[\bm{\gamma}^{-1}\in \left[1440\epsilon^{-1}\log (|V|-1), 1440\epsilon^{-1}\log |V|\right]\right]\\
    &=\frac{1440\epsilon^{-1}(\log(|V|)-\log(|V|-1))}{720\epsilon^{-1}\log(|V|)}
    =2\cdot \left(1-\frac{\log(|V|-1)}{\log(|V|)}\right)\leq \frac{2}{|V|}.\qedhere
\end{align*}
\end{proof}

\begin{lemma}\label{lem:recdep}
The recursion depth is at most $4\log |V|$.
\end{lemma}
\begin{proof}
At each call of \Call{Rec}{}, we have $\max\{\OPT_G(s,\bm{v}),\OPT_G(\bm{v},t)\}\leq \frac{3}{4}\OPT_G(s,t)$. Hence, the recursion depth is at most $\log_{4/3}\OPT(s,t)\leq 4\log \OPT(s,t)\leq 4\log |V|$. 
\end{proof}

Using the lemmas above and Lemma~\ref{lem:TVD}, we prove the following:
\begin{lemma}\label{lem:emd}
\begin{align*}
    \EMU\left(\Call{Rec}{G,s,t,\gamma},\Call{Rec}{G/e,s,t,\gamma}\right)\leq \frac{1}{\gamma^2}|V|^{14\gamma}\left(|V|^{360\gamma}-1\right)\log |V|.
\end{align*}
\end{lemma}
\begin{proof}
First, we bound the expected number of active recursive calls. 
We have
\begin{align*}
    \E\left[\text{The number of active recursive calls}\right]&\leq \sum_{k=0}^{\floor{4\log |V|}}(1+72\gamma)^k \\
    &= \frac{\left((1+72\gamma)^{\floor{4\log |V|}+1}-1\right)}{72\gamma}\\
    &\leq \frac{(1+72\gamma)^{5\log |V|}-1}{72\gamma}\\
    &= \frac{|V|^{5\log (1+72\gamma)}-1}{72\gamma}\leq \frac{|V|^{360\gamma}-1}{72\gamma}.
\end{align*}
where the first inequality is from Lemmas~\ref{lem:bothactive} and~\ref{lem:recdep}, the first equality is from
\begin{align*}
    \sum_{k=0}^{a}(1+x)^k=\left((1+x)^a-(1+x)^{-1}\right)\cdot \frac{1}{1-\frac{1}{1+x}}=\frac{(1+x)^{a+1}-1}{x},
\end{align*}
and the last inequality is from $\log(1+x)\leq x$.
From Lemmas~\ref{lem:approx} and~\ref{lem:TVD}, each active recursion call \Call{Rec}{$G,s',t',\gamma$} contributes to $\EMU\left(\Call{Rec}{G,s,t,\gamma}, \Call{Rec}{G/e,s,t,\gamma}\right)$ by
\begin{align*}
    \OPT_G(s',t')^{1+14\gamma}\cdot \frac{4\log |V|}{\gamma \cdot \OPT_G(s',t')}.
\end{align*}
The same holds even for the case with $\OPT_{G}(s',t') > \gamma^{-1}$ and $\OPT_{G/e}(s',t')\leq \gamma^{-1}$ because we have
\begin{align*}
    \OPT_G(s',t')^{1+14\gamma}\cdot \frac{4\log |V|}{\gamma \cdot \OPT_G(s',t')}\geq 4\gamma^{-1}\geq \gamma^{-1}+1\geq \OPT_{G}(s',t'). 
\end{align*}
Thus, we have
\begin{align*}
    &\EMU\left(\Call{Rec}{G,s,t,\gamma}, \Call{Rec}{G/e,s,t,\gamma}\right)\\
    &\leq \E\left[\sum_{\Call{Rec}{G,s',t',\gamma}\text{ is active}}\OPT_G(s',t')^{1+14\gamma}\cdot \frac{4\log |V|}{\gamma \cdot \OPT_G(s',t')}\right]\\
    &\leq \frac{|V|^{360\gamma}-1}{72\gamma}\cdot |V|^{14\gamma}\cdot \frac{4\log |V|}{\gamma}\\
    &= \frac{4}{72\gamma^2}|V|^{14\gamma}\left(|V|^{360\gamma}-1\right)\log |V|\leq \frac{1}{\gamma^2}|V|^{14\gamma}\left(|V|^{360\gamma}-1\right)\log |V|,
\end{align*}
where the second inequality is from $\OPT_G(s',t')\leq |V|$. 
\end{proof}

\begin{proof}[Proof of Lemma~\ref{lem:spsens}]
If $\epsilon^{-1}>\OPT_G(s,t)$, then the claim clearly holds because the sensitivity is at most $\OPT_G(s,t) = O(\epsilon^{-1}\log^3 |V|)$.
Therefore, we assume that $\epsilon^{-1}\leq \OPT_G(s,t)$.
\begin{align*}
    &\EMU\left(\Call{SP}{G,s,t,\epsilon},\Call{SP}{G/e,s,t,\epsilon}\right)\\
    &\leq \frac{1}{\gamma^2}|V|^{14\gamma}\left(|V|^{360\gamma}-1\right)\log |V|
    +\frac{2}{|V|}\cdot (1+\epsilon)\OPT(s,t)\\
    &\leq \frac{1}{\gamma^2}|V|^{14\gamma}\left(|V|^{360\gamma}-1\right)\log |V|+4\\
    &\leq\frac{4147200\log^3|V|}{\epsilon} +4 =O\left(\frac{\log^3 |V|}{\epsilon}\right),
\end{align*}
where the first inequality is from Lemmas~\ref{lem:contsp_samplegamma} and~\ref{lem:emd}, and the last inequality is from
\begin{align*}
    &\frac{1}{\gamma^2}|V|^{14\gamma}\left(|V|^{360\gamma}-1\right)\log |V|\\
    &\leq \left(\frac{1440\log |V|}{\epsilon}\right)^2\cdot |V|^{\frac{14\epsilon}{720\log |V|}}\cdot \left(|V|^{\frac{360\epsilon}{720\log |V|}}-1\right)\cdot \log |V|\\
    &= \left(\frac{1440\log |V|}{\epsilon}\right)^2\cdot \exp\left(\frac{14\epsilon}{720}\right)\cdot \left(\exp\left(\frac{360\epsilon}{720}\right)-1\right)\cdot \log |V|\\
    &\leq \left(\frac{1440\log |V|}{\epsilon}\right)^2\cdot (1+\epsilon)\cdot \epsilon\cdot \log |V|\\
    &\leq \frac{4147200\log^3 |V|}{\epsilon}.
\end{align*}
Here, the first inequality is from $\frac{\epsilon}{720\log |V|}\leq \gamma\leq \frac{\epsilon}{1440\log |V|}$, the second inequality is from $\exp(x)\leq 1+2x$, which holds for $x\leq \frac{1}{2}$, and the last inequality is from $\epsilon\leq 1$. 
\end{proof}
\begin{proof}[Proof of Theorem~\ref{thm:spsensitivity}]
The proof follows from Lemmas~\ref{lem:spapp} and~\ref{lem:spsens}.
\end{proof}


\section{Lipschitz Continuity of Shortest Path}\label{sec:sp-Lipschitz}
In this section, we provide our algorithm for the shortest path problem with a bounded Lipschitz constant and prove the following:
\begin{theorem}\label{thm:sp}
For any $\epsilon \in (0,1)$, there exists a polynomial-time $(1+\epsilon)$-approximation algorithm for the shortest path problem with Lipschitz constant $O\left(\epsilon^{-1}\log^3 n\right)$, where $n$ is the number of vertices in the input graph. 
\end{theorem}

Our idea is to convert the input weighted graph $(G,w)$ into an unweighted graph $\widehat{\bm{G}}$ by replacing each edge with a path whose length is determined by its weight, apply Algorithm~\ref{alg:main} on $\widehat{\bm{G}}$, and then convert the obtained walk on $\widehat{\bm{G}}$ back to a walk on $G$.

However, it is unclear how to convert a walk in $\widehat{\bm{G}}$ back to a walk in $G$ in such a way that a small contraction sensitivity of the former implies a small Lipschitz constant of the latter because the former may walk around in a path in $\widehat{\bm{G}}$ corresponding to an edge in $G$.
To resolve this issue, we replace each edge of $G$ with two \emph{directed} paths oriented in opposite directions.
Then, we can naturally map a walk in $\widehat{\bm{G}}$ to a walk in $G$ because a walk entering one end of a path must leave from the other end.
For this idea to work, we extend Algorithm~\ref{alg:main} for directed graphs.
We discuss this directed version in detail in Section~\ref{subsec:sp-directed}.

We describe our algorithm in Section~\ref{subsec:lipsp-algorithm}, and analyze its approximation and Lipschitz continuity in Sections~\ref{subsec:lipsp-approximation} and~\ref{subsec:lipsp-lipschitz}, respectively.



\subsection{Directed Shortest Path}\label{subsec:sp-directed}
In the \emph{(unweighted) directed shortest path problem}, we are given an unweighted directed graph $G=(V,E)$, a source vertex $s \in V$, and a target vertex $t \in V$, we want to compute a directed path from $s$ to $t$ with the shortest length. 
We extend the algorithm in Section~\ref{sec:sp} for the directed shortest path problem.

We must be careful when discussing contraction sensitivity for the directed shortest path problem because the contraction of a directed edge may affect the reachability relation between vertices and significantly change the optimal values.
To avoid this complex issue that is unnecessary for our reduction, we call an edge $(u,v) \in E$ in a directed graph $G=(V,E)$ \emph{contractible} if $d^+(u)=d^-(u)=1$ and $d^+(v)=d^-(v)=1$ hold, where $d^+(a)$ and $d^-(a)$ denote the outdegree and indegree, respectively, of a vertex $a \in V$, and only consider contraction of contractible edges.
Note that contracting contractible edges do not change reachability relations of vertices (except for the deleted and newly introduced vertices).
Then, we define the sensitivity for directed graphs to edge contraction as follows:
\begin{definition}
  Let $\mathcal{A}$ be a (randomized) algorithm that, given a directed graph $G=(V,E)$, outputs an edge (multi)set $\bm{F} \subseteq E$.
  Then, the \emph{directed contraction sensitivity} of $\mathcal{A}$ on a directed graph $G=(V,E)$ is 
  \[
    \max_{e \in E:\text{ contractible}}\EMU(\mathcal{A}(G),\mathcal{A}(G/e)).
  \]
\end{definition}


Now, we can show the following:
\begin{lemma}\label{lem:directed}
For any $\epsilon>0$, there exists a $(1+\epsilon)$-approximation algorithm for the directed shortest path problem with directed contraction sensitivity $O\left(\epsilon^{-1}\log^3 |V|\right)$.
\end{lemma}
The algorithm of Lemma~\ref{lem:directed} is obtained from Algorithm~\ref{alg:main} by replacing $\OPT_G(u,v)$ with the directed counterpart.
The analysis of the approximation ratio and contraction sensitivity is almost identical to that for Algorithm~\ref{alg:main}. 
We use contractibility of edges only to prove Lemmas~\ref{lem:Vinclusion} and~\ref{lem:inactivetoitself}.
In the proof of Lemma~\ref{lem:Vinclusion}, we can use contractibility to show that $\OPT_G(s,v)\leq \OPT_{G/e}(s,v)+1$ and $\OPT_G(v,t)\leq \OPT_{G/e}(v,t)+1$.
Next, we discuss required changes to the proof of Lemma~\ref{lem:inactivetoitself}.
For a (directed) edge $e$, let $\OPT_G(s,e)$ (resp., $\OPT_G(e,t)$) denote the length of the shortest path from $s$ to the tail of $e$ (resp., from the head of $e$ to $t$) 
In the calculation of~\eqref{eq:breakdirected1}, we can replace the first inequality with $\OPT_G(s,v)+\OPT_G(v,t)\geq \OPT_G(s,t,e)$ for $v\in e$, which holds because contractibility ensures that a path that passes through $v$ also passes through $e$. 
Moreover, in the calculation of~\eqref{eq:breakdirected2}, we can replace the first equality with $\OPT_G(s,e)+\OPT_G(e,t)=\OPT_G(s,t,e)+1$.
To see this, note that $\OPT_{G/e}(s,v_e)=\OPT_{G}(s,e)$ holds because contractibility ensures that any path from $s$ to the head of $e$ passes through $e$, and similarly, $\OPT_{G/e}(v_e,t)=\OPT_{G}(e,t)$ holds, implying the equality.




\subsection{Algorithm Description}\label{subsec:lipsp-algorithm}

\begin{algorithm}[t!]
\caption{Lipschitz continuous algorithm for the shortest path problem}\label{alg:splip}
\Procedure{\emph{\Call{LipSP}{$G=(V,E),w,s,t,\epsilon$}}}{
    \If{$\OPT_{G,w}(s,t)=0$}{
        \Return an arbitrary $s$-$t$ path with length $0$.
    }
    Sample $\bm{b}$ uniformly from $\left[\frac{\epsilon\cdot \OPT_{G,w}(s,t)}{12|V|}, \frac{\epsilon\cdot \OPT_{G,w}(s,t)}{6|V|}\right]$\;\label{line:lipsp_sampleb}
    Let $\widehat{\bm{V}}=V, \widehat{\bm{E}}=\emptyset$\;
    \For{$e=(u,v)\in E$}{
        Let $\bm{l}(e)=\floor{\frac{w(e)}{\bm{b}}}$\;
        Sample $\bm{x}(e)$ uniformly from $[0,1)$\;\label{line:lipsp_samplex}
        Let $\widehat{\bm{w}}(e)=\bm{l}(e)+2$ if $\bm{x}(e) \leq \frac{(\bm{l}(e)+1)\bm{b}-\bm{w}(e)}{\bm{b}}$ and $\widehat{\bm{w}}(e)=\bm{l}(e)+3$ otherwise\;\label{line:lipsp_choosew}
        \If{$\widehat{\bm{w}}(e)\leq 12\epsilon^{-1}|V|+3$}{\label{line:omitlong}
            $\widehat{\bm{V}}\leftarrow \widehat{\bm{V}}\cup \{e_{1},\dots, e_{\widehat{\bm{w}}(e)-1}\}\cup \{e'_{1},\dots, e'_{\widehat{\bm{w}}(e)-1}\}$\;
            $\widehat{\bm{E}}\leftarrow \widehat{\bm{E}}\cup \{(u=e_0,e_1),\dots,(e_{\widehat{\bm{w}}(e)-1},e_{\widehat{\bm{w}}(e)}=v)\}\cup \{(v=e'_{\widehat{\bm{w}}(e)},e'_{\widehat{\bm{w}}(e)-1}),\dots, (e_{1},e_{0}=u)\}$.
        }
    }
    \Return the corresponding path in $G$ of \Call{DiSP}{$\widehat{\bm{G}}=(\widehat{\bm{V}},\widehat{\bm{E}}),s,t,\epsilon/4$}.
}
\end{algorithm}

Our algorithm, \Call{LipSP}{}, is presented in Algorithm~\ref{alg:splip}, where \Call{DiSP}{} is the shortest path algorithm for directed graphs presented in Section~\ref{subsec:sp-directed}. 
This algorithm first transforms the given weighted graph $(G=(V,E),w)$ into a (random and unweighted) directed graph $\widehat{\bm{G}}$ and then applies \Call{DiSP}{} to it.
To construct $\widehat{\bm{G}}$, we first sample a discretization parameter $\bm{b}$ and then round the length of each edge $e \in E$ to an integer $\widehat{\bm{w}}(e)$ that is roughly $\bm{l}(e) := \floor{\frac{w(e)}{\bm{b}}}$ (we  discuss the exact choice later). 
Then, $\widehat{\bm{G}}$ is constructed by replacing each edge $e=\{u,v\}$ in $G$ with two directed paths of lengths $\widehat{\bm{w}}(e)$, one is from $u$ to $v$, and the other is from $v$ to $u$. 

We construct a directed graph instead of an undirected graph because we want to transform the output of \Call{DiSP}{} back to an $s$-$t$ walk on $G$. 
Because \Call{DiSP}{} outputs a walk (not necessarily a path), if $\widehat{\bm{G}}$ is undirected, then the corresponding ``walk'' in $G$ may turn around in the middle of an edge.

To attain Lipschitz continuity, we need some more tricks. 
First, we ignore lengthy edges in $G$ (Line~\ref{line:omitlong}). 
By doing this, we can bound the number of vertices in $\widehat{\bm{G}}$, and  keep the contraction sensitivity of \Call{DiSP}{$\widehat{\bm{G}},s,t,\epsilon/4$} low.  
This also makes the algorithm run in strongly polynomial time. 
Second, we choose $\widehat{\bm{w}}(e)$ from $\bm{l}(e)+2$ or $\bm{l}(e)+3$, not $\bm{l}(e)$ or $\bm{l}(e)+1$. 
Thus, we can ensure that reducing the weight of an edge in $G$ always corresponds to contracting a contractible edge in $\widehat{\bm{G}}$.


\subsection{Approximation Ratio}\label{subsec:lipsp-approximation}

We analyze the approximation ratio of Algorithm~\ref{alg:splip}.
We denote the minimum length of a directed $s$-$t$ path in $(G,w)$ as $\OPT_{G,w}(s,t)$.
\begin{lemma}\label{lem:splipapp}
The approximation ratio of \Call{LipSP}{} is $1+\epsilon$. 
\end{lemma}
\begin{proof}
Let $P$ be a path that attains $\OPT_{G,w}(s,t)$.
Note that $P$ does not include any edge $e$ with $\widehat{\bm{w}}(e) > 12\epsilon^{-1}|V|+3$, because we have $w(e)\geq (\widehat{\bm{w}}(e)-3)\bm{b} > 12\epsilon^{-1}|V|\bm{b}\geq \OPT_{G,w}(s,t)$.
Therefore, we can take a path $P'$ in $\widehat{\bm{G}}$ corresponding to $P$.
Then, we have
\begin{align*}
    \Call{LipSP}{G,w,s,t,\epsilon}&\leq \left(1+\frac{\epsilon}{4}\right)\OPT_{\widehat{\bm{G}}}(s,t)\bm{b} \\
    &\leq \left(1+\frac{\epsilon}{4}\right)\left|E(P')\right|\bm{b}\\
    &\leq \left(1+\frac{\epsilon}{4}\right)\sum_{e\in E(P)}\left(\floor{\frac{w(e)}{\bm{b}}}+3\right)\bm{b} \\
    &\leq \left(1+\frac{\epsilon}{4}\right)\left(\sum_{e\in E(P)}w(e)+3\bm{b}\left|E(P)\right|\right)\\
    &= \left(1+\frac{\epsilon}{4}\right)\left(\OPT_{G,w}(s,t)+3\bm{b}\left|E(P)\right|\right)\\
    &\leq \left(1+\frac{\epsilon}{4}\right)\left(\OPT_{G,w}(s,t)+\frac{\epsilon|E(P)|\OPT_{G,w}(s,t)}{2|V|}\right)\\
    &\leq \left(1+\frac{\epsilon}{4}\right)\left(1+\frac{\epsilon}{2}\right)\OPT_{G,w}(s,t)\\
    &\leq (1+\epsilon)\OPT_{G,w}(s,t),
\end{align*}
where the first inequality is from Lemma~\ref{lem:spapp} and the fact that $w(e)\leq \widehat{\bm{w}}(e)\bm{b}$ holds for all $e\in E$, the second inequality is from the fact that each edge $e$ in $G$ corresponds to a path in $\widehat{\bm{G}}$ with length at most $\floor{\frac{w(e)}{\bm{b}}}+3$, and the last inequality is from $\epsilon\leq 1$.
\end{proof}

\subsection{Lipschitz Continuity}\label{subsec:lipsp-lipschitz}

Next, we analyze the Lipschitz continuity of Algorithm~\ref{alg:splip}.
Let $\delta>0$ and $f\in E$. 
By Lemma~\ref{lem:seeoneelement}, it suffices to bound 
\begin{align*}
    \frac{1}{\delta}\EMW\left(\Call{LipSP}{G,w,s,t,\epsilon},\Call{LipSP}{G,w+\delta \mathbf{1}_f,s,t,\epsilon}\right).
\end{align*}
for sufficiently small $\delta>0$.
Specifically, we assume $\delta\leq \frac{\epsilon\cdot \OPT_{G,w}(s,t)}{12|V|}$.
To distinguish the variables used in \Call{LipSP}{$G,w,s,t,\epsilon$} from those used in \Call{LipSP}{$G,w+\delta \mathbf{1}_f,s,t,\epsilon$}, we add superscripts $\delta f$ to the latter, e.g., $\bm{b}^{\delta f}$.

We consider the coupling $\mathcal{D}$ between $(\bm{b},\bm{x})$ and $(\bm{b}^{\delta f},\bm{x}^{\delta f})$ defined as follows: For each parameter $b \in \mathbb{R}_{\geq 0}$, we transport the probability mass for $\bm{b}=b$ to that for $\bm{b}^{\delta f}=b$ as far as possible in such a way that for each choice of $x \in [0,1)$, the probability mass for $\bm{x}=x$ is transported to that for
\begin{align}
    \bm{x}^{\delta f}(e)=
    \begin{cases}
        x(e) & \text{if }e\in E\setminus \{f\}, \\
        x(e)-\frac{\delta}{b} & \text{if }e=f, x(f) > \frac{\delta}{b}, \\
        x(e)-\frac{\delta}{b} + 1 & \text{if }e=f, x(f) \leq \frac{\delta}{b}. \\
    \end{cases}
    \label{eq:mapping-of-x}
\end{align}
The remaining mass is transported arbitrarily.
Clearly,~\eqref{eq:mapping-of-x} is a one-to-one correspondence from $[0,1)$ to itself for all $e\in E$. 
Therefore, any pair $\left((b,x),(b^{\delta f},x^{\delta f})\right)$ in the support of the remaining mass satisfies $b\neq b^{\delta f}$.

We begin by presenting the following useful property of the coupling.
\begin{lemma}\label{lem:xetransport}
Let $\left((\bm{b},\bm{x}),(\bm{b}^{\delta f},\bm{x}^{\delta f})\right)$ be a pair sampled from $\mathcal{D}$ conditioned on $\bm{b}=\bm{b}^{\delta f}$.
Then we have $\widehat{\bm{w}}(e) = \widehat{\bm{w}}^{\delta f}(e) - 1$ if $e=f$ and $\bm{x}(e)\leq \frac{\delta}{\bm{b}}$ hold, and we have $\widehat{\bm{w}}(e) = \widehat{\bm{w}}^{\delta f}(e)$ otherwise.
\end{lemma}
We defer the proof to Section~\ref{subsubsec:xetransport} and continue the analysis of the Lipschitz continuity of \Call{LipSP}{}.

When $\widehat{\bm{w}}(e)=\widehat{\bm{w}}^{\delta f}(e)$ holds, $G$ and $G^{\delta f}$ are exactly the same and so are \Call{LipSP}{$G,w,s,t,\epsilon$} and \Call{LipSP}{$G,w+\delta \mathbf{1}_f,s,t,\epsilon$}.
The case of $\widehat{\bm{w}}(e)=\widehat{\bm{w}}^{\delta f}(e)-1$ can be divided into the following three cases:
\begin{enumerate}
    \item[(i)] $\bm{b}\neq \bm{b}^{\delta f}$;
    \item[(ii)] $\bm{b} = \bm{b}^{\delta f}$, and the two paths corresponding to $f$ exist in $\widehat{\bm{G}}$ but do not exist in $\widehat{\bm{G}}^{\delta f}$; and
    \item[(iii)] $\bm{b} = \bm{b}^{\delta f}$, and the two paths corresponding to $f$ exist in both $\widehat{\bm{G}}$ and $\widehat{\bm{G}}^{\delta f}$, but the length of the paths in $\widehat{\bm{G}}$ is smaller than that in $\widehat{\bm{G}}^{\delta f}$ by $1$. 
\end{enumerate}

Let $p_{\neq}$, $p_{\mathrm{path}}$, and $p_{\mathrm{len}}$ be the probabilities that Cases (i), (ii), and (iii) occur, respectively, in the coupling $\mathcal{D}$. 
We will bound these probabilities.
First, we consider Case~(i).
\begin{lemma}\label{lem:splip_sampleb}
We have
\begin{align*}
    p_{\neq}\leq \frac{2\delta}{\OPT_{G,w}(s,t)}.
\end{align*}
\end{lemma}
\begin{proof}
We have
\begin{align*}
    p_{\neq}=\TV\left(\bm{b},\bm{b}^{\delta f}\right)&\leq \Pr\left[\bm{b}^{\delta f}\in \left[\frac{\epsilon\cdot \OPT_{G,w}(s,t)}{6|V|}, \frac{\epsilon\cdot \OPT_{G,w+\delta \mathbf{1}_f}(s,t)}{6|V|}\right]\right]\\
    &\leq \frac{\delta}{6|V|}\cdot \frac{12|V|}{\OPT_{G,w+\delta \mathbf{1}_f}(s,t)}
    = \frac{2\delta}{\OPT_{G,w+\delta \mathbf{1}_f}(s,t)}
    \leq \frac{2\delta}{\OPT_{G,w(s,t)}},
\end{align*}
where the first inequality is from $\delta > 0$ and hence the interval of $\bm{b}$ is to the left of that of $\bm{b}^{\delta f}$, and the second inequality is from $\OPT_{G,w+\delta \mathbf{1}_f}(s,t)\leq \OPT_{G,w}(s,t)+\delta$.
\end{proof}

Next we consider Case~(ii).
\begin{lemma}
We have
\[
    p_{\mathrm{path}}\leq \frac{6\delta}{ \OPT_{G,w}(s,t)}.
\]
\end{lemma}
\begin{proof}
Assume $b = b^{\delta f}$.
Case (ii) happens only when $\widehat{\bm{w}}(f) = \floor{12\epsilon^{-1}|V|}+3$ and $\widehat{\bm{w}}^{\delta f}(f) = \floor{12\epsilon^{-1}|V|}+4$. 
This implies $\bm{l}(f)$ is either $\floor{12\epsilon^{-1}|V|}$ or $\floor{12\epsilon^{-1}|V|}+1$, and it follows that
\begin{align*}
    \frac{w(f)}{\bm{b}}-1& \leq \floor{\frac{w(f)}{\bm{b}}}=\bm{l}(f)\leq \floor{\frac{12|V|}{\epsilon}}+1\leq \frac{12|V|}{\epsilon}+1,\\
    \frac{w(f)}{\bm{b}} & \geq \floor{\frac{w(f)}{\bm{b}}}=\bm{l}(f)\geq \floor{\frac{12|V|}{\epsilon}}\geq \frac{12|V|}{\epsilon}-1.
\end{align*}
Equivalently, we have
\begin{align*}
    \bm{b}\in \left[\frac{w(f)}{12\epsilon^{-1}|V|+2}, \frac{w(f)}{12\epsilon^{-1}|V|-1}\right].
\end{align*}
Therefore, we have
\begin{align*}
    p_{\mathrm{path}}
    &\leq \Pr\left[\bm{b}=\bm{b}^{\delta f}\in \left[\frac{w(f)}{12\epsilon^{-1}|V|+2}, \frac{w(f)}{12\epsilon^{-1}|V|-1}\right]\right]\cdot \Pr\left[\bm{x}(f)\leq \frac{\delta}{\bm{b}}\right]\\
    &\leq \Pr\left[\bm{b}\in \left[\frac{w(f)}{12\epsilon^{-1}|V|+2}, \frac{w(f)}{12\epsilon^{-1}|V|-1}\right]\right]\cdot \Pr\left[\bm{x}(f)\leq \frac{\delta}{\bm{b}}\right] \\ 
    &\leq \frac{6}{12\epsilon^{-1}|V|+2}\cdot \Pr\left[\bm{x}(f)\leq \frac{\delta}{\bm{b}}\right] \\
    &= \frac{6}{12\epsilon^{-1}|V|+2}\cdot \frac{\delta\cdot 12|V|}{\epsilon\cdot \OPT_{G,w}(s,t)}\\
    &\leq \frac{6}{12\epsilon^{-1}|V|}\cdot \frac{\delta\cdot 12|V|}{\epsilon\cdot \OPT_{G,w}(s,t)}
    \leq \frac{6\delta}{\OPT_{G,w}(s,t)},
\end{align*}
where the second inequality is from
\begin{align*}
    &\Pr\left[\bm{b}\in \left[\frac{w(f)}{12\epsilon^{-1}|V|+2}, \frac{w(f)}{12\epsilon^{-1}|V|-1}\right]\right]\\
    &=\frac{12\epsilon^{-1}|V|}{\OPT_{G,w}(s,t)}\cdot \max\left(0,\min\left(\frac{w(f)}{12\epsilon^{-1}|V|-1},\frac{\OPT_{G,w}(s,t)}{6\epsilon^{-1}|V|}\right)-\max\left(\frac{w(f)}{12\epsilon^{-1}|V|+2},\frac{\OPT_{G,w}(s,t)}{12\epsilon^{-1}|V|}\right)\right)\\
    &\leq \frac{12\epsilon^{-1}|V|}{\OPT_{G,w}(s,t)}\cdot \max\left(0,\min\left(\frac{w(f)}{12\epsilon^{-1}|V|-1},\frac{\OPT_{G,w}(s,t)}{6\epsilon^{-1}|V|}\right)-\frac{w(f)}{12\epsilon^{-1}|V|+2}\right)\\
    &\leq \frac{12\epsilon^{-1}|V|}{\OPT_{G,w}(s,t)}\cdot \max\left(0,\min\left(\frac{w(f)}{12\epsilon^{-1}|V|-1},\frac{\OPT_{G,w}(s,t)}{6\epsilon^{-1}|V|}\right)\cdot \frac{12\epsilon^{-1}|V|-1}{12\epsilon^{-1}|V|+2}\right)\\
    &\leq \frac{12\epsilon^{-1}|V|}{\OPT_{G,w}(s,t)}\cdot \max\left(0,\frac{\OPT_{G,w}(s,t)}{6\epsilon^{-1}|V|}\cdot \frac{\left(12\epsilon^{-1}|V|+2\right)-\left(12\epsilon^{-1}|V|-1\right)}{12\epsilon^{-1}|V|+2}\right)\\
    &= \frac{12\epsilon^{-1}|V|}{\OPT_{G,w}(s,t)}\cdot \frac{\OPT_{G,w}(s,t)}{6\epsilon^{-1}|V|}\cdot \frac{3}{12\epsilon^{-1}|V|+2}
    = \frac{6}{12\epsilon^{-1}|V|+2}.
    \qedhere
\end{align*}
\end{proof}

Finally, we consider Case~(iii).
\begin{lemma}\label{lem:case-3}
For $b\in \mathbb{R}_{\geq 0}$, let $\mathrm{E}_b$ be the event where Case~(iii) occurs with $\bm{b}=\bm{b}^{\delta f}=b$.
Then, we have
\begin{align*}
    &\E_{((\bm{b},\bm{x}),(\bm{b}^{\delta f},\bm{x}^{\delta f}))\sim \mathcal{D}}\left[\EMW\left(\Call{LipSP}{G,w,s,t,\epsilon}\right),\left(\Call{LipSP}{G,w+\delta \mathbf{1}_f,s,t,\epsilon}\right)\mid 
    \mathrm{E}_b\right]\\
    &\leq O\left(\frac{\log^3 (|V||E|)}{\epsilon}\right)\cdot b+\delta.
\end{align*}
\end{lemma}
\begin{proof}
Let $\bm{Q}_1,\bm{Q}_2$ be the two paths in $\widehat{\bm{G}}^{\delta f}$ corresponding to $f$. 
Because $\widehat{w}(f)\geq 2$ (and thus $\widehat{w}^{\delta f}(f)\geq 3$), we can take edges $\bm{e}_1, \bm{e}_2$ on $\bm{Q}_1, \bm{Q}_2$, respectively, such that both of the endpoints have indegree and outdegree $1$, i.e., we can assume that $\bm{e}_1$ and $\bm{e}_2$ are contractible. 
Therefore by Lemma~\ref{lem:directed}, we have
\begin{align*}
    \EMU\left(\Call{DiSP}{\widehat{\bm{G}}^{\delta f},s,t,\epsilon/4},\Call{DiSP}{\widehat{\bm{G}}^{\delta f}/e_1,s,t,\epsilon/4}\right)\leq O\left(\frac{\log^3 |\widehat{\bm{V}}|}{\epsilon}\right).
\end{align*}
Clearly, $\bm{e}_2$ is still contractible in $\widehat{\bm{G}}/\bm{e}_1$. 
Thus, we have
\begin{align*}
    &\EMU\left(\Call{DiSP}{\widehat{\bm{G}}^{\delta f},s,t,\epsilon/4},\Call{DiSP}{\widehat{\bm{G}},s,t,\epsilon/4}\right)\\
    &=\EMU\left(\Call{DiSP}{\widehat{\bm{G}}^{\delta f},s,t,\epsilon/4},\Call{DiSP}{(\widehat{\bm{G}}^{\delta f}/e_1)/e_2,s,t,\epsilon/4}\right)\\
    &\leq \EMU\left(\Call{DiSP}{\widehat{\bm{G}}^{\delta f},s,t,\epsilon/4},\Call{DiSP}{\widehat{\bm{G}}^{\delta f}/e_1,s,t,\epsilon/4}\right)\\
    &\quad +\EMU\left(\Call{DiSP}{\widehat{\bm{G}}^{\delta f}/e_1,s,t,\epsilon/4},\Call{DiSP}{(\widehat{\bm{G}}^{\delta f}/e_1)/e_2,s,t,\epsilon/4}\right)\\
    &\leq O\left(\frac{\log^3 |\widehat{\bm{V}}|}{\epsilon}\right).
\end{align*}
Let $b$ and $x,x^{\delta f}\in [0,1)^{E}$ be such that $((b,x),(b,x^{\delta f}))$ is in the support of $\mathcal{D}$ and Case~(iii) occurs when $\bm{b}=\bm{b}^{\delta f}=b, \bm{x}=x, \bm{x}^{\delta f}=x^{\delta f}$.
Let $\widehat{G}$ (resp., $\widehat{G}^{\delta f}$) be the graph constructed from $(G,w)$ (resp., $(G,w+\delta \bm{1}_f)$) using $(b,x)$ (resp., $(b,x^{\delta f})$).
Let $(\widehat{\bm{P}}^{\delta f},\widehat{\bm{P}})$ be a pair sampled from the coupling $\mathcal{D}_{\widehat{G}^{\delta f}, \widehat{G}}$ that attains $\EMU\left(\Call{DiSP}{\widehat{G}^{\delta f},s,t,\epsilon/4},\Call{DiSP}{\widehat{G},s,t,\epsilon/4}\right)$. 
Let $\bm{P}$ and $\bm{P}^{\delta f}$ be the walks in $G$ corresponding to $\widehat{\bm{P}}$ and $\widehat{\bm{P}}^{\delta f}$, respectively.
Then, we have
\begin{align*}
    & \left\|\sum_{e\in E(\bm{P})}w(e)\bm{1}_e-\sum_{e\in E(\bm{P}^{\delta f})}(w(e)+\delta \bm{1}_{e=f})\bm{1}_e\right\|_1
    \leq \left\|\sum_{e\in E(\bm{P})}w(e)\bm{1}_e-\sum_{e\in E(\bm{P}^{\delta f})}w(e)\bm{1}_e\right\|_1 + \delta\\
    &= \sum_{e\in E(\bm{P})\triangle E(\bm{P}^{\delta f})}w(e) + \delta
    \leq \sum_{e\in E(\bm{P})\triangle E(\bm{P}^{\delta f})}(l(e)+1)b + \delta \\
    &\leq \sum_{e\in E(\bm{P})\triangle E(\bm{P}^{\delta f})}\widehat{w}(e)b + \delta
    = \left|E(\widehat{\bm{P}})\triangle E(\widehat{\bm{P}}^{\delta f})\right|\cdot b + \delta. 
\end{align*}
Therefore, we have
\begin{align*}
    &\E_{((\bm{b},\bm{x}),(\bm{b}^{\delta f},\bm{x}^{\delta f}))\sim \mathcal{D}}\left[\EMW\left(\Call{LipSP}{G,w,s,t,\epsilon}\right),\left(\Call{LipSP}{G,w+\delta \mathbf{1}_f,s,t,\epsilon}\right)\mid \mathrm{E}_b\right]\\
    &\leq \E_{\substack{((\bm{b},\bm{x}),(\bm{b}^{\delta f},\bm{x}^{\delta f}))\sim \mathcal{D},\\ (\widehat{\bm{P}}^{\delta f},\widehat{\bm{P}})\sim \mathcal{D}_{\widehat{\bm{G}}^{\delta f}, \widehat{\bm{G}}}}}\left[\left\|\sum_{e\in E(\bm{P})}w(e)\bm{1}_e-\sum_{e\in E(\bm{P}^{\delta f})}(w(e)+\delta \bm{1}_{e=f})\bm{1}_e\right\|_1\mid \mathrm{E}_b \right]\\
    &\leq \E_{\substack{((\bm{b},\bm{x}),(\bm{b}^{\delta f},\bm{x}^{\delta f}))\sim \mathcal{D},\\ (\widehat{\bm{P}}^{\delta f},\widehat{\bm{P}})\sim \mathcal{D}_{\widehat{\bm{G}}^{\delta f}, \widehat{\bm{G}}}}}\left[\left|E(\widehat{\bm{P}})\triangle E(\widehat{\bm{P}}^{\delta f})\right|\cdot b + \delta\mid \mathrm{E}_b \right]\\
    &=\left(\E_{\substack{((\bm{b},\bm{x}),(\bm{b}^{\delta f},\bm{x}^{\delta f}))\sim \mathcal{D},\\ (\widehat{\bm{P}}^{\delta f},\widehat{\bm{P}})\sim \mathcal{D}_{\widehat{\bm{G}}^{\delta f}, \widehat{\bm{G}}}}}\left[\left|E(\widehat{\bm{P}})\triangle E(\widehat{\bm{P}}^{\delta f})\right|\mid \mathrm{E}_b \right]\right)\cdot b+\delta\\
    &= \E_{((\bm{b},\bm{x}),(\bm{b}^{\delta f},\bm{x}^{\delta f}))\sim \mathcal{D}}\left[\EMU\left(\Call{DiSP}{\widehat{\bm{G}}^{\delta f},s,t,\epsilon/4},\Call{DiSP}{\widehat{\bm{G}}^{\delta f}/e_1,s,t,\epsilon/4}\right)\mid \mathrm{E}_b \right]\cdot b+\delta\\
    &\leq \E_{((\bm{b},\bm{x}),(\bm{b}^{\delta f},\bm{x}^{\delta f}))\sim \mathcal{D}}\left[O\left(\frac{\log^3 |\widehat{\bm{V}}|}{\epsilon}\right)\mid \mathrm{E}_b \right]\cdot b+\delta\\
    &\leq O\left(\frac{\log^3 (|V||E|)}{\epsilon}\right)\cdot b+\delta,
\end{align*}
where the last inequality is from $|\widehat{\bm{V}}|\leq |E|\cdot \left(12\epsilon^{-1}|V|+1\right)+|V|\leq O(\epsilon^{-1}|V||E|)$.
\end{proof}

Now, we show that Algorithm~\ref{alg:splip} is Lipschitz continuous.
\begin{lemma}\label{lem:splipsens}
We have
\begin{align*}
    \EMW\left(\left(\Call{LipSP}{G,w,s,t,\epsilon}\right),\left(\Call{LipSP}{G,w+\delta \mathbf{1}_f,s,t,\epsilon}\right)\right)\leq O\left(\frac{\delta \log^3 (|V||E|)}{\epsilon}\right).
\end{align*}
\end{lemma}
\begin{proof}
    We bound the Lipschitz constant by summing up the contributions of Events (i), (ii), and (iii). 
    We have
    \begin{align*}
        &\EMW\left(\left(\Call{LipSP}{G,w,s,t,\epsilon}\right),\left(\Call{LipSP}{G,w+\delta \mathbf{1}_f,s,t,\epsilon}\right)\right)\\
        &\leq p_{\neq}\cdot (1+\epsilon)\OPT_{G,w}(s,t) + p_{\mathrm{path}}\cdot (1+\epsilon)\OPT_{G,w}(s,t) + p_{\mathrm{len}}\cdot \left(O\left(\frac{\log^3 (|V||E|)}{\epsilon}\right)\cdot \bm{b}+\delta\right)\\
        &\leq \left(\frac{2\delta}{\OPT_{G,w}(s,t)}\right)\cdot (1+\epsilon)\OPT_{G,w}(s,t) + \left(\frac{6\delta} {\OPT_{G,w}(s,t)}\right)\cdot (1+\epsilon)\OPT_{G,w}(s,t)\\
        &\quad + \left(O\left(\frac{\log^3 (|V||E|)}{\epsilon}\right)\cdot \bm{b}+\delta\right)\cdot \frac{\delta}{\bm{b}}\\
        &\leq O\left(\delta\right) + O(\delta) + O\left(\frac{\delta \log^3 (|V||E|)}{\epsilon}\right)+\frac{\delta^2}{\bm{b}}\\
        &\leq O\left(\frac{\delta \log^3 (|V||E|)}{\epsilon}\right),
    \end{align*}
    where the last inequality is from $\delta\leq \frac{\epsilon\cdot \OPT_{G,w}(s,t)}{12|V|}\leq \bm{b}$.
\end{proof}

\begin{proof}[Proof of Theorem~\ref{thm:sp}]
Combining Lemmas~\ref{lem:splipapp} and~\ref{lem:splipsens} yields the claim.
\end{proof}

\subsubsection{Proof of Lemma~\ref{lem:xetransport}}\label{subsubsec:xetransport}

To prove Lemma~\ref{lem:xetransport}, we first present several properties of the coupling $\mathcal{D}$.
\begin{lemma}
In the coupling $\mathcal{D}$, we have $\widehat{\bm{w}}(e)=\widehat{\bm{w}}^{\delta f}(e)$ for any $e\in E\setminus \{f\}$.
\end{lemma}
\begin{proof}
The claim holds because $\bm{x}(e)=\bm{x}^{\delta f}(e)$ and $\frac{(\bm{l}(e)+1)b-w(e)}{b}=\frac{(\bm{l}^{\delta f}(e)+1)b-w(e)}{b}$ holds if $e\neq f$.
\end{proof}

\begin{lemma}
In the coupling $\mathcal{D}$, we have $\widehat{\bm{w}}(f)=\widehat{\bm{w}}^{\delta f}(f)$ if $\bm{x}(f) > \frac{\delta}{b}$.
\end{lemma}
\begin{proof}
Note that $\bm{l}^{\delta f}(e)$ is equal to either $\bm{l}(e)$ or $\bm{l}(e)+1$ because $\delta\leq \frac{\epsilon\cdot \OPT_G(s,t)}{12|V|}$.

Suppose $\bm{l}^{\delta f}(f)=\bm{l}(f)$. Then, we have
\begin{align*}
    \widehat{\bm{w}}(f)=\bm{l}(f)+2
    &\iff \bm{x}(f)\leq \frac{(\bm{l}(f)+1)b-w(f)}{b}\\
    &\iff \bm{x}(f)-\frac{\delta}{b}\leq \frac{(\bm{l}(f)+1)b-w(f)}{b}-\frac{\delta}{b}\\
    &\iff \bm{x}^{\delta f}(f)\leq \frac{(\bm{l}(f)+1)b-(w(f)+\delta)}{b}\\
    &\iff \widehat{\bm{w}}^{\delta f}(f) = \bm{l}(f)+2.
\end{align*}

Suppose $\bm{l}^{\delta f}(f)=\bm{l}(f)+1$. 
Then from the definitions of $\bm{l}(f)$ and $\bm{l}^{\delta f}(f)$, we have
\begin{align*}
    \frac{w(f)}{b} < \bm{l}(f)+1 \leq \frac{w(f)+\delta}{b}.
\end{align*}
It follows that
\begin{align*}
    \bm{x}(f) > \frac{\delta}{b} = \frac{w(f)+\delta}{b}-\frac{w(f)}{b}\geq \bm{l}(f)+1-\frac{w(f)}{b} = \frac{(\bm{l}(f)+1)b-w(f)}{b},
\end{align*}
which implies $\widehat{\bm{w}}(f)=\bm{l}(f)+3$.
Moreover, we have
\begin{align*}
    \bm{x}^{\delta f}(f)=\bm{x}(f)-\frac{\delta}{b}\leq 1-\frac{\delta}{b}\leq \frac{(\bm{l}(f)+1)b-w(f)}{b}+1-\frac{\delta}{b} = \frac{(\bm{l}^{\delta f}(f)+1)b-(w(f)+\delta)}{b},
\end{align*}
where the last inequality is from $(\bm{l}(f)+1)b\geq w(f)$.
Therefore $\widehat{\bm{w}}^{\delta f}(f)=\bm{l}(f)+3=\widehat{\bm{w}}(f)$, and the claim holds.
\end{proof}

\begin{lemma}\label{lem:diff1case}
In the coupling $\mathcal{D}$, we have $\widehat{\bm{w}}(f)=\widehat{\bm{w}}^{\delta f}(f) - 1$ if $\bm{x}(f) \leq \frac{\delta}{b}$.
\end{lemma}
\begin{proof}
Suppose $\bm{l}^{\delta f}(f)=\bm{l}(f)$. Then, we have
\begin{align*}
    \bm{x}(f)\leq  \frac{\delta}{b}=\frac{w(e)+\delta}{b}-\frac{w(e)}{b} < \bm{l}(f)+1-\frac{w(e)}{b} =  \frac{(\bm{l}(f)+1)b-w(e)}{b}.
\end{align*}
Thus, $\widehat{\bm{w}}(f)=\bm{l}(f)+2$. 
Moreover, we have
\begin{align*}
    \bm{x}^{\delta f}(f) = \bm{x}(f)-\frac{\delta}{b}+1 > 1-\frac{\delta}{b}\geq \frac{\bm{l}(f)b-w(f)}{b}+1-\frac{\delta}{b} = \frac{(\bm{l}(f)+1)b-(w(f)+\delta)}{b},
\end{align*}
where the last inequality is from $\bm{l}(f)b\leq w(f)$. Therefore $\widehat{\bm{w}}^{\delta f}(f)=\bm{l}(f)+3=\widehat{\bm{w}}(f)+1$.

Suppose $\bm{l}^{\delta f}(f)=\bm{l}(f)+1$. Then, we have
\begin{align*}
    \widehat{\bm{w}}(f)=\bm{l}(f)+2
    &\iff \bm{x}(f)\leq \frac{(\bm{l}(f)+1)b-w(f)}{b}\\
    &\iff \bm{x}(f)-\frac{\delta}{b} + 1\leq \frac{(\bm{l}(f)+1)b-w(f)}{b}-\frac{\delta}{b} + 1\\
    &\iff \bm{x}^{\delta f}(f)\leq \frac{(\bm{l}(f)+2)b-(w(f)+\delta)}{b}\\
    &\iff \widehat{\bm{w}}^{\delta f}(f) = \bm{l}(f)+3.
\end{align*}
Therefore, the lemma is proved.
\end{proof}

\begin{proof}[Proof of Lemma~\ref{lem:xetransport}]
Combining the three lemmas above yields the lemma.
\end{proof}



\section{Maximum Weight Matching}\label{sec:maximum-weight-matching}
In this section, we present a Lipschitz continuous algorithm for the  maximum weight matching problem and prove the following:
\begin{theorem}\label{thm:matching}
    For any $\epsilon\in (0,1/8)$, there exists a polynomial-time $(1/8-\epsilon)$-approximation algorithm with Lipschitz constant $O(\epsilon^{-1})$.
\end{theorem}

Let $\alpha>2$. 
Our algorithm is almost identical to the one used in the algorithm for the maximum weight matching problem with a sensitivity bound~\cite{Yoshida2021}.
We classify edges according to their weights using the geometric sequence $\bm{b} \cdot \alpha^i$, where $\bm{b}$ is sampled uniformly from $[1,\alpha]$, and then iteratively apply the randomized greedy to the edge sets from the heaviest one to the lightest one to construct the output matching.
The only difference from the algorithm used in~\cite{Yoshida2021} is that we use $\bm{b}$ to perturb the thresholds.

\begin{algorithm}[t!]
\caption{Lipschitz continuous algorithm for the maximum matching problem}\label{alg:match_scale}
\Procedure{\emph{\Call{LipMWM}{$G=(V,E), w,\alpha$}}}{
    Sample $\bm{b}$ uniformly from $[1,\alpha]$\;\label{line:MWM_sampleb}
    Let $\bm{\pi}$ be a random permutation of $V$\;\label{line:MWM_samplepi}
    \For{$i\in \mathbb{Z}$}{
        Let $\bm{E}_i$ be the set of edges $e \in E$ with $w(e)\geq \bm{b}\cdot \alpha^{i}$.
    }
    $\bm{M}\leftarrow \emptyset$\;
    \For{$i$ with $\bm{E}_{i}\setminus \bm{E}_{i+1}\neq \emptyset$ in decreasing order}{
        \For{$e\in \bm{E}_i$ in the order of $\bm{\pi}$}{
            \If{$e$ is not incident to $\bm{M}$}{
                Add $e$ to $\bm{M}$.
            }
        }
    }
    \Return $\bm{M}$.
}
\end{algorithm}

The approximability analysis is identical to that presented in~\cite{Yoshida2021}.
\begin{lemma}\label{lem:match_approx}
The approximation ratio of Algorithm~\ref{alg:match_scale} is (at least) $\frac{1}{4\alpha}$.
\end{lemma}

Now, we analyze the Lipschitz continuity. 
By Lemma~\ref{lem:seeoneelement}, it suffices to bound
\begin{align*}
    \EMW\left((\Call{LipMWM}{G,w,\alpha},w), (\Call{LipMWM}{G,w-\delta \bm{1}_f, \alpha},w-\delta \bm{1}_f)\right).
\end{align*}
We use $\bm{b}^{\delta f}$ and $\bm{M}^{\delta f}$ to denote the random variables $\bm{b}$ and $\bm{M}$, respectively, used in $\Call{LipMWM}{G,w-\delta \bm{1}_f, \alpha}$.
Let $\bm{i}_{\bm{b}}$ be the (unique) integer such that $\bm{b}\cdot \alpha^{\bm{i}_{\bm{b}}}\leq w(f) < \bm{b}\cdot \alpha^{{\bm{i}_{\bm{b}}}+1}$.
The following is immediate from the algorithm.
\begin{lemma}\label{lem:sameset}
Let $b\in [1,\alpha]$.
If $b\cdot \alpha^{i_b}\leq w(f)-\delta$, then we have
\begin{align*}
    \EMW\left((\left(\Call{LipMWM}{G,w,\alpha}\mid \bm{b}=b\right),w), \left(\left(\Call{LipMWM}{G,w-\delta \bm{1}_f, \alpha}\mid \bm{b}^{\delta f}=b\right),w-\delta \bm{1}_f\right)\right)=0.
\end{align*}
\end{lemma}
Let us consider the case where the condition of Lemma~\ref{lem:sameset} does not hold.  Then, we have the following.
We omit the proof because a similar argument was presented in~\cite{Yoshida2021}.
\begin{lemma}\label{lem:matching-coupling}
Let $b\in [1,\alpha]$. 
Suppose that $\bm{b}=\bm{b}^{\delta f}=b$ and $w(f)-\delta< b\cdot \alpha^{i_b}$. 
Then, for any $i>i_b$, the distributions of $E_i\cap \bm{M}_i$ and $E_i^{\delta f}\cap \bm{M}^{\delta f}$ are exactly the same, where we write $E_i$ and $E_i^{\delta f}$ instead of $\bm{E}_i$ and $\bm{E}_i^{\delta f}$, respectively, because they are deterministically constructed from $b$.
Additionally, there is a coupling $\mathcal{M}$ between $\bm{M}$ and $\bm{M}^{\delta f}$ such that 
\begin{align*}
    \E_{(M,M^{\delta f})\sim \mathcal{M}}\left[\left|(E_i\cap M)\triangle (E_i^{\delta f}\cap M^{\delta f})\right|\right]\leq 2^{i_b-i+1}
\end{align*}
holds for all $i\leq i_b$.
\end{lemma}

Now we have the following.
\begin{lemma}\label{lem:lipschitzness-matching-same-b}
Let $b\in [1,\alpha]$ satisfy $w(f)-\delta< b\cdot \alpha^{i_b}$. Then, we have
\begin{align*}
    \EMW\left(\left(\Call{LipMWM}{G,\alpha}\mid \bm{b}=b\right), \left(\Call{LipMWM}{G-\delta f, \alpha}\mid \bm{b}^{\delta f}=b\right)\right)\leq \frac{8\alpha}{\alpha-2} \cdot w(f).
\end{align*}
\end{lemma}
\begin{proof}
We have
\begin{align*}
    &\EMW\left(\left(\Call{LipMWM}{G,\alpha}\mid \bm{b}=b\right), \left(\Call{LipMWM}{G-\delta f, \alpha}\mid \bm{b}^{\delta f}=b\right)\right)\\
    & \leq \sum_{i=-\infty}^{i_b}\left(\min_{\mathcal{M}}\E_{(M,M^{\delta f})\sim \mathcal{M}}\left[\left|\left((E_i\cap M)\triangle (E_i^{\delta f}\cap M^{\delta f})\right)\setminus \left((E_{i+1}\cap M)\triangle (E_{i+1}^{\delta f}\cap M^{\delta f})\right)\right|\right]\cdot b\cdot \alpha^{i+1}\right) \\
    & \leq \sum_{i=-\infty}^{i_b}\left(\min_{\mathcal{M}}\E_{(M,M^{\delta f})\sim \mathcal{M}}\left[\left|(E_i\cap M)\triangle (E_i^{\delta f}\cap M^{\delta f})\right|\right]\cdot b\cdot \alpha^{i+1}\right) \\
    &\leq \sum_{i=-\infty}^{i_b}\left(2^{i_b-i+1}\cdot b\cdot \alpha^{i+1}\right) \tag{by Lemmas~\ref{lem:sameset} and~\ref{lem:matching-coupling}}\\
    &= b\cdot 2^{i_b+2}\cdot \frac{1}{1-\frac{2}{\alpha}}\cdot \left(\frac{\alpha}{2}\right)^{i_b}\\
    &\leq 4b\cdot \frac{1}{1-\frac{2}{\alpha}}\cdot \alpha^{i_b}
    \leq \frac{8\alpha}{\alpha-2} \cdot w(f).
    \qedhere
\end{align*}
\end{proof}

Now, we have the following bound on the Lipschitz constant.
\begin{lemma}\label{lem:match_lip}
We have
\begin{align*}
    \EMW\left((\Call{LipMWM}{G,w,\alpha},w), (\Call{LipMWM}{G,w-\delta \bm{1}_f, \alpha},w-\delta \bm{1}_f)\right)\leq \frac{12\alpha^3}{\alpha-2}\cdot \delta.
\end{align*}
\end{lemma}
\begin{proof}
First, we evaluate the probability that the condition of Lemma~\ref{lem:sameset} does not hold. 
Let $k=\floor{\log_{\alpha} w(f)}-1$. Note that $\bm{i}_{\bm{b}}$ is either $k$ or $k+1$, because
\begin{align*}
    k=\floor{\log_{\alpha} w(f)}-1 = \floor{\log_{\alpha}\frac{w(f)}{\alpha}}\leq \floor{\log_{\alpha}\frac{w(f)}{\bm{b}}} = \bm{i}_{\bm{b}} \leq \floor{\log_{\alpha}w(f)} = k+1,
\end{align*}
where the inequalities are from $\bm{b}\in [1,\alpha]$.
Now, we have
\begin{align*}
    \Pr\left[w(f)-\delta < \bm{b}\cdot \alpha^{\bm{i}_{\bm{b}}}\right]
    &= \Pr\left[\bm{b}\in \left[\frac{w(f)-\delta}{\alpha^{\bm{i}_{\bm{b}}}}, \frac{w(f)}{\alpha^{\bm{i}_{\bm{b}}}}\right]\right]\\
    &= \sum_{i=k}^{k+1}\Pr\left[\bm{b}\in \left[\frac{w(f)-\delta}{\alpha^{i}}, \frac{w(f)}{\alpha^{i}}\right]\land \bm{i}_{\bm{b}}=i\right]\\
    &\leq \sum_{i=k}^{k+1}\Pr\left[\bm{b}\in \left[\frac{w(f)-\delta}{\alpha^{i}}, \frac{w(f)}{\alpha^{i}}\right]\right]\\
    &\leq \frac{1}{\alpha - 1}\cdot \frac{\delta}{\alpha^{k}} + \frac{1}{\alpha - 1}\cdot \frac{\delta}{\alpha^{k+1}}\\
    &=\frac{1}{\alpha - 1}\cdot \left(1+\frac{1}{\alpha}\right)\cdot \frac{\delta}{\alpha^{\floor{\log_{\alpha} w(f)}-1}}\cdot 
    \leq \frac{3}{2}\cdot \frac{\delta}{\alpha^{-2}\cdot w(f)} 
    = \frac{3}{2}\cdot \frac{\delta\cdot \alpha^2}{w(f)}.
\end{align*}
Thus, we have
\begin{align*}
    &\EMW\left((\Call{LipMWM}{G,w,\alpha},w), (\Call{LipMWM}{G, w-\delta \mathbf{1}_f, \alpha},w-\delta \mathbf{1}_f)\right)\\
    &\leq  \Pr\left[w(f)-\delta < \bm{b}\cdot \alpha^{\bm{i}_{\bm{b}}}\right]\cdot \frac{8\alpha}{\alpha-2} \cdot w(f) \tag{by Lemma~\ref{lem:lipschitzness-matching-same-b}}\\
    &= \frac{3}{2}\cdot \frac{\delta\cdot \alpha^2}{w(f)}\cdot \frac{8\alpha}{\alpha-2} \cdot w(f)
    \leq \frac{12\alpha^3}{\alpha-2}\cdot \delta.
    \qedhere
\end{align*}
\end{proof}

\begin{proof}[Proof of Theorem~\ref{thm:matching}]
Set $\alpha=2+\epsilon$. Then by Lemma~\ref{lem:match_approx}, the approximation ratio is at least 
\begin{align*}
    \frac{1}{4\alpha}=\frac{1}{8+4\epsilon}=\frac{1}{8}-\frac{4\epsilon}{8(8+4\epsilon)}\geq \frac{1}{8}-\frac{4\epsilon}{64}\geq \frac{1}{8}-\epsilon.
\end{align*}
From Lemma~\ref{lem:match_lip}, the Lipschitz constant is at most
\begin{align*}
    \frac{12\alpha^3}{\alpha-2} = \frac{12\cdot (2+\epsilon)^3}{\epsilon}\leq \frac{12\cdot \left(2+\frac{1}{8}\right)^3}{\epsilon}\leq O(\epsilon^{-1}),
\end{align*}
where the first inequality is from $\epsilon\leq \frac{1}{8}$.
\end{proof}
\section{Lower Bounds}\label{sec:lower-bounds}
In this section, we provide several lower bounds on Lipschitzness.

\begin{proof}[Proof of Theorem~\ref{thm:sp-deterministic-lower-bound}]
    Consider a weighted graph $G$ constructed by introducing two vertices $s$ and $t$ and connecting them via two edges $e_1$ and $e_2$ having weights $0$ and $1$, respectively.

    Let $\mathcal{A}$ be an arbitrarily deterministic algorithm for the $s$-$t$ shortest path problem with a finite approximation ratio.
    Then, $\mathcal{A}$ must output $e_1$.\footnote{$\mathcal{A}$ may output a walk passing through $e_1$ several times, but we can show a (larger) lower bound for this case similarly.}

    Consider a sequence of weighted graphs $G_0 = G, G_1, \ldots, G_k$, where $G_k$ is constructed from $G_{i-1}$ by decreasing and increasing the weights of $e_1$ and $e_2$, respectively, by $1/k$.
    Note that in $G_k$ the weights of $e_1$ and $e_2$ are $1$ and $0$, respectively, and hence the algorithm $\mathcal{A}$ must output $e_2$ to achieve a finite approximation ratio.
    This implies that there exists some $i$ such that $\mathcal{A}$ outputs $e_1$ and $e_2$ on $G_{i-1}$ and $G_i$, respectively.
    Hence, the Lipschitz constant of  $\mathcal{A}$ is at least $(w_{G_{i-1}}(P) + w_{G_i}(Q))/{1/k} = \Omega(k)$.
    Since $k$ is arbitrary, the Lipschitz constant of $\mathcal{A}$ must be unbounded.
\end{proof}

\begin{theorem}\label{thm:mst-weighted-lower-bound}
    Any (randomized) $(1+\epsilon)$-approximation algorithm for the minimum spanning tree problem has Lipschitz constant $\Omega(\epsilon^{-1})$.
\end{theorem}
\begin{proof}
    Consider a weighted graph $G$ constructed by introducing two vertices $s$ and $t$ and connecting them via two edges $e_1$ and $e_2$ both with weight $1$.
    
    Let $\mathcal{A}$ be an arbitrary $(1+\epsilon)$-approximation algorithm for the minimum spanning tree problem.
    Without loss of generality, we can assume that $\mathcal{A}$ outputs $e_1$ with probability at least half.

    Consider another graph $G'$ obtained from $G$ by replacing the weight of $e_2$ with $1-10\epsilon$.
    Let $p'_1$ and $p'_2 := 1-p'_1$ be the probabilities that $A$ on $G'$ outputs the edges $e_1$ and $e_2$, respectively.
    Then, to achieve $(1+\epsilon)$-approximation, we must have
    \[
        p'_1 \cdot 1 + p'_2 \cdot (1-10\epsilon) \leq (1+\epsilon)(1-10\epsilon),
    \]
    which implies $p'_2 > 9/10+\epsilon$.
    It follows that the sensitivity is at least
    \[
        \frac{1}{10\epsilon} \left(\left(\frac{9}{10} + \epsilon\right)(1-10\epsilon) - \frac{1}{2}\cdot 1\right) = \Omega\left(\frac{1}{\epsilon}\right). \qedhere
    \]
\end{proof}


\begin{theorem}\label{thm:sp-weighted-lower-bound}
    Any (randomized) $(1+\epsilon)$-approximation algorithm for the shortest path problem has Lipschitz constant $\Omega(\epsilon^{-1})$.
\end{theorem}
\begin{proof}
    We reuse the graph and analysis for the proof of Theorem~\ref{thm:mst-weighted-lower-bound} to prove the claim.
\end{proof}

\begin{theorem}\label{thm:matching-lower-bound}
    Any (randomized) $\alpha$-approximation algorithm for the maximum weight matching problem has Lipschitz constant $\Omega(\alpha)$.
\end{theorem}
\begin{proof}
    Consider a graph $G=(V,E)$ consisting of two vertices $u,v$ and two edges $e_1,e_2$ connecting them.
    For $i \in \{1,2\}$, let $w_i\in \mathbb{R}_{\geq 0}^E$ be a weight vector such that $w_i(e_i) = 1$ and $w_i(e_{3-i}) = 0$.

    Let $\mathcal{A}$ be an $\alpha$-approximation algorithm for the maximum weight matching problem with $\alpha$.
    Then, $\mathcal{A}(G,w_i)$ must output $e_i$ with probability at least $\alpha$.
    Hence, we have
    \[
        \frac{\EMW((\mathcal{A}(G,w_1),w_1), (\mathcal{A}(G,w_2),w_2))}{\|w_1 - w_2\|_1}
        \geq \frac{2|1 \cdot \alpha - 0|}{2} = \alpha,
    \]
    and the claim holds.
\end{proof}


\section{Pointwise Lipschitzness for Unweighted Mapping}\label{sec:unweighted}

In this section, we consider pointwise Lipschitz continuity of graph problems with respect to the unweighted mapping.
We discuss the minimum spanning tree and maximum weight bipartite matching problems in Sections~\ref{subsec:unweighted-mst} and~\ref{subsec:unweighted-bipartite-matching}, respectively.

The following lemma is useful to bound the pointwise Lipschitz constant, which is a counterpart of Lemma~\ref{lem:seeoneelement}.
We defer the proof to Appendix~\ref{sec:seeoneelement}.
\begin{lemma}\label{lem:seeoneelement-unweighted}
Let $w\in \mathbb{R}_{\geq 0}^E$ be a weight vector. 
Suppose that there exist some $c>0$ and $L>0$ such that for all $w^*\in \mathbb{R}_{\geq 0}^E$ with $\|w-w^*\|_1\leq c$, 
\[
    \EMU\left(\mathcal{A}(G,w^*), \mathcal{A}(G,w^*+\delta \mathbf{1}_f)\right)\leq \delta L
\]
holds for any $f\in E$ and $0<\delta\leq c$. Then, the pointwise Lipschitz constant of $\mathcal{A}$ on $G$ at $w\in \mathbb{R}_{\geq 0}^E$ is at most $L$.
\end{lemma}

\subsection{Minimum Spanning Tree}\label{subsec:unweighted-mst}

In this section, we prove the following:
\begin{theorem}\label{thm:mst-unweighted}
For any $\epsilon>0$, there exists a polynomial-time $(1+\epsilon)$-approximation algorithm for the minimum spanning tree problem with the pointwise Lipschitz constant with respect to the unweighted mapping $O(\epsilon^{-1} n / \OPT)$, where $n$ is the number of vertices in the input graph and $\OPT$ is the minimum weight of a spanning tree.
\end{theorem}

\begin{algorithm}[t!]
\caption{(unweighted) Lipschitz-continuous algorithm for the minimum spanning tree problem}\label{alg:mst_main_u}
\Procedure{\emph{\Call{PLipMST}{$G=(V,E), w,\epsilon$}}}{
    Sample $\bm{b}$ uniformly from $\left[\frac{1}{2}\cdot \frac{\epsilon\cdot \OPT}{|V|-1},\frac{\epsilon\cdot \OPT}{|V|-1}\right]$\;\label{line:mstu_sampleb}
    \For{$e\in E$}{
        Sample $\widehat{\bm{w}}(e)$ uniformly from $\left[w(e),w(e)+\bm{b}\right]$\;\label{line:mstu_samplew}
    }
    \Return any minimum spanning tree of $G$ with respect to the weight vector $\widehat{\bm{w}}$.
}
\end{algorithm}

Our algorithm, \Call{PLipMST}{}, is presented in Algorithm~\ref{alg:mst_main_u}. 
Although the algorithm is almost identical to that of \Call{LipMST}{} presented in Section~\ref{sec:mst-weighted}, the analysis requires some care because we use $\OPT$, which can vary depending on the edge weights, in the algorithm.
First, we analyze the approximation ratio.
\begin{lemma}\label{lem:unweighted-mst-approximation}
For any graph $G=(V,E)$, a weight vector $w\in \mathbb{R}_{\geq 0}^E$, and $\epsilon > 0$, $\Call{PLipMST}{G,w,\epsilon}$ outputs a $(1+\epsilon)$-approximate spanning tree.
\end{lemma}
\begin{proof}
Let $T^*$ be an optimal spanning tree of $G$, and let $\bm{T}$ be the output of $\Call{PLipMST}{G,w,\epsilon}$.
Let $\OPT$ and $\widehat{\bm{\OPT}}$ be the minimum weights of spanning trees in $G$ with respect to $w$ and $\widehat{\bm{w}}$, respectively.
Then, we have
\begin{align*}
    \widehat{\bm{\OPT}} = \sum_{e\in \bm{T}}\widehat{\bm{w}}(e)\leq \sum_{e\in T^*}\widehat{\bm{w}}(e)\leq \sum_{e\in T^*}\left(w(e)+\frac{\epsilon\cdot \OPT}{|V|-1}\right)=(1+\epsilon)\OPT,
\end{align*}
where the first inequality is from the fact that $\bm{T}$ is a minimum spanning tree with respect to $\widehat{\bm{w}}$, and the claim holds.
\end{proof}

Now for a graph $G=(V,E)$ and $\epsilon > 0$, we analyze the pointwise Lipschitzness of $\Call{PLipMST}{G,\cdot,\epsilon}$ at a weight vector $w\in \mathbb{R}_{\geq 0}^E$.
By Lemma~\ref{lem:seeoneelement-unweighted}, it suffices to bound the following for $f \in E$ at $\delta >0$:
\[
    \EMU\left(\Call{PLipMST}{G,w,\epsilon}, \Call{PLipMST}{G,w+\delta \mathbf{1}_f,\epsilon}\right). \label{eq:mstu}
\]
To this end, we consider the coupling $\mathcal{D}$ between $\left(\bm{b}, \widehat{\bm{w}}\right)$ and $(\bm{b}^{\delta f},\widehat{\bm{w}}^{\delta f})$ defined as follows: For each $b\in \mathbb{R}_{\geq 0}$ and weight vector $\widehat{w}\in \mathbb{R}_{\geq 0}^E$, we transport the probability mass for $(\bm{b},\widehat{\bm{w}})=(b,\widehat{w})$ to that for $(\bm{b}^{\delta f},\widehat{\bm{w}}^{\delta f})=(b,\widehat{w})$ as far as possible in such way that the probability mass for $\widehat{\bm{w}}|_{E\setminus \{f\}}=\widehat{w}|_{E\setminus \{f\}}$ is transported to that for $\widehat{\bm{w}}^{\delta f}|_{E\setminus \{f\}}=\widehat{w}^{\delta f}|_{E\setminus \{f\}}$. The remaining mass is transported arbitrarily.
For clarity, we denote the minimum weights of spanning trees in $G$ with respect to weight vectors $w$ and $w+\delta \mathbf{1}_f$ by $\OPT(G,w)$ and $\OPT(G,w+\delta \mathbf{1}_f)$, respectively.

\begin{lemma}\label{lem:unweighted-mst-b}
We have
\begin{align*}
    \TV\left(\bm{b},\bm{b}^{\delta f}\right)\leq \frac{2\delta}{\OPT(G,w)}.
\end{align*}
\end{lemma}
\begin{proof}
We have
\begin{align*}
    \TV\left(\bm{b},\bm{b}^{\delta f}\right)
    &= \Pr\left[\bm{b}^{\delta f}\in \left[\frac{\epsilon\cdot \OPT(G,w)}{|V|-1},\frac{\epsilon\cdot \OPT(G,w+\delta \mathbf{1}_f)}{|V|-1}\right]\right]\\
    &= 2\cdot \frac{\OPT(G,w+\delta \mathbf{1}_f)-\OPT(G,w)}{\OPT(G,w+\delta \mathbf{1}_f)}\\
    &\leq \frac{2\delta}{\OPT(G,w+\delta \mathbf{1}_f)}\leq \frac{2\delta}{\OPT(G,w)}.\qedhere
\end{align*}
\end{proof}

\begin{lemma}\label{lem:unweighted-mst-remaining-mass}
    For any $b\in \mathbb{R}_{\geq 0}$, we have
  \[
    \E_{((\bm{b},\widehat{\bm{w}}),(\bm{b}^{\delta f},\widehat{\bm{w}}^{\delta f}))\sim \mathcal{D}}\left[\TV(\widehat{\bm{w}},\widehat{\bm{w}}^{\delta f})\mid \bm{b}=\bm{b}^{\delta f}=b\right]\leq \frac{\delta (|V|-1)}{\epsilon\cdot \OPT(G,w)}.
  \]
\end{lemma}
\begin{proof}
Recall that if $\bm{b}=\bm{b}^{\delta f}=b$, then the distributions of $\widehat{\bm{w}}|_{E\setminus \{f\}}=\widehat{w}|_{E\setminus \{f\}}$ and $\widehat{\bm{w}}^{\delta f}|_{E\setminus \{f\}}=\widehat{w}^{\delta f}|_{E\setminus \{f\}}$ are the same. 
Thus it suffices to bound
\begin{align*}
    \E_{((\bm{b},\widehat{\bm{w}}),(\bm{b}^{\delta f},\widehat{\bm{w}}^{\delta f}))\sim \mathcal{D}}\left[\TV(\widehat{\bm{w}}(f),\widehat{\bm{w}}^{\delta f}(f))\mid \bm{b}=\bm{b}^{\delta f}=b\right].
\end{align*}
Note that $\widehat{\bm{w}}(f)$ and $\widehat{\bm{w}}^{\delta f}(f)$ are sampled uniformly from 
\[
\left[w(f),w(f)+b\right] \quad \text{and} \quad  \left[w(f)+\delta,w(f)+\delta+b\right],
\]
respectively.
Hence, we have
\begin{align*}
    \E_{((\bm{b},\widehat{\bm{w}}),(\bm{b}^{\delta f},\widehat{\bm{w}}^{\delta f}))\sim \mathcal{D}}\left[\TV(\widehat{\bm{w}}(f),\widehat{\bm{w}}^{\delta f}(f))\mid \bm{b}=\bm{b}^{\delta f}=b\right] &\leq \Pr\left[\widehat{\bm{w}}(f)\in \left[w(f),w(f)+\delta\right]\right]\\
    &= \frac{\delta}{b}\leq \frac{\delta (|V|-1)}{\epsilon\cdot \OPT(G,w)}.\qedhere
\end{align*}
\end{proof}

\begin{lemma}\label{lem:unweighted-mst-lipschitzness}
The pointwise Lipschitz constant of $\Call{PLipMST}{G,\cdot,\epsilon}$ at a weight vector $w\in \mathbb{R}_{\geq 0}^E$ with respect to the unweighted mapping is $O\left(\frac{|V|}{\epsilon \cdot \OPT(G,w)}\right)$.
\end{lemma}
\begin{proof}
Let $w^*\in \mathbb{R}^E_{\geq 0}$ be a weight vector with $\|w-w^*\|\leq \frac{\OPT(G,w)}{2}$. 
By Lemmas~\ref{lem:mst-local-update},~\ref{lem:unweighted-mst-b} and~\ref{lem:unweighted-mst-remaining-mass}, we have
\begin{align*}
    \EMU\left(\Call{MST}{G,w^*,\epsilon},\Call{MST}{G,w^*+\delta \mathbf{1}_f, \epsilon}\right)
    &\leq \frac{2\delta}{\OPT(G,w^*)}\cdot (|V|-1) + \frac{\delta (|V|-1)}{\epsilon\cdot \OPT(G,w^*)}\cdot 2\\
    &\leq O\left(\frac{\delta|V|}{\epsilon \cdot \OPT(G,w^*)}\right)\leq O\left(\frac{\delta|V|}{\epsilon \cdot \OPT(G,w)}\right).
\end{align*}
By Lemma~\ref{lem:seeoneelement-unweighted}, the pointwise Lipschitz constant is bounded by
\[
    \min_{\|w-w^*\|\leq \frac{\OPT(G,w)}{2}, \delta > 0, f\in E}\frac{\EMU\left(\Call{MST}{G,w^*,\epsilon},\Call{MST}{G,w^*+\delta \mathbf{1}_f, \epsilon}\right)}{\delta}\leq O\left(\frac{|V|}{\epsilon \cdot \OPT(G,w)}\right).
    \qedhere
\]
\end{proof}
Theorem~\ref{thm:mst-unweighted} follows by combining Lemmas~\ref{lem:unweighted-mst-approximation} and~\ref{lem:unweighted-mst-lipschitzness}.

\subsection{Maximum Weight Bipartite Matching}\label{subsec:unweighted-bipartite-matching}



In this section, we present a pointwise Lipschitz continuous algorithm for the maximum weight bipartite matching problem and prove the following:
\begin{theorem}\label{thm:mm-unweighted}
    There exists a polynomial-time $(1/2-\epsilon)$-approximation algorithm for the maximum weight bipartite matching problem with a pointwise Lipschitz constant with respect to the unweighted mapping $O( \epsilon^{-1} \cdot n^{3/2}\log m /\OPT)$, where $n$ and $m$ are the numbers of vertices in the left and right parts of the input bipartite graph, respectively, and $\OPT$ is the maximum weight of a matching.
\end{theorem}

Let $G=(U \cup V, E)$ be a complete bipartite graph, that is, $E = U \times V$, and $w \in \mathbb{R}_{\geq 0}^E$ be a weight vector.
We assume that vertices in $U$ and $V$ are indexed as $\{1,\dots, |U|\}$ and $\{1,\dots, |V|\}$, respectively.
Let us consider the following LP relaxation for the maximum weight bipartite matching problem.
\begin{align*}
    \begin{array}{rl}
        \text{LP}(G,w):=\max & \displaystyle \sum_{e \in E}w(e)x_{e}\\
        \text{s.t.} & \begin{cases}
        \displaystyle \sum_{j=1}^{|V|}x_{i,j} \leq 1 & (i=1,\dots, |U|)\\
        \displaystyle \sum_{i=1}^{|U|}x_{i,j} \leq 1 & (j=1,\dots, |V|)\\
        x_{e} \geq 0 &(e \in E).
        \end{cases}
    \end{array}
\end{align*}
The downside of this LP is that its optimal solution is not stable for small changes in the input weight vector $w$.
To avoid this issue, we use \emph{entropy regularization}~\cite{cuturi2013sinkhorn}. 
Specifically, we consider the following optimization problem. 
\begin{align*}
    \begin{array}{rl}
        \LPmod(G,w,B):=\max & \displaystyle \sum_{e \in E}w(e)x_e-B\sum_{e \in E}x_e\log x_e\\
        \text{s.t.} & \begin{cases}
        \displaystyle \sum_{j=1}^{|V|}x_{i,j} \leq 1 & (i=1,\dots, |U|)\\
        \displaystyle \sum_{i=1}^{|U|}x_{i,j} \leq 1 & (j=1,\dots, |V|)\\
        x_e \geq 0 &(e \in E),
        \end{cases}
    \end{array}\label{eq:modified_LP}
\end{align*}
where $B>0$ is a parameter.
The second term (including the negative sign) of the objective function is non-negative; $\log x_e$ is always non-positive for $0\leq x_e \leq 1$.

Our algorithm solves $\LPmod(G,w,\bm{B})$ for a suitably chosen $\bm{B}$ and then rounds the obtained fractional solution $\bm{x}$. 
The rounding process is as follows: First, for each $i\in \{1,\dots, |U|\}$ we sample an index $\bm{p}(i)\in \{1,\dots, |V|\}$ using $\{\bm{x}_{i,j}\}_{j\in V}$. Intuitively, $\bm{p}(i)$ represents the candidate of the vertex in $|V|$ that will be matched to $i$. 
Because there may be multiple indices $i$ with the same $\bm{p}(i)$ value, for each $j\in \{1,\dots, |V|\}$, we choose an index $\bm{q}(j)\in \{1,\dots, |U|\}$ with $\bm{p}(i)=j$ and include $(\bm{q}(j),j)$ in the output matching.
See Algorithm~\ref{alg:match_main} for details.

\begin{algorithm}[t!]
\caption{Pointwise Lipschitz continuous algorithm for the maximum weight bipartite matching problem}\label{alg:match_main}
\Procedure{\emph{\Call{PLipMWBM}{$G=(U\cup V,E), w, \epsilon$}}}{
    Sample $\bm{B}$ uniformly from $\left[\frac{\epsilon\cdot \OPT(G,w)}{|U|\log |V|},\frac{2\epsilon \cdot \OPT(G,w)}{|U|\log |V|}\right]$\;\label{line:MWBM_sampleb}
    Let $\bm{x}$ be the optimal solution to $\LPmod(G,w,\bm{B})$\;
    Let $\bm{C}_{j}\leftarrow \emptyset$ for all $j \in \{1,\dots, |V|\}$\;
    \For{$i\in \{1,\dots, |U|\}$}{
        Set
        \[
        \bm{p}(i) = \begin{cases}
        j & \text{with probability }\bm{x}_{i,j}\text { for }j \in \{1,\ldots,|V|\}, \\
        \bot & \text{with the remaining probability}.
        \end{cases}
        \]\label{line:MWBM_samplep}
        \If{$\bm{p}(i) \neq \bot$}{Add $i$ to $\bm{C}_{\bm{p}(i)}$.}
    }
    $\bm{M}\leftarrow \emptyset$\;
    \For{$j\in \{1,\dots, |V|\}$}{
        \If{$\bm{C}_j\neq \emptyset$}{
            Sample $\bm{q}(j)$ uniformly from $\bm{C}_j$\;\label{line:MWBM_sampleq}
            Add the edge $\{u_{\bm{q}(j)},v_j\}$ to $\bm{M}$.
        }
    }
    \Return $\bm{M}$.
}
\end{algorithm}

\subsubsection{Approximation Ratio}
The goal of this section is to show the following:
\begin{lemma}\label{lem:mwbm_approx}
The approximation ratio of \Call{PLipMWBM}{} is (at least)  $\frac{1}{2}-\epsilon$.
\end{lemma}

Throughout this section, we fix a complete bipartite graph $G=(U\cup V, E=U\times V)$ and a weight vector $w\in \mathbb{R}_{\geq 0}^E$.
First, we analyze the loss caused by entropy regularization.
\begin{lemma}\label{lem:LPapprox}
Let $\bm{x} \in \mathbb{R}^E$ be the vector as constructed in $\Call{PLipMWBM}{G,w,\epsilon}$.
We have
\begin{align*}
    \sum_{e \in E}w(e)\bm{x}_{e}\geq (1-2\epsilon)\OPT(G,w).
\end{align*}
\end{lemma}
\begin{proof}
Let $x^*$ be the optimal solution to $\LP(G,w)$.  
Then, we have
\begin{align*}
    & \OPT(G,w) = \sum_{e \in E}w(e)x^*_e
    \leq \sum_{e \in E}w(e)x^*_e - \bm{B}\sum_{e \in E}x^*_e\log x^*_e \\
    &\leq \sum_{e \in E}w(e)\bm{x}_e - \bm{B}\sum_{e \in E}\bm{x}_e\log \bm{x}_e
    \leq \sum_{e \in E}w(e)\bm{x}_e - \bm{B} |U|\log |V|\\
    &\leq \sum_{e \in E}w(e)\bm{x}_e - 2\epsilon \cdot \OPT(G,w)
    \leq (1-2\epsilon)\OPT(G,w),
\end{align*}
where the first inequality is from the non-negativeness of the regularization term, the second inequality is from the optimality of $x^*$, the third inequality is from
\begin{align*}
    -\sum_{e \in E}\bm{x}_e\log \bm{x_e}\leq \sum_{i=1}^{|U|}\log |V|\leq |U|\log |V|,
\end{align*}
and the last inequality is from $\bm{B}\leq \frac{2\epsilon\cdot \OPT(G,w)}{|U|\log |V|}$.
\end{proof}

Next, we analyze the loss caused by rounding.
Let $y \in [0,1]^n$ be a vector, and consider tossing $n$ independent coins with the $t$-th coin turning up with probability $y_i$.
Then, the \emph{Poisson binomial distribution}~\cite{hoeffding1956distribution} with respect to $y$, denoted $P(y)$, is the probability distribution of the number of coins turning up.
The following shows that this distribution plays an important role in the analysis of the performance of the rounding.
\begin{lemma}
Let $i'\in \{1,\dots, |U|\}$. Then, $|\bm{C}_{\bm{p}(i')}|-1$ follows the Poisson binomial distribution with respect to the vector $(x_{1,\bm{p}(i')},\dots, x_{i'-1,\bm{p}(i')},x_{i'+1,\bm{p}(i')},\dots,x_{|U|,\bm{p}(i')})$.
\end{lemma}
\begin{proof}
For each $i\in \{1,\dots, |U|\}$ and $j\in \{1,\dots, |V|\}$, we have $\Pr\left[\bm{p}(i)=j\right]=\bm{x}_{i,j}$. 
Additionally, $\bm{p}(i)\; (i\in\{1,\dots, |U|\})$ are pairwise independent. 
Therefore, $|\bm{C}_{p(i')}|-1$ can be regarded as the number of coins turning up when we toss $|U|$ independent coins with turning-up probabilities of $x_{1,\bm{p}(i')},\ldots,x_{i'-1,\bm{p}(i')},x_{i'+1,\bm{p}(i')},\dots,x_{|U|,\bm{p}(i')}$.
\end{proof}

Let 
\begin{align*}
    Y = \inf_{n\in \mathbb{Z}_{\geq 1},\sum_{t=1}^{n}y_t\leq 1}\sum_{k=0}^{\infty}\left(\frac{1}{k+1}\Pr_{\bm{k} \sim P(y)}\left[\bm{k}=k\right]\right).
\end{align*}
Then, the output matching size can be bounded by using $Y$ as can be seen in the following:
\begin{lemma}\label{lem:roundapprox}
    Let $\bm{x} \in \mathbb{R}^E$ be the vector and $\bm{M}$ be the matching as constructed in $\Call{PLipMWBM}{G,w,\epsilon}$.
    Then, we have
    \begin{align*}
        \E\left[\sum_{e \in \bm{M}}w(e)\right] \geq 
        Y \sum_{e \in E}w(e)\bm{x}_e.
    \end{align*}
\end{lemma}
\begin{proof}
We have
\begin{align*}
    & \E\left[\sum_{e\in \bm{M}}w(e)\right] 
    = \sum_{e \in E}\left(w(e)\cdot \Pr[e \in \bm{M}]\right)
    = \sum_{i=1}^{|U|}\sum_{j=1}^{|V|}\left(w(i,j)\cdot \Pr\left[\bm{p}(i)=j\right]\cdot \Pr\left[\bm{q}(j)=i\mid \bm{p}(i)=j\right]\right)\\
    &= \sum_{i=1}^{|U|}\left(\Pr\left[\bm{q}(\bm{p}(i))=i\right]\cdot \sum_{j=1}^{|V|}\left(w(i,j)\cdot \Pr\left[\bm{p}(i)=j\right]\right)\right)
    = \sum_{i=1}^{|U|}\left(\E\left[\frac{1}{\left|\bm{C}_{\bm{p}(i)}\right|}\right]\cdot \sum_{j=1}^{|V|}w(i,j)\bm{x}_{i,j}\right)\\
    &= \sum_{i=1}^{|U|}\left(\sum_{k=0}^{\infty}\left(\frac{1}{k+1}\cdot \Pr\left[P((\bm{x}_{1,\bm{p}(i)},\dots,\bm{x}_{i-1,\bm{p}(i)},\bm{x}_{i+1,\bm{p}(i)}\dots,\bm{x}_{|U|,\bm{p}(i)}))=k\right]\right)\cdot \sum_{j=1}^{|V|}w(i,j)\bm{x}_{i,j}\right)\\
    &\geq Y \sum_{e \in E}w(e)\bm{x}_e.
    \qedhere
\end{align*}
\end{proof}

Now we evaluate $Y$. 
\begin{lemma}\label{lem:evalY}
$Y = \frac{1}{2}$.
\end{lemma}
\begin{proof}
For $y_1,\ldots,y_n \in [0,1]$ with $\sum_{i=1}^n y_i \leq 1$, let
\begin{align*}
    Y(y_1,\dots, y_n) =\sum_{k=0}^{\infty}\left(\frac{1}{k+1}\Pr_{\bm{k} \sim P(y_1,\ldots,y_n)}\left[\bm{k}=k\right]\right).
\end{align*}

Suppose $n\geq 2$. 
We show that deleting two elements of $(y_1,\dots, y_n)$, say $y_1$ and $y_2$, and appending their sum does not increase the value. 
To show this, for $k\in \mathbb{Z}_{\geq 0}$, we let $Q(k)= \Pr_{\bm{k} \sim P(y_3,\ldots,y_n)}\left[\bm{k}=k\right]$.
We have
\begin{align*}
    &Y(y_1,y_2,y_3,\dots, y_n)\\
    &=\sum_{k=0}^{\infty}\left(\left(\frac{1}{k+1}\cdot (1-y_1)(1-y_2) + \frac{1}{k+2}\cdot ((1-y_1)y_2+y_1(1-y_2)) + \frac{1}{k+3}\cdot y_1y_2\right)Q_k\right),\\
    &Y(y_1+y_2,y_3,\dots, y_n)\\
    &=\sum_{k=0}^{\infty}\left(\left(\frac{1}{k+1}\cdot (1-(y_1+y_2)) + \frac{1}{k+2}\cdot (y_1+y_2)\right)Q_k\right).
\end{align*}
Then, we have
\begin{align*}
    &Y(y_1,y_2,y_3,\dots, y_n)-Y(y_1+y_2,y_3,\dots, y_n)\\
    &=\sum_{k=0}^{\infty}\left(\left(\frac{1}{k+1}\cdot y_1y_2 - \frac{1}{k+2}\cdot 2y_1y_2 + \frac{1}{k+3}\cdot y_1y_2\right)Q_k\right)\\
    &=y_1y_2\cdot \sum_{k=0}^{\infty}\left(\left(\frac{1}{k+1} - \frac{2}{k+2} + \frac{1}{k+3}\right)Q_k\right)
    \geq 0.
\end{align*}
Hence, the minimum of $Y(y_1,\ldots,y_n)$ is attained when $n=1$.
We have
\begin{align*}
    Y(y_1)=(1-y_1)+\frac{1}{2}\cdot y_1\geq \frac{1}{2},
\end{align*}
where the equality holds when $y_1 = 1$, and hence the lemma holds.
\end{proof}

\begin{proof}[Proof of Lemma~\ref{lem:mwbm_approx}]
Combining Lemmas~\ref{lem:LPapprox},~\ref{lem:roundapprox} and~\ref{lem:evalY} yields the claim.
\end{proof}

\subsubsection{Stability of Entropy Regularization}
In this section, we discuss the stability of the solution of the LP with entropy regularization.


For a graph $G=(U\cup V, E)$, a weight vector $w\in \mathbb{R}_{\geq 0}^E$, and $B>0$, let
\[
    f(x) = \sum_{e \in E}w(e) x_e - B\sum_{e \in E}x_e \log x_e
\]
be the objective function of $\LPmod$ with a weight vector $w\in \mathbb{R}_{\geq 0}^E$.
We first discuss the following property of $f$.
\begin{lemma}\label{lem:modified-strong-convexity}
    For any $x,y \in \mathbb{R}^E$ with $\sum_{j\in V}x_{i,j}\leq 1$ and $\sum_{j \in V}y_{i,j} \leq 1$, we have
    \[
        f(x) - f(y) \geq \nabla f(x)^\top (x-y) + \frac{B}{2}\sum_{i \in U} \left(\sum_{j \in V}|x_{i,j}-y_{i,j}|\right)^2
    \]
\end{lemma}
\begin{proof}
    Because $f$ is concave, for $z = x - y$, it suffices to show
    \[
        z^\top H(x) z \leq - B \sum_{i \in U}\left(\sum_{j \in V}|z_{i,j}|\right)^2,
    \]
    where $H(x)$ is the Hessian of $f$ at $x$.
    Indeed, we have
    \begin{align*}
        & z^\top H(x) z = - B \sum_{e \in E}\frac{z_e^2}{x_e}
        = - B \sum_{i\in U} \sum_{j \in V}\frac{z_{i,j}^2}{x_{i,j}}  \\
        & \leq  - B \sum_{i\in U} \left(\sum_{j \in V}\frac{z_{i,j}^2}{x_{i,j}} \cdot \sum_{j \in V}x_{i,j}\right)
        \leq  - B \sum_{i\in U} \left(\sum_{j \in V}|z_{i,j}|\right)^2.
        \qedhere
    \end{align*}
\end{proof}

\begin{lemma}\label{lem:stability-of-entropic-regularization}
Let $G=(U\cup V, E)$ be a graph, $w\in \mathbb{R}_{\geq 0}^E$ be a weight vector, $f\in E$,  and $\delta>0$. Let $x, x^{\delta f} \in \mathbb{R}^E$ be the optimal solutions to $\LPmod(G,w,B)$ and $\LPmod(G,w+\delta \mathbf{1}_f, B)$, respectively.
Then, we have
\[
    \|x - x^{\delta f}\|_1 = O\left(\frac{\delta \sqrt{|U|}}{B}\right).
\]
\end{lemma}
\begin{proof}
Let 
\begin{align*}
    f^{\delta f}(x)&=\sum_{e\in E \setminus \{f\}}w(e)x_{e} + (w(f)+\delta)x_f -B\sum_{e\in E}x_{e}\log x_e.
\end{align*}
be the objective function of $\LPmod$ with the weight vector $w+\delta \bm{1}_f$.

By Lemma~\ref{lem:modified-strong-convexity}, we have
\begin{align*}
    f(x) - f(x^{\delta f}) \geq \nabla f(x)^\top (x - x^{\delta f}) + \frac{B}{2} \sum_{i \in U}\left(\sum_{j \in V}|x_{i,j} - x^{\delta f}_{i,j}|\right)^2 \geq \frac{B}{2} \sum_{i \in U}\left(\sum_{j \in V}|x_{i,j} - x^{\delta f}_{i,j}|\right)^2.
\end{align*}
Similarly, we have
\begin{align*}
    f^{\delta f}(x^{\delta f}) - f(x) \geq \frac{B}{2} \sum_{i \in U}\left(\sum_{j \in V}|x_{i,j} - x^{\delta f}_{i,j}|\right)^2.
\end{align*}
Also, we have
\begin{align*}
    f(x) - f(x^{\delta f}) + f^{\delta f}(x^{\delta f}) - f^{\delta f}(x) = \delta (x^{\delta f}_f - x_f).
\end{align*}
Combining the above three, we have
\begin{align}
    B \sum_{i \in U}\left(\sum_{j \in V}|x_{i,j} - x^{\delta f}_{i,j}|\right)^2 \leq \delta (x^{\delta f}_f - x_f).    
    \label{eq:stability-of-entropic-regularization}
\end{align}
In particular, this implies
\[
    B (x^{\delta f}_f - x_f)^2 \leq \delta (x^{\delta f}_f - x_f),
\]
and hence we have $|x^{\delta f}_f - x_f| \leq \delta/B$.
By substituting this into~\eqref{eq:stability-of-entropic-regularization}, we obtain
\[
    \sum_{i \in U}\left(\sum_{j \in V}|x_{i,j} - x^{\delta f}_{i,j}|\right)^2 \leq \frac{\delta^2}{B^2}.
\]
Then, we have
\[
    \|x - x^{\delta f}\|_1 = \sum_{i \in U}\sum_{j \in V}|x_{i,j} - x^{\delta f}_{i,j}| \leq  \sqrt{\sum_{i \in U}\left(\sum_{j \in V}|x_{i,j} - x^{\delta f}_{i,j}|\right)^2} \cdot \sqrt{\sum_{i \in U}1^2} = \frac{\sqrt{|U|} \delta}{B}.
    \qedhere
\]

\end{proof}

\subsubsection{Pointwise Lipschitzness}

The goal of this section is to show the following:
\begin{lemma}\label{lem:mwbm_main}
    For any bipartite graph $G=(U \cup V,E)$, a weight vector $w\in \mathbb{R}_{\geq 0}^E$, $\delta >0$, and $f \in E$, we have
    \begin{align*}
        \EMU(\Call{PLipMWBM}{G,w,\epsilon}, \Call{PLipMWBM}{G, w+\delta \mathbf{1}_f, \epsilon}) \leq O\left(\frac{\delta\epsilon^{-1} |U|^{\frac{3}{2}}\log |V|}{\OPT(G,w)}\right).
    \end{align*}
\end{lemma}

To bound the earth mover's distance between the output matchings for $(G,w)$ and $(G,w+\delta \bm{1}_f)$, we consider the coupling $\mathcal{W}$ between $(\bm{B},\bm{p},\bm{q})$ and $(\bm{B}^{\delta f},\bm{p}^{\delta f},\bm{q}^{\delta f})$, defined as follows: For each number $B \in \mathbb{R}$, we transport the probability mass for $\bm{B}=B$ to that for $\bm{B}^{\delta f}=B$ as far as possible in such a way that 
\begin{align*}
    \E_{((\bm{B},\bm{p},\bm{q}),(\bm{B}^{\delta f},\bm{p}^{\delta f},\bm{q}^{\delta f}))\sim \mathcal{W}}\left[\|\bm{q}-\bm{q}^{\delta f}\|_0\mid \bm{B}=\bm{B}^{\delta f}=B\right],
\end{align*}
is minimized, where $\|\bm{q}-\bm{q}^{\delta f}\|_0$ is the number of indices $j$ with $\bm{q}(j)\neq \bm{q}^{\delta f}(j)$.
The remaining mass is transported arbitrarily.
Note that $\|\bm{q}-\bm{q}^{\delta f}\|_0$ is equal to $d_u(\bm{M},\bm{M}^{\delta f})$, where $\bm{M}$ and $\bm{M}^{\delta f}$ are the matchings constructed using $\bm{q}$ and $\bm{q}^{\delta f}$, respectively.

First, we evaluate the remaining mass.
\begin{lemma}\label{lem:match_u_tvdb}
We have
\begin{align*}
    \TV\left(\bm{B},\bm{B}^{\delta f}\right)\leq \frac{2\delta}{\OPT(G,w)}
\end{align*}
\end{lemma}
\begin{proof}
We have
\begin{align*}
    \TV\left(\bm{B},\bm{B}^{\delta f}\right) 
    &= \Pr\left[\bm{B}^{\delta f}\in \left[\frac{2\OPT(G,w)}{|U|\log |V|},\frac{2\OPT(G,w+\delta \mathbf{1}_{f})}{|U|\log |V|}\right]\right]\\
    &= 2\cdot \frac{\OPT(G,w+\delta \mathbf{1}_{f})-\OPT(G,w)}{\OPT(G,w+\delta \mathbf{1}_{f})}\\
    &\leq \frac{2\delta}{\OPT(G,w+\delta \mathbf{1}_{f})}\leq \frac{2\delta}{\OPT(G,w)}.
    \qedhere
\end{align*}
\end{proof}

Now we evaluate the distance between $\bm{p}$ and $\bm{p}^{\delta f}$.
\begin{lemma}\label{lem:distp}
Let $B\in \mathbb{R}_{\geq 0}$ and let $x, x^{\delta f} \in \mathbb{R}_{\geq 0}^{E}$ be the optimal solutions to $\LPmod(G,w)$ and $\LPmod(G, w+\mathbf{1}_{\delta})$, respectively. Then, we have
\begin{align*}
    \E_{((\bm{B},\bm{p},\bm{q}),(\bm{B}^{\delta f},\bm{p}^{\delta f},\bm{q}^{\delta f}))\sim \mathcal{W}}\left[\|\bm{p}-\bm{p}^{\delta f}\|_0\mid \bm{B}=\bm{B}^{\delta f}=B\right]\leq \|x-x^{\delta f}\|_1.
\end{align*}
\end{lemma}
\begin{proof}
First, we have
\begin{align*}
    \left|\Pr\left[\bm{p}(i)=j\right]-\Pr[\bm{p}^{\delta f}(i)=j]\right| = \left|x_{i,j}-x^{\delta f}_{i,j}\right|
\end{align*}
for all $i\in \{1,\dots, |U|\}$ and $j\in \{1,\dots, |V|\}$. 
Because the $\bm{p}(i)$ values are independent of each other when $\bm{B}=\bm{B}^{\delta f}=B$, taking the product of the transportation for all $i\in \{1,\dots, |U|\}$ yields the lemma.
\end{proof}

Now we evaluate the distance between $\bm{q}$ and $\bm{q}^{\delta f}$.
\begin{lemma}\label{lem:distq}
Let $B\in \mathbb{R}_{\geq 0}$ and let $x$ and $x^{\delta f}$ be the optimal solutions to $\LPmod(G,w)$ and $\LPmod(G, w+\mathbf{1}_{\delta})$, respectively. Then, we have
\begin{align*}
    \E_{((\bm{B},\bm{p},\bm{q}),(\bm{B}^{\delta f},\bm{p}^{\delta f},\bm{q}^{\delta f}))\sim \mathcal{W}}\left[\|\bm{q}-\bm{q}^{\delta f}\|_0\mid \bm{B}=\bm{B}^{\delta f}=B\right]\leq 2\cdot \|x-x^{\delta f}\|_1.
\end{align*}
\end{lemma}
\begin{proof}
Fix $B,p,p^{\delta f} \in \mathbb{R}_{\geq 0}$, and let $(\bm{q},\bm{q}^{\delta f})$ be sampled from the marginal distribution of $\mathcal{W}$ conditioned on $\bm{B} = \bm{B}^{\delta f} = B$, $\bm{p}=p$, and $\bm{p}^{\delta f}=p^{\delta f}$. Note that under this condition, $C$ and $C^{\delta f}$ are no longer random variables because they are  deterministically determined by $p$ and $p^{\delta f}$, respectively.
Then, we have 
\begin{align}
    &\E_{((\bm{B},\bm{p},\bm{q}),(\bm{B}^{\delta f},\bm{p}^{\delta f},\bm{q}^{\delta f}))\sim \mathcal{W}}\left[\|\bm{q}-\bm{q}^{\delta f}\|_0\mid \bm{B}=\bm{B}^{\delta f}=B, \bm{p}=p, \bm{p}=p^{\delta f}\right] \nonumber \\
    &\leq \sum_{i=1}^{|U|}\sum_{j=1}^{|V|}\left|\Pr\left[\bm{q}(j)=i\right]-\Pr[\bm{q}^{\delta f}(j)=i]\right| \nonumber \\ 
    &= \sum_{i=1}^{|U|}\sum_{j=1}^{|V|}\left|\frac{1}{|C_j|}\bm{1}_{p(i)=j}-\frac{1}{|C^{\delta f}_j|}\bm{1}_{p^{\delta f}(i)=j}\right| \nonumber \\
    &\leq \sum_{i=1}^{|U|}\sum_{j=1}^{|V|}\left|\bm{1}_{p(i)=j}-\bm{1}_{p^{\delta f}(i)=j}\right|+\sum_{j=1}^{|V|}\left(|C_j|\cdot \left|\frac{1}{|C_j|}-\frac{1}{|C^{\delta f}_j|}\right|\right) \nonumber \\
    &= \|p-p^{\delta f}\|_0+\sum_{j=1}^{|V|}\left(|C_j|\cdot \left|\frac{1}{|C_j|}-\frac{1}{|C^{\delta f}_j|}\right|\right) \nonumber \\
    &\leq \|p-p^{\delta f}\|_0+\sum_{j=1}^{|V|}\left|\frac{|C^{\delta f}_j|-|C_j|}{|C^{\delta f}_j|}\right| \nonumber \\
    &\leq \|p-p^{\delta f}\|_0+\sum_{j=1}^{|V|}\left||C^{\delta f}_j|-|C_j|\right| \nonumber \\
    &= \|p-p^{\delta f}\|_0+\sum_{j=1}^{|V|}\sum_{i=1}^{|U|}\left|\bm{1}_{p(i)=j}-\bm{1}_{p^{\delta f}(i)=j}\right| \nonumber \\
    &= 2\cdot \|p-p^{\delta f}\|_0, \label{eq:distq}
\end{align}
where the first inequality is from the fact that $\bm{q}(j)$ (and $\bm{q}^{\delta f}(j)$) are pairwise independent for $j\in \{1,\dots, |V|\}$ and the second inequality is from
\begin{align*}
    &\left|\frac{1}{|C_j|}\bm{1}_{p(i)=j}-\frac{1}{|C^{\delta f}_j|}\bm{1}_{p^{\delta f}(i)=j}\right|\leq \begin{cases}
        1 & (\bm{1}_{p(i)=j}\neq \bm{1}_{p^{\delta f}(i)=j}), \\
        \left|\frac{1}{|C_j|}-\frac{1}{|C^{\delta f}(j)|}\right| & (\bm{1}_{p(i)=j}=\bm{1}_{p^{\delta f}(i)=j}=1),\\
        0 & (\bm{1}_{p(i)=j}=\bm{1}_{p^{\delta f}(i)=j}=0).
    \end{cases}
\end{align*}
and the fact that there are at most $|C_j|$ indices $i$ with $\bm{1}_{p(i)=j}=\bm{1}_{p^{\delta f}(i)=j}=1$.
By integrating~\eqref{eq:distq} over $p$ and $p^{\delta f}$, we obtain
\begin{align*}
    &\E_{((\bm{B},\bm{p},\bm{q}),(\bm{B}^{\delta f},\bm{p}^{\delta f},\bm{q}^{\delta f}))\sim \mathcal{W}}\left[\|\bm{q}-\bm{q}^{\delta f}\|_0\mid \bm{B}=\bm{B}^{\delta f}=B\right]\\
    &\leq \E_{((\bm{B},\bm{p},\bm{q}),(\bm{B}^{\delta f},\bm{p}^{\delta f},\bm{q}^{\delta f}))\sim \mathcal{W}}\left[2\cdot \|\bm{p}-\bm{p}^{\delta f}\|_0\mid \bm{B}=\bm{B}^{\delta f}=B\right] 
    \leq 2\cdot \|x-x^{\delta f}\|_1,
\end{align*}
where the second inequality is from Lemma~\ref{lem:distp}. 
\end{proof}

Now we can bound the pointwise Lipschitz constant as follows.
\begin{proof}[Proof of Lemma~\ref{lem:mwbm_main}]
We have
\begin{align*}
    &\EMU(\Call{PLipMWBM}{G,w,\epsilon}, \Call{PLipMWBM}{G, w+\delta \mathbf{1}_f, \epsilon})\\
    &\leq \TV\left(\bm{B},\bm{B}^{\delta f}\right)\cdot |U| + \E_{((\bm{B},\bm{p},\bm{q}),(\bm{B}^{\delta f},\bm{p}^{\delta f},\bm{q}^{\delta f}))\sim \mathcal{W}}\left[\|\bm{q}-\bm{q}^{\delta f}\|_0\mid \bm{B}=\bm{B}^{\delta f}\right] \\
    &\leq \TV\left(\bm{B},\bm{B}^{\delta f}\right)\cdot |U| + \E_{((\bm{B},\bm{p},\bm{q}),(\bm{B}^{\delta f},\bm{p}^{\delta f},\bm{q}^{\delta f}))\sim \mathcal{W}}\left[2\cdot \|\bm{x}-\bm{x}^{\delta f}\|_1\mid \bm{B}=\bm{B}^{\delta f}\right]\\
    &\leq \frac{2\delta}{\OPT(G,w)}\cdot |U|+O\left(\delta \sqrt{|U|}\cdot \frac{|U|\log |V|}{\epsilon\cdot \OPT(G,w)}\right)= O\left(\frac{\delta\epsilon^{-1} |U|^{\frac{3}{2}}\log |V|}{\OPT(G,w)}\right),
\end{align*}
where the second inequality is from Lemma~\ref{lem:distq}, and the third inequality is from Lemma~\ref{lem:stability-of-entropic-regularization} and~\ref{lem:match_u_tvdb}.
\end{proof}

\begin{proof}[Proof of Theorem~\ref{thm:mm-unweighted}]
The approximability result follows by Lemma~\ref{lem:mwbm_approx}.
Let $w, w^*\in \mathbb{R}_{\geq 0}^E$ be two weight vectors with $\OPT(G,w)>0$ and $\|w-w^*\|\leq \frac{\OPT(G,w)}{2}$. Then by Lemma~\ref{lem:mwbm_main}, we have
\begin{align*}
    \EMU(\Call{PLipMWBM}{G,w^*,\epsilon}, \Call{PLipMWBM}{G, w^*+\delta \mathbf{1}_f, \epsilon})
    &\leq O\left(\frac{\delta\epsilon^{-1} |U|^{\frac{3}{2}}\log |V|}{\OPT(G,w^*)}\right)\\
    &\leq O\left(\frac{\delta\epsilon^{-1} |U|^{\frac{3}{2}}\log |V|}{\OPT(G,w)}\right).
\end{align*}
Thus, the claim follows by Lemma~\ref{lem:seeoneelement-unweighted}. 
\end{proof}

\section{Lipschitz Continuity under Shared Randomness}\label{sec:shared-randomness}

In this section, we show that the Lipschitz continuous algorithms we provided in the previous sections can be made Lipschitz continuous under shared randomness. 
In Sections~\ref{subsec:continuous} and~\ref{subsec:discrete}, we provide Lipschitz continuous algorithms for sampling continuous and discrete, respectively, distributions under shared randomness.
In Section~\ref{subsec:rewrite}, we explain how we modify the algorithms given in the previous sections to make them Lipschitz continuous under shared randomness.

\subsection{Sampling from Continuous Uniform Distributions}\label{subsec:continuous}

In our algorithms, we often sample a value from a uniform distribution over some range that slightly perturbs if we perturb the input weight vector. 
In this section, we consider sampling procedures such that, even when the range is slightly perturbed, we obtain the same value with a high probability over internal randomness.

Let $\mathcal{U}\left([0,1]\right)$ and $\mathcal{U}\left([0,1]^{\mathbb{Z}_{\geq 0}}\right)$ be the uniform distributions over $[0,1]$ and $[0,1]^{\mathbb{Z}_{\geq 0}}$, respectively.
Let $l$ and $r$ be functions that take weight vectors as arguments and return real values.
For $c \geq 1$, we say that a sampling process \Call{Sample}{$\cdot$} is $c$-\emph{stable} for a pair of functions $l,r$ if it samples a value uniformly from $[l(w),r(w)]$ when $p$ follows the uniform distribution over $[0,1]^{\mathbb{Z}_{\geq 0}}$ and
\begin{align*}
    &\E_{\bm{p}\sim \mathcal{U}\left([0,1]^{\mathbb{Z}_{\geq 0}}\right)}\left[\TV\left(\Call{Sample}{[l(w),r(w)],\bm{p}},\Call{Sample}{[l(w'),r(w')],\bm{p}}\right)\right]\\
    &\leq c\cdot\TV\left(\mathcal{U}([l(w),r(w)]),\mathcal{U}([l(w'),r(w')])\right)
\end{align*}
holds for all $w,w'\in \mathbb{R}_{\geq 0}^{E}$. 

Let us explain how $c$-stability can be used to bound the Lipschitz constant in the shared randomness setting using the example of the minimum spanning tree problem. In Algorithm~\ref{alg:mst_main_w}, for each edge $e$, a modified weight $\widehat{\bm{w}}(e)$ is sampled from a uniform distribution over $[w(e),(1+\epsilon)w(e)]$. The analysis of the Lipschitz constant proceeds as follows: For weights $w$ and $w+\delta \bm{1}_f$, a coupling $\mathcal{W}$ between $\widehat{\bm{w}}$ and $\widehat{\bm{w}}^{\delta f}$ is defined, and then the Lipschitz constant of the algorithm is proportional to $\Pr_{(\widehat{\bm{w}},\widehat{\bm{w}}^{\delta f}) \sim \mathcal{W}}\left[\widehat{\bm{w}}(f) \neq \widehat{\bm{w}}^{\delta f}(f)\right]$.

The call to \Call{Sample}{$[w(e),(1+\epsilon)w(e)],p$} naturally induces a coupling $\mathcal{W}'$ between $\widehat{\bm{w}}$ and $\widehat{\bm{w}}^{\delta f}$, that is, we use the same $p$. 
The same analysis goes through by considering $\mathcal{W}'$ instead of $\mathcal{W}$. Since \Call{Sample}{$\cdot$} is $c$-stable for the pair of functions $(x\mapsto x, x\mapsto (1+\epsilon)x)$, we have
\begin{align*}
    \Pr_{(\widehat{\bm{w}},\widehat{\bm{w}}^{\delta f}) \sim \mathcal{W}'}\left[\widehat{\bm{w}}(f) \neq \widehat{\bm{w}}^{\delta f}(f)\right]
    &\leq c\cdot \TV\left(\mathcal{U}([l(w),r(w)]),\mathcal{U}([l(w'),r(w')])\right)\\
    &\leq c\cdot \Pr_{(\widehat{\bm{w}},\widehat{\bm{w}}^{\delta f}) \sim \mathcal{W}}\left[\widehat{\bm{w}}(f) \neq \widehat{\bm{w}}^{\delta f}(f)\right].
\end{align*}
Therefore, in the shared randomness setting, we have obtained an algorithm with a Lipschitz constant that is at most $c$ times that of the original algorithm.

\begin{algorithm}[t!]
\caption{Algorithm for sampling $x$ uniformly from fixed $[l,r]$}\label{alg:sample_fixed}
\Procedure{\emph{\Call{Uniform}{$[l,r],p\in [0,1]$}}}{
    \Return $l+p\cdot (r-l)$\;
}
\end{algorithm}

\begin{algorithm}[t!]
\caption{Algorithm for sampling $x$ uniformly from $[l,cl]$}\label{alg:sample_ratio}
\Procedure{\emph{\Call{UniformFixedRatio}{$l,c,p\in [0,1]^{\mathbb{Z}_{\geq 0}}$}}}{
    $b \leftarrow \Call{Uniform}{[1,c],p(0)}$\;
    Let $k$ be the unique integer that satisfies $b\cdot c^{k}\leq l < b\cdot c^{k+1}$\;
    \For{$i=1,2,\ldots$}{
        $x \leftarrow \Call{Uniform}{[1,c^2],p_i}$\;\label{line:ratio_samplex}
        \If{$b\cdot c^{k}\cdot x\in [l,cl]$}{
            \Return $b\cdot c^{k}\cdot x$\;
        }
    }
}
\end{algorithm}

\begin{algorithm}[t!]
\caption{Algorithm for sampling $x$ uniformly from $[l,l+c]$}\label{alg:sample_width}
\Procedure{\emph{\Call{UniformFixedWidth}{$l,c,p\in [0,1]$}}}{
    $x \leftarrow \Call{Uniform}{[0,c],p}$\;
    \If{$x > (l \bmod c)$}{
        \Return $c\cdot \floor{\frac{l}{c}}+x$\;
    }
    \Else{
        \Return $c\cdot \left(\floor{\frac{l}{c}}+1\right)+x$\;
    }
}
\end{algorithm}

We consider three algorithms for \Call{Sample}{$[l(w),r(w)],p$}. Algorithm~\ref{alg:sample_fixed} will be used when $l$ and $r$ are constant functions and do not depend on the weight $w$. Algorithm~\ref{alg:sample_ratio} will be used when there exists a constant $c$ such that $l$ and $r$ satisfy $r(w)=cl(w)$. Algorithm~\ref{alg:sample_width} will be used when there exists a constant $c$ such that $l$ and $r$ satisfy $r(w)=l(w)+c$.

It is clear that Algorithms~\ref{alg:sample_fixed} and~\ref{alg:sample_width} run in constant time. 
The running time of Algorithm~\ref{alg:sample_ratio} can be analyzed as follows. In Line~\ref{line:ratio_samplex}, if $p_i$ is uniformly sampled from $[0,1]$, the $i$-th loop ends with probability at least $1/c$ since $[l,cl]\subseteq [b\cdot c^k, b\cdot c^2\cdot c^k]$ according to the definition of $k$. Therefore, if we choose $p$ randomly, the expected number of loops is at most $1/c$.

Let us analyze the stability of these algorithms.
The following is clear from the definition.
\begin{lemma}
If $l$ and $r$ are constant functions, \emph{\Call{Uniform}{$[l,r],p$}} is $1$-stable for $l$ and $r$.
\end{lemma}

The following lemma will be used to make Algorithms~\ref{alg:mst_main_w},~\ref{alg:main}, and~\ref{alg:splip} Lipschitz continuous and Algorithms~\ref{alg:mst_main_u} and~\ref{alg:match_main} pointwise Lipschitz continuous even in the shared randomness setting.

\begin{lemma}
Let $c > 1$ and suppose $r(w)=cl(w)$ holds for all $w$. 
Then, \emph{\Call{UniformFixedRatio}{$[l(w),r(w)=cl(w)],\bm{p}$}} is $\left(1+c\right)$-stable for $l$ and $r$.
\end{lemma}
\begin{proof}
It is clear that the distribution of \Call{UniformFixedRatio}{$[l(w),r(w)=cl(w)],\bm{p}$} is uniform over $[l(w),cl(w)]$ when $\bm{p}\sim \mathcal{U}\left([0,1]^{\mathbb{Z}_{\geq 0}}\right)$. For weight vectors $w$ and $w'$, let us prove
\begin{align*}
    &\E_{\bm{p}\sim \mathcal{U}\left([0,1]^{\mathbb{Z}_{\geq 0}}\right)}\left[\TV\left(\Call{UniformFixedRatio}{[l(w),cl(w)],\bm{p}},\Call{UniformFixedRatio}{[l(w'),cl(w')],\bm{p}}\right)\right]\\
    &\leq \left(1+c\right)\cdot\TV\left(\mathcal{U}([l(w),cl(w)]),\mathcal{U}([l(w'),cl(w')])\right).
\end{align*}
Without loss of generality, we can assume $l(w)\leq l(w')$.
We can also assume $l(w')< cl(w)$ because otherwise the right hand side is equal to $1+c$ and the lemma trivially holds.

Let $k$ and $k'$ be the values of $k$ chosen in \Call{UniformFixedRatio}{$[l(w),cl(w)],p$} and \Call{UniformFixedRatio}{$[l(w'),cl(w')],p$}, respectively. 

Assume $k=k'$. Then, it clearly holds
\begin{align*}
    &\TV\left(\Call{UniformFixedRatio}{[l(w),cl(w)],p},\Call{UniformFixedRatio}{[l(w'),cl(w')],p}\right)\\
    &= \TV\left(\mathcal{U}([l(w),cl(w)]),\mathcal{U}([l(w'),cl(w')])\right).
\end{align*}
Let us evaluate the probability $\bm{k}\neq \bm{k}'$ happens if $\bm{p}$ is sampled from $\mathcal{U}\left([0,1]^{\mathbb{Z}_{\geq 0}}\right)$. From $l(w)\leq l(w')$, we have $c^{\bm{k}}\leq \frac{l(w)}{\bm{b}}< c^{\bm{k}+1}\leq \frac{l(w')}{\bm{b}} < \frac{cl(w)}{\bm{b}}<c^{\bm{k}+2}$. This can happen only when $\bm{k}=\floor{\log_{c}{l(w)}}=:\widehat{k}$ or $\bm{k}=\widehat{k}-1$ because of $\bm{b}\in [1,c]$. 
Furthermore, $\frac{l(w)}{\bm{b}}< c^{\widehat{k}}\leq \frac{l(w')}{\bm{b}}$ and $\frac{l(w)}{\bm{b}}< c^{\widehat{k}+1}\leq \frac{l(w')}{\bm{b}}$ cannot happen at the same time because $l(w')<cl(w)$.
If $l(w') < c^{\widehat{k}+1}$, we have
\begin{align*}
    \Pr_{\bm{p}\sim \mathcal{U}\left([0,1]^{\mathbb{Z}_{\geq 0}}\right)}[\bm{k}\neq \bm{k}'] 
    &= \Pr\left[\frac{l(w)}{\bm{b}}< c^{\widehat{k}}\leq \frac{l(w')}{\bm{b}}\lor \frac{l(w)}{\bm{b}}< c^{\widehat{k}+1}\leq \frac{l(w')}{\bm{b}}\right]\\
    &= \Pr\left[\bm{b}\in \left[\frac{l(w)}{c^{\widehat{k}}},\frac{l(w')}{c^{\widehat{k}}}\right]\lor \bm{b}\in \left[\frac{l(w)}{c^{\widehat{k}+1}},\frac{l(w')}{c^{\widehat{k}+1}}\right]\right]\\
    &= \Pr\left[\bm{b}\in \left[\frac{l(w)}{c^{\widehat{k}}},\frac{l(w')}{c^{\widehat{k}}}\right]\right]\\
    &\leq \frac{l(w')-l(w)}{c^{\widehat{k}}}\cdot \frac{1}{c-1}\leq \frac{l(w')-l(w)}{l(w)}\cdot \frac{c}{c-1},
\end{align*}
where the third equality is from $\left[\frac{l(w)}{c^{\widehat{k}+1}},\frac{l(w')}{c^{\widehat{k}+1}}\right]\cap [1,c]=\emptyset$.
Otherwise, we have
\begin{align*}
    \Pr_{\bm{p}\sim \mathcal{U}\left([0,1]^{\mathbb{Z}_{\geq 0}}\right)}[\bm{k}\neq \bm{k}'] 
    &= \Pr\left[\frac{l(w)}{\bm{b}}< c^{\widehat{k}}\leq \frac{l(w')}{\bm{b}}\lor \frac{l(w)}{\bm{b}}< c^{\widehat{k}+1}\leq \frac{l(w')}{\bm{b}}\right]\\
    &= \Pr\left[\bm{b}\in \left[\frac{l(w)}{c^{\widehat{k}}},\frac{l(w')}{c^{\widehat{k}}}\right]\lor \bm{b}\in \left[\frac{l(w)}{c^{\widehat{k}+1}},\frac{l(w')}{c^{\widehat{k}+1}}\right]\right]\\
    &= \Pr\left[\bm{b}\in \left[\frac{l(w)}{c^{\widehat{k}}},c\right]\lor \bm{b}\in \left[1,\frac{l(w')}{c^{\widehat{k}+1}}\right]\right]\\
    &\leq \left(\frac{c^{\widehat{k}+1}-l(w)}{c^{\widehat{k}}}+\frac{l(w')-c^{\widehat{k}+1}}{c^{\widehat{k}+1}}\right)\cdot \frac{1}{c-1}\\
    &\leq \left(\frac{c^{\widehat{k}+1}-l(w)}{c^{\widehat{k}}}+\frac{l(w')-c^{\widehat{k}+1}}{c^{\widehat{k}}}\right)\cdot \frac{1}{c-1}
    \leq \frac{l(w')-l(w)}{c^{\widehat{k}}}\cdot \frac{1}{c-1}\leq \frac{l(w')-l(w)}{l(w)}\cdot \frac{c}{c-1},
\end{align*}
where the third equality is from $\frac{l(w')}{c^{\widehat{k}}}\geq \frac{c^{\widehat{k}+1}}{c^{\widehat{k}}}=c$ and $\frac{l(w)}{c^{\widehat{k}+1}} < \frac{l(w)}{c^{\log_c l(w)}}=1$.
Furthermore, we have
\begin{align*}
    &\TV\left(\mathcal{U}([l(w),cl(w)]),\mathcal{U}([l(w'),cl(w')])\right)\\
    &=\frac{cl(w')-cl(w)}{cl(w')-l(w')}
    = \frac{c}{c-1}\cdot \frac{l(w')-l(w)}{l(w')}
    \geq \frac{c}{c-1}\cdot \frac{l(w')-l(w)}{cl(w)}
    \geq \frac{1}{c}\cdot \Pr_{p\sim \mathcal{U}\left([0,1]^{\mathbb{Z}_{\geq 0}}\right)}[\bm{k}\neq \bm{k}'].
\end{align*}
Thus, we have
\begin{align*}
    &\E_{p\sim \mathcal{U}\left([0,1]^{\mathbb{Z}_{\geq 0}}\right)}\left[\TV\left(\Call{UniformFixedRatio}{[l(w),cl(w)],\bm{p}},\Call{UniformFixedRatio}{[l(w'),cl(w')],\bm{p}}\right)\right]\\
    &\leq \TV\left(\mathcal{U}([l(w),cl(w)]),\mathcal{U}([l(w'),cl(w')])\right) + 1\cdot \Pr_{\bm{p}\sim \mathcal{U}\left([0,1]^{\mathbb{Z}_{\geq 0}}\right)}[\bm{k}\neq \bm{k}']\\
    &\leq \left(1+c\right)\cdot\TV\left(\mathcal{U}([l(w),cl(w)]),\mathcal{U}([l(w'),cl(w')])\right).\qedhere
\end{align*}
\end{proof}

The following will be used to make Algorithm~\ref{alg:match_main} pointwise Lipschitz continuous even in the shard randomness setting.
\begin{lemma}
Let $c>0$ and suppose $r(w)=l(w)+c$ holds for all $w$.
Then, \emph{\Call{UniformFixedWidth}{$[l(w),r(w)=l(w)+c],p$}} is $1$-uniform for $l$ and $r$.
\end{lemma}
\begin{proof}
It is clear that the distribution of \Call{UniformFixedWidth}{$[l(w),r(w)=l(w)+c],\bm{p}$} is uniform distribution over $[l(w),l(w)+c]$ when $\bm{p}\sim \mathcal{U}\left([0,1]^{\mathbb{Z}_{\geq 0}}\right)$. 
Let $w$ and $w'$ be weight vectors. Without loss of generality, we can assume $l(w)\leq l(w')$. We can also assume $l(w')\leq l(w)+c$ because otherwise $\TV\left(\mathcal{U}([l(w),l(w)+c]),\mathcal{U}([l(w'),l(w')+c])\right)=1$ and the lemma trivially holds. Thus, we can assume $\floor{\frac{l(w')}{c}}$ is equal to either $\floor{\frac{l(w)}{c}}$ or $\floor{\frac{l(w)}{c}}+1$.

Assume $\floor{\frac{l(w')}{c}}=\floor{\frac{l(w)}{c}}$. Then, \Call{UniformFixedWidth}{$[l(w),l(w)+c],p$} and \Call{UniformFixedWidth}{$[l(w'),l(w')+c],p$} will be distinct only when $(l(w)\bmod c)<x\leq (l(w')\bmod c)$. Thus we have
\begin{align*}
    &\E_{\bm{p}\sim \mathcal{U}\left([0,1]^{\mathbb{Z}_{\geq 0}}\right)}\left[\TV\left(\Call{UniformFixedWidth}{[l(w),l(w)+c],\bm{p}},\Call{UniformFixedWidth}{[l(w'),l(w')+c],\bm{p}}\right)\right]\\
    &=\frac{l(w')-l(w)}{c}=\TV\left(\mathcal{U}([l(w),l(w)+c]),\mathcal{U}([l(w'),l(w')+c])\right).
\end{align*}

Assume $\floor{\frac{l(w')}{c}}=\floor{\frac{l(w)}{c}}+1$. Then, \Call{UniformFixedWidth}{$[l(w),l(w)+c],p$} and \Call{UniformFixedWidth}{$[l(w'),l(w')+c],p$} will be distinct only when $x > (l(w)\bmod c)$ or $x\leq (l(w')\bmod c)$. Thus we have
\begin{align*}
    &\E_{\bm{p}\sim \mathcal{U}\left([0,1]^{\mathbb{Z}_{\geq 0}}\right)}\left[\TV\left(\Call{UniformFixedWidth}{[l(w),l(w)+c],\bm{p}},\Call{UniformFixedWidth}{[l(w'),l(w')+c],\bm{p}}\right)\right]\\
    &=\frac{(l(w')\bmod c)+c-(l(w)\bmod c)}{c}=\frac{l(w')-l(w)}{c}=\TV\left(\mathcal{U}([l(w),l(w)+c]),\mathcal{U}([l(w'),l(w')+c])\right),
\end{align*}
where the second equality is from $\floor{\frac{l(w')}{c}}=\floor{\frac{l(w)}{c}}+1$.
Therefore the lemma is proved.
\end{proof}



\subsection{Sampling from Discrete Distributions}\label{subsec:discrete}

In our algorithms, we often sample a value from a discrete universe. 
Let $\mathcal{X}$ be a universe and $x\in [0,1]^{\mathcal{X}}$ be a probability vector, i.e., $\sum_{v\in \mathcal{X}}x_v=1$.
Then, we say that a sampling process \Call{Sample}{$\mathcal{X},x,p$} is \emph{stable} if it samples a value from $v\in \mathcal{X}$ with probability $x_v$ when $\bm{p}\sim \mathcal{U}\left([0,1]^{\mathbb{Z}}\right)$ and for two probability vectors $x$ and $x'$, it holds that
\begin{align*}
    \Pr_{\bm{p}\sim \mathcal{U}\left(\mathcal{X}^{\mathbb{Z}_{\geq 0}}\right)}\left[\Call{Sample}{\mathcal{X},x,\bm{p}}\neq \Call{Sample}{\mathcal{X},x',\bm{p}}\right]
    =\TV\left(x,x'\right).
\end{align*}
Here, we abuse notation $\TV\left(x,x'\right)$ to represent the total variation distance between the probability distribution for sampling an element $v$ from $\mathcal{X}$ with probability $x_v$ and the probability distribution for sampling an element $v$ from $\mathcal{X}$ with probability $x'_v$.

\begin{algorithm}[t!]
\caption{Algorithm to sample $v$ from a universe $\mathcal{X}$ with probability $x_v$}\label{alg:sample_discrete}
\Procedure{\emph{\Call{UniformDiscrete}{$\mathcal{X}, x\in [0,1]^{\mathcal{X}}, p\in [0,1]^{\mathbb{Z}_{\geq 0}}$}}}{
    \For{$i=0,1,\dots$}{
        $t\leftarrow \floor{p_{2i}\cdot |\mathcal{X}|}+1$\;
        Let $v$ be the $t$-th element of $\mathcal{X}$\;
        \If{$p(2i+1)\leq x_v$}{
            \Return $v$\;
        }
    }
}
\end{algorithm}
We denote the uniform distribution over $\mathcal{X}$ by $\mathcal{U}[\mathcal{X}]$.
Algorithm~\ref{alg:sample_discrete} is the algorithm we consider, and it is stable as shown in the following.
\begin{lemma}
\Call{UniformDiscrete}{$U, \mathcal{X}, p\in \mathcal{X}^{\mathbb{Z}_{\geq 0}}$} is stable.
\end{lemma}
\begin{proof}
It is clear that \Call{UniformDiscrete}{$\mathcal{X}, x\in [0,1]^{\mathcal{X}}, \bm{p}\in [0,1]^{\mathbb{Z}_{\geq 0}}$} outputs an element $v\in \mathcal{X}$ with probability $x_v$ when $\bm{p}  \sim \mathcal{U}\left([0,1]^{\mathbb{Z}}\right)$.
Let $x,x'\in [0,1]^{\mathcal{X}}$ be probability vectors such that $\sum_{v\in \mathcal{X}}x_v=\sum_{v\in \mathcal{X}}x'_v=1$. Let $i^*$ denote the value of $i$ when \Call{UniformDiscrete}{$\mathcal{X}, x\in [0,1]^{\mathcal{X}}, p\in [0,1]^{\mathbb{Z}_{\geq 0}}$} terminates. Then,
\begin{align*}
    &\Pr_{\bm{p}\sim \mathcal{U}\left(\mathcal{X}^{\mathbb{Z}_{\geq 0}}\right)}\left[\Call{UniformDiscrete}{\mathcal{X},x,\bm{p}}\neq \Call{UniformDiscrete}{\mathcal{X},x',\bm{p}}\mid i^*=k\right]\\
    &=\sum_{v\in \mathcal{X}}\left|x_v-x'_v\right|=\TV\left(x,x'\right).
\end{align*}
Thus, we have
\begin{align*}
    \Pr_{\bm{p}\sim \mathcal{U}\left(\mathcal{X}^{\mathbb{Z}_{\geq 0}}\right)}\left[\Call{UniformDiscrete}{\mathcal{X},x,\bm{p}}\neq \Call{UniformDiscrete}{\mathcal{X},x',\bm{p}}\right]
    =\TV\left(x,x'\right).
    \qedhere
\end{align*}
\end{proof}

When $\bm{p}\sim [0,1]^{\mathbb{Z}{\geq 0}}$, it can be verified that Algorithm~\ref{alg:sample_discrete} runs in expected time $O(|\mathcal{X}|)$, as each step of the loop terminates with probability $\sum_{v\in \mathcal{X}}\left(\frac{1}{|\mathcal{X}|}x_v\right)=\frac{1}{|\mathcal{X}|}$.

\subsection{Modifications to Ensure Lipschitz Continuity under Shared Randomness}\label{subsec:rewrite}

\subsubsection{Algorithm~\ref{alg:mst_main_w}}

By modifying Algorithm~\ref{alg:mst_main_w} to use \Call{SampleFixedRatio}{$[w(e),(1+\epsilon)w(e)],p$} to sample $\widehat{\bm{w}}(e)$ at Line~\ref{line:mst_main_w_sample}, it can be converted into a Lipschitz continuous algorithm in the shared randomness setting. Although the right hand side of the statement in Lemma~\ref{lem:mst-remaining} becomes $2+\epsilon$ times as large, the proof of Lemma~\ref{lem:mst-weighted-Lipschitz} still works in the same way. Therefore, an algorithm achieving an approximation ratio of $1+\epsilon$ and a Lipschitz constant of $O(\epsilon^{-1})$ is obtained in the shared randomness setting.

\subsubsection{Algorithm~\ref{alg:splip}}

To make Algorithm~\ref{alg:splip} Lipschitz continuous in the shared randomness setting, first, we use \Call{UniformFixedRatio}{} to sample $\bm{b}$ at Line~\ref{line:lipsp_sampleb}. Although the right hand side of the statement in Lemma~\ref{lem:splip_sampleb} becomes three times as large, the proof of Lemma~\ref{lem:splipsens} still works in the same way. 
In Line~\ref{line:lipsp_samplex} of Algorithm~\ref{alg:splip}, $\bm{x}$ is sampled, but it is only used for sampling $\widehat{w}$. Therefore, it is permissible to replace the sampling of $\bm{x}(e)$ with the operation of sampling an integer value $\widehat{\bm{w}}(e)$. This can be rewritten using \Call{UniformDiscrete}{} as follows. 
If $\widehat{\bm{w}}(e)$ is greater than $12\epsilon^{-1}|V|+3$, then no path corresponding to $e$ is created in $\widehat{\bm{G}}$, so if $\bm{l}(e) > 12\epsilon^{-1}|V|+1$, do nothing. Otherwise, let $\mathcal{X}=\left\{2,\dots, \lfloor{12\epsilon^{-1}|V|+4}\rfloor\right\}$. Let a probability vector $\bm{y}$ be
\begin{align*}
    \bm{y}(i)=
    \begin{cases}
        \frac{(\bm{l}(e)+1)\bm{b}-\bm{w}(e)}{\bm{b}} & (i=\bm{l}(e)+2)\\
        1-\frac{(\bm{l}(e)+1)\bm{b}-\bm{w}(e)}{\bm{b}} & (i=\bm{l}(e)+3)\\
        0 & (\text{otherwise})        
    \end{cases}
\end{align*}
Then, $\widehat{\bm{w}}(e)$ is sampled by \Call{UniformDiscrete}{$\mathcal{X},\bm{y},p$}.

We also need to make \Call{DiSP}{$\widehat{\bm{G}}$} compatible with the shared randomness setting. In Algorithm~\ref{alg:main}, we use \Call{UniformFixedRatio}{} to sample $\bm{\gamma}$ at Line~\ref{line:main_samplegamma} and to sample $\bm{d}$ at Line~\ref{line:choosed}, and use \Call{Uniform}{} to sample $\bm{l}$ at Line~\ref{line:choosel}. This makes Lemma~\ref{lem:contsp_samplegamma} and Lemma~\ref{lem:TVD} have their right-hand sides tripled, but it does not affect the analysis.

To sample $\bm{v}$ in Line~\ref{line:main-sample}, we use \Call{UniformDiscrete}{}. We would like to take the universe $\mathcal{X}$ as the set of vertices $\widehat{V}$ of the graph $\widehat{G}$. However, because \Call{LipSP}{$G,w,s,t,\epsilon$} in Algorithm~\ref{alg:splip} may call \Call{DiSP}{$\widehat{G},\cdot$} for different graphs $\widehat{G}$ depending on the choice of weights $w$ and internal randomness, we cannot set $\mathcal{X}=\widehat{V}$.
For this reason, we cannot obtain a low-sensitivity algorithm for the shortest path problem in the shared randomness setting just by modifying Algorithm~\ref{alg:main} alone. Instead, let us recall that each edge of $G$ corresponds to two directed paths of length at most $12\epsilon^{-1}|V|+3$. 
Thus, any vertex $v$ in $\widehat{G}$ that does not correspond to a vertex in $G$ can be represented as a triple $(e, t, i)$, where $e$ is an edge in $G$ that corresponds to $v$, $t\in{0,1}$ indicates whether $v$ belongs to the first or second path in $\widehat{G}$ corresponding to $e$, and $i\in \{1,\dots, \floor{12\epsilon^{-1}|V|+3}\}$ is the distance from the start of the path. Therefore, we can take $\mathcal{X}=V\cup \left(E\times \{0,1\}\times \{1,\dots, \floor{12\epsilon^{-1}|V|+3\}}\right)$, and use \Call{UniformDiscrete}{} to sample $\bm{v}$.

\subsubsection{Algorithm~\ref{alg:match_scale}}

To sample $\bm{b}$ in Line~\ref{line:MWM_sampleb} of Algorithm~\ref{alg:match_scale}, we use \Call{Uniform}{$[1,\alpha],p$}. In Line~\ref{line:MWM_samplepi}, we need to uniformly sample a permutation $\pi$ of $\{1,\dots, |V|\}$. It can be achieved by repeatedly applying \Call{SampleDiscrete}{$V, \mathcal{U}\left[V\setminus \{\pi(1),\dots, \pi(k)\}\right], p$}, where we abuse the notation $\mathcal{U}\left[X\right]$ to represent the distribution that samples elements in $X$ uniformly at random, and elements in $V\setminus X$ with probability $0$. The analysis is not affected by these modifications, so the maximum weight matching problem can be approximated with a $(1/8-\epsilon)$ approximation and a Lipschitz constant of $O(\epsilon^{-1})$ in the shared randomness setting.

\subsubsection{Algorithm~\ref{alg:mst_main_u}}

To sample $\bm{b}$ in Line~\ref{line:mstu_sampleb} of Algorithm~\ref{alg:mst_main_u}, we use \Call{UniformFixedRatio}{$\cdot$}. This triples the right-hand side of Lemma~\ref{lem:unweighted-mst-b}, but it does not affect the analysis. Moreover, Line~\ref{line:mstu_samplew} samples $\widehat{\bm{w}}(e)$ using \Call{UniformFixedWidth}{$[w(e),w(e)+\bm{b}],p$}. This also does not affect the analysis, so in the shared randomness setting, a $(1-\epsilon)$ approximation algorithm with pointwise Lipschitzness $O(\epsilon^{-1}n/\OPT)$ can be obtained for the minimum weight spanning tree problem.

\subsubsection{Algorithm~\ref{alg:match_main}}

To sample $\bm{B}$ in Line~\ref{line:MWBM_sampleb} of Algorithm~\ref{alg:match_main}, \Call{UniformFixedRatio}{$\cdot$}. This triples the right-hand side of Lemma~\ref{lem:match_u_tvdb} but does not affect the analysis. We use \Call{SampleDiscrete}{$V\cup \{\bot \}, (\bm{x}_{i,1},\cdots, \bm{x}_{i,|V|}, 1-\sum_{j=1}^{|V|}\bm{x}_{i,j}),p$} to sample $\bm{p}(i)$ in Line~\ref{line:MWBM_samplep}, and \Call{SampleDiscrete}{$U, \mathcal{U}[\bm{C}_j],p$} to sample $\bm{q}(j)$ in Line~\ref{line:MWBM_sampleq}. These do not affect the analysis of Lemma~\ref{lem:mwbm_main}, and thus we obtain a $(1/2-\epsilon)$ approximation algorithm with pointwise Lipschitzness $O\left(\frac{\epsilon^{-1} |U|^{\frac{3}{2}}\log |V|}{\OPT(G,w)}\right)$ for the maximum weight bipartite matching problem even in the shared randomness setting.

\bibliography{main}

\appendix


\section{Missing Proofs}\label{sec:seeoneelement}

\begin{proof}[Proof of Lemma~\ref{lem:seeoneelement}]
First, note that the inequality in the statement holds even for $\delta > c$ because then
\begin{align*}
    &\EMW\left((\mathcal{A}(G,w),w), (\mathcal{A}(G,w+\delta \mathbf{1}_f),w+\delta \mathbf{1}_f)\right)\\
    &\leq \EMW\left((\mathcal{A}(G,w),w), (\mathcal{A}(G,w+(\delta-c) \mathbf{1}_f),w+(\delta-c) \mathbf{1}_f)\right)\\
    &\quad +\EMW\left((\mathcal{A}(G,w+(\delta-c)\mathbf{1}_f),w+(\delta-c)\mathbf{1}_f), (\mathcal{A}(G,w+\delta \mathbf{1}_f),w+\delta \mathbf{1}_f)\right)\\
    &\leq (\delta - c)L + cL = \delta L,
\end{align*}
where the last inequality can be proven by induction.
Also because $\EMW$ is commutative, the inequality 
\[
    \EMW\left((\mathcal{A}(G,w),w), (\mathcal{A}(G,w+\delta \mathbf{1}_f),w+\delta \mathbf{1}_f)\right)\leq |\delta| L
\]
holds even for $\delta < 0$.

Let $e_1,\ldots,e_m$ be an arbitrary ordering of elements in $E$, where $m = |E|$.
For two weight vectors $w,w'\in \mathbb{R}_{\geq 0}^E$, by the triangle inequality, we have
\begin{align*}
    \EMW\left((\mathcal{A}(G,w),w),(\mathcal{A}(G,w'),w')\right)
    \leq\sum_{i=1}^{m}\EMW\left((\mathcal{A}(G,w_{i-1}),w_{i-1}),(\mathcal{A}(G,w_i),w_i)\right),
\end{align*}
where $w_i \in \mathbb{R}_{\geq 0}^E$ is the weight vector given by
\begin{align*}
    w_i(e_j)=\begin{cases}
        w'(e_j) & (j\leq i) \\
        w(e_j) & (j > i).
        \end{cases}
\end{align*}
Since $w_i=w_{i-1}+(w'(e_i)-w(e_i))\mathbf{1}_{e_i}$, we have
\begin{align*}
    &\sup_{\substack{w,w' \in \mathbb{R}_{\geq 0}^E,\\w\neq w'}}\frac{\EMW\left((\mathcal{A}(G,w),w),(\mathcal{A}(G,w'),w')\right)}{\|w-w'\|_1}\\
    & \leq \sup_{\substack{w,w' \in \mathbb{R}_{\geq 0}^E,\\w\neq w'}}\frac{\sum_{i=1}^{m}\EMW\left((\mathcal{A}(G,w_{i-1}),w_{i-1}),(\mathcal{A}(G,w_{i}),w_{i})\right)}{\sum_{i=1}^{m}|w'(e_i)-w(e_i)|}\\
    &\leq \sup_{\substack{w,w' \in \mathbb{R}_{\geq 0}^E,\\w\neq w'}}\max_{i=1,\ldots,m}\frac{\EMW\left((\mathcal{A}(G,w_{i-1}),w_{i-1}),(\mathcal{A}(G,w_{i}),w_{i})\right)}{|w'(e_i)-w(e_i)|}\\
    &\leq L.
    \qedhere
\end{align*}
\end{proof}

\begin{proof}[Proof of Lemma~\ref{lem:seeoneelement-unweighted}]
Let $e_1,\ldots,e_m$ be an arbitrary ordering of elements in $E$, where $m = |E|$.
For two weight vectors $w,w'\in \mathbb{R}_{\geq 0}^E$ with $\|w-w'\|_1\leq c$, by the triangle inequality, we have
\begin{align*}
    \EMU\left(\mathcal{A}(G,w),\mathcal{A}(G,w')\right)
    \leq\sum_{i=1}^{m}\EMU\left(\mathcal{A}(G,w_{i-1}),\mathcal{A}(G,w_i)\right),
\end{align*}
where $w_i \in \mathbb{R}_{\geq 0}^E$ is defined as in the proof of Lemma~\ref{lem:seeoneelement}.
Since $w_i=w_{i-1}+(w'(e_i)-w(e_i))\mathbf{1}_{e_i}$, we have
\begin{align*}
    \sup_{\substack{w,w' \in \mathbb{R}_{\geq 0}^E,\\w\neq w'}}\frac{\EMU\left(\mathcal{A}(G,w),\mathcal{A}(G,w')\right)}{\|w-w'\|_1}
    & \leq \sup_{\substack{w,w' \in \mathbb{R}_{\geq 0}^E,\\w\neq w'}}\frac{\sum_{i=1}^{m}\EMU\left(\mathcal{A}(G,w_{i-1}),\mathcal{A}(G,w_{i})\right)}{\sum_{i=1}^{m}|w'(e_i)-w(e_i)|}\\
    &\leq \sup_{\substack{w,w' \in \mathbb{R}_{\geq 0}^E,\\w\neq w'}}\max_{i=1,\ldots,m}\frac{\EMU\left(\mathcal{A}(G,w_{i-1}),\mathcal{A}(G,w_{i})\right)}{|w'(e_i)-w(e_i)|}\\
    &\leq L,
\end{align*}
where the last inequality is from the fact that $\|w-w_{i-1}\|_1\leq \|w-w'\|_1\leq c$ for $w(e_i)\leq w'(e_i)$ and $\|w-w_{i}\|_1\leq \|w-w'\|_1\leq c$ for $w(e_i) > w'(e_i)$. 
Therefore, $\mathcal{A}$ has pointwise Lipschitz constant $L$ on $G$ at $w$,  considering the neighborhood $\{w'\colon \|w-w'\|\leq c\}$.
\end{proof}

\end{document}